\documentclass[pdflatex, 11pt]{article}
\usepackage[top=15truemm,bottom=15truemm,left=25truemm,right=25truemm]{geometry}
\usepackage{amsmath, amssymb, amsfonts, amsthm}
\usepackage{ascmac} 
\usepackage[pdftex]{graphicx}
\usepackage{enumerate} 
\usepackage{color} 
\usepackage{authblk} 
\usepackage[T1]{fontenc}
\usepackage[utf8]{inputenc}
\usepackage{bm} 
\usepackage{braket} 
\usepackage{dsfont} 
\usepackage{mathrsfs} 
\usepackage{mathtools} 
\usepackage{booktabs} 
\usepackage{adjustbox} 
\usepackage{tikz} 
\usepackage{tikz-cd} 
\usetikzlibrary{calc,trees,positioning,arrows,fit,shapes,calc,arrows.meta,decorations.pathreplacing,bending,matrix,quotes,math}
\usepackage{caption}
\usepackage{subcaption}
\theoremstyle{definition}
\newtheorem{definition}{Definition}[section]
\newtheorem{theorem}{Theorem}[section]
\newtheorem{lemma}{Lemma}[section]
\newtheorem{corollary}{Corollary}[section]
\newtheorem{remark}{Remark}[section]
\newtheorem{proposition}{Proposition}[section] 
\newtheorem{assumption}{Assumption}[section]

\newcommand{\D}{\mathrm{d}} 

\title{{\Large\bf Trade execution games in a Markovian environment}}

\author{Masamitsu OHNISHI\S\dag}
\author{Makoto SHIMOSHIMIZU\P}
\affil{\S Faculty of Informatics, Yamato University\thanks{Address: Suita City, Osaka, 564--0082, Japan \\
E-mail address: onishi.masamitsu@yamato-u.ac.jp}}
\affil{\hspace{-10mm} \dag Center for Mathematical Modeling and Data Science, Osaka University\thanks{Address: Toyonaka City, Osaka, 560--8531, Japan \\
E-mail address: ohnishi@econ.osaka-u.ac.jp}}
\affil{\P Department of Industrial and Systems Engineering, Tokyo University of Science\thanks{Address: 2641 Yamazaki, Noda City, Chiba, 278--8510, Japan \\
E-mail address: shimosshi1q94@gmail.com}}

\date{\today}

\begin{document}

\maketitle

\begin{abstract}
This paper examines a trade execution game for two large traders in a generalized price impact model. We incorporate a stochastic and sequentially dependent factor that exogenously affects the market price into financial markets. Our model accounts for how strategic and environmental uncertainties affect the large traders' execution strategies. We formulate an expected utility maximization problem for two large traders as a Markov game model. Applying the backward induction method of dynamic programming, we provide an explicit closed-form execution strategy at a Markov perfect equilibrium. Our theoretical results reveal that the execution strategy generally lies in a dynamic and non-randomized class; it becomes deterministic if the Markovian environment is also deterministic. In addition, our simulation-based numerical experiments suggest that the execution strategy captures various features observed in financial markets.\\


\end{abstract}

\section{Introduction} \label{Section_1}

The last two decades have witnessed considerable interest in trade execution problems for a (single) large trader among academic researchers and practitioners. A \textit{large trader} refers to an institutional trader who trades (or \textit{execute}s) large amounts of orders in a short time window and moves the asset prices in an unpreferable direction. Typical examples of large traders are insurance companies or pension funds. The cost incurred by large traders is referred to as \textit{price impact} (or \textit{market impact}). Large traders then divide their orders into small pieces and gradually submit relatively small orders to mitigate the price impact. On the contrary, submitting small orders leads to the exposition (or risk) of price fluctuation. Therefore, large traders have to manage the following two facets when considering \textit{execution strategy}: price impact and future fluctuation of a financial asset price. We call the issue concerned with (trade) execution as \textit{execution problem}. We refer to Gu\'{e}ant~\cite{GO16}, Cartea et~al.~\cite{CAJSPJ15}, and Shimoshimizu~\cite{SM24} for more details on execution problems.

In addition, exogenous factors with \textit{mean-reverting property}, such as \textit{order-flow of small traders} and \textit{order book imbalance} (OBI), influence an execution strategy through market price. The order flow of small traders refers to (random) aggregate orders posed by \textit{small traders}. Empirical and theoretical studies demonstrate that the order flow of small traders incurs a price impact (Potters and Bouchaud~\cite{PMBJP03}, Cartea and Jaimungal~\cite{CAJS16MFE}, Cartea and Jaimungal~\cite{CAJS16SIAM}) and influences an (optimal) execution strategy (Cartea and Jaimungal~\cite{CAJS16MFE}, Ohnishi and Shimoshimizu~\cite{OMSM20QF}, Fukasawa et~al.~\cite{FMOMSM21IJTAF}, Ohnishi and Shimoshimizu~\cite{OMSM22APFM}). The OBI refers to the difference between the volume at the best buy price and the one at the best sell price divided by its sum, and has a mean-reversion property (Lehalle and Neuman~\cite{LCANE19}). Existing studies show that OBI causes price fluctuation (e.g., Cont et~al.~\cite{CRKASS14}, Stoikov~\cite{SS18}) and thus influences an (optimal) execution strategy (Lehalle and Neuman~\cite{LCANE19}). Therefore, incorporating exogenous factors with mean-reverting property is essential when one analyzes execution problems.

Although a growing body of literature has investigated execution problems for a single large trader, execution problems for multiple large traders are worth examining from practitioners' standpoints. Multiple large traders interact with each other, and the interaction affects the market price in a real market. Consider, for example, the following situation: an investor orders two or more institutional brokerages to buy (sell, respectively) a large amount of one financial asset in a short period. The institutional brokerages must compete to buy (sell) the financial asset at as low (high) price as possible. The competition among multiple institutional brokerages, in turn, may cause country-wide or worldwide destabilization. Typical examples of the destabilization are the ``Flash Clash,'' caused by the rapid execution of the E-Mini S\&P 500 on May 6, 2010, and a ``hot potato game,'' a quick buy and sell by high-frequency traders (HFTs), as mentioned in Kirilenko et~al.~\cite{KAKASSMTT17} and Schied and Zhang~\cite{SAZT19}. These examples highlight the importance of examining in-depth execution problems for multiple large traders.

With the above background in mind, we analyze an execution problem for two risk-averse large traders (as considered in Schied and Zhang~\cite{SAZT19}, Luo and Schied~\cite{LXSA20}, Ohnishi and Shimoshimizu~\cite{OMSM20QF}). In particular, our model sheds light on the effect of what we define as the \textit{Markovian environment} on an execution strategy at an \textit{equilibrium} for two large traders. We define the Markovian environment as an exogenous factor with mean-reverting property and is described by an AR (1)-type random sequences. The main purpose of this paper is to take a new look at how \textit{strategic uncertainty} (characterized by the existence of another large trader) and \textit{environmental uncertainty} (described by the Markovian environment) affect the execution strategy of large traders. Theorem~\ref{Theorem_3.1}, our main theorem, shows that there exists a \textit{Markov perfect equilibrium} at which the execution strategy is an \textit{affine function} of the counterpart's remaining execution volume, the residual effect of past price impact, and the Markovian environment. This result suggests that strategic and environmental uncertainties play an indispensable role in determining execution strategy.

In addition, simulation-based numerical experiments reveal intriguing results. For example, when two large traders exist in a financial market, large traders execute faster than when only one large trader exists. The execution speed, however, seems indifferent when the environmental uncertainty (driven by the Markovian environment) is large. Moreover, the execution strategy at the equilibrium becomes \textit{asymmetric} among large traders with \textit{opposite initial inventories} by the condition of the Markovian environment. This execution strategy remarkably captures the situation that the total traded volume in a ﬁnancial market seemingly forms a U-shaped trading curve.

This paper contributes to the literature examining execution problems in the following ways. Our model incorporates two (i.e., multiple) large traders that allow us to investigate strategic uncertainty. We also consider environmental uncertainty described by exogenous factors with mean-reverting properties. Our theoretical results show that accounting for these two types of uncertainty is indispensable for rational execution strategies. Simulation-based results also validate the importance of simultaneously incorporating strategic and environmental uncertainties.

The organization of this paper is as follows. Section~\ref{Section_2} summarizes related literature. In Section~\ref{Section_3}, we describe a market model and obtain an explicit execution strategy. Section~\ref{Section_4} is devoted to simulation-based numerical experiments. Finally, Section~\ref{Section_5} concludes. We show the proof for the main theorem (Theorem~\ref{Theorem_3.1}) in the Appendix.

\section{Related literature} \label{Section_2}

This section gives a general panorama of related literature to clarify the contribution of this paper. In Section~\ref{Subsection_2.1}, we review how studies on price impacts have developed and some representative types of exogenous factors that affect a large trader's execution strategy. Section~\ref{Subsection_2.2} then overviews execution problems for multiple large traders. In Section~\ref{Subsection_2.3}, a concise review of market microstructural notions follows to illustrate the role of the Markovian environment.

\subsection{Execution problem with transient impact and exogenous factors} \label{Subsection_2.1}

There is considerable literature on execution problems for a single large trader. Bertsimas and Lo~\cite{BSLAW98} is the first paper that investigates an optimal execution strategy in a discrete-time framework and shows that the optimal strategy becomes a basket of equally-divided trading volumes. Subsequently, Almgren and Chriss~\cite{ARCN01} constructs a model that integrates a large trader's execution cost and the large trader's risk-averse property. Both studies incorporate only \textit{temporary} and \textit{permanent} price impacts into their models. The temporary price impact refers to the short-lived effect on the price of an asset resulting from a specific transaction or event, such as a large buy or sell order. It typically disappears once the market adjusts to the new information or activity. The permanent price impact, on the other hand, denotes a lasting change in the price of an asset caused by, for example, fundamental shifts in supply or demand dynamics or changes in market structure.

Succeeding studies, however, demonstrate that a price impact has a \textit{transient property}: a price impact dissipates over the trading window (e.g., Bouchaud et~al.~\cite{BJPGYPMWM04}, Gatheral~\cite{GJ10}). We call such a price impact a \textit{transient price impact}. Obizhaeva and Wang~\cite{OAAWJ13} analyzes an execution problem with transient price impact and derive an optimal execution strategy for a single large trader. Following this seminal paper, there has been an increasing number of studies that focus on execution problems with transient price impact (e.g., Kuno and Ohnishi~\cite{KSOM15}, Kuno et~al.~\cite{KSOMSP17}, Schied and Zhang~\cite{SAZT19}, Luo and Schied~\cite{LXSA20}, Ohnishi and Shimoshimizu~\cite{OMSM20QF, OMSM22APFM}, Fukasawa et~al.~\cite{FMOMSM21IJTAF, FMOMSM21RIMS}, Cordoni and Lillo~\cite{CFLF22, CFLF23}). Our model also incorporates a transient price impact as a key factor in determining large traders' execution strategies.

Existing studies show that exogenous random factors affect a large trader's execution strategy. For example, a random order ﬂow of small traders influences an optimal execution strategy (e.g., Bechler and Ludkovski~\cite{BKLM15}, Cartea and Jaimungal~\cite{CAJS16MFE, CAJS16SIAM}, Ohnishi and Shimoshimizu~\cite{OMSM20QF}, Fukasawa et~al.~\cite{FMOMSM21IJTAF, FMOMSM21RIMS}, Ohnishi and Shimoshimizu~\cite{OMSM22APFM}). We can summarize other examples of exogenous factors as follows: trader's view of the market (Cartea et~al.~\cite{CAJSKD16}), effect of other assets that large traders do not trade (Cartea~et~al.~\cite{CAGLJS19}, Ohnishi and Shimoshimizu~\cite{OMSM22APFM}), difference between the weighted returns of two assets in the same industry or with similar characteristics (Avellaneda and Lee~\cite{AMLJH10}).

\subsection{Execution game problem} \label{Subsection_2.2}

Since multiple large traders interact with each other and affect a financial asset price in the real world, game-theoretical analysis helps us investigate an execution problem thoroughly. Schied and Zhang~\cite{SAZT19} and Luo and Schied~\cite{LXSA20}, motivated by Sch{\"{o}}neborn~\cite{ST08}, investigate a \textit{market impact game} (or what we call the \textit{trade execution game}) with transient price impact for one financial asset. These studies explicitly derive an execution strategy at a \textit{Nash equilibrium} for a cost minimization problem as well as a utility maximization problem. They also examine under what conditions the execution strategy \textit{does not oscillate} (i.e., not repeatedly buy/sell a financial asset). These results offer compelling evidence for how financial markets can be stable. Cordoni and Lillo~\cite{CFLF22} subsequently extend their analysis and examine multiple financial assets. Furthermore, Cordoni and Lillo~\cite{CFLF23} derive a transient price impact from an obtained execution strategy in a similar market model of Schied and Zhang~\cite{SAZT19}.

Other existing studies show that an execution strategy at an equilibrium can be adaptive for random effects on the future price. Ohnishi and Shimoshimizu~\cite{OMSM19RIMS, OMSM20QF}, for example, address an execution problem for multiple large traders under a transient price impact model and derive an execution strategy at a \textit{Markov perfect equilibrium} in a dynamic and non-randomized class. Other studies, for example, Huang et~al.~\cite{HXJSNM19} and Casgrain and Jaimungal~\cite{CPJS20}, analyze an execution problem with temporary and permanent price impacts via a mean-field game approach. Mean-field game theory enables us to effectively analyze strategic decisions regarding trade execution by \textit{small} interacting traders in a very large population of traders.

We consider a more general execution problem for two large traders under the following assumption: (i) a transient price impact exists, and (ii) a sequentially (or auto-)correlated exogenous factor affects a financial asset price. The results show that the execution strategy at a Markov perfect equilibrium is dynamic and non-randomized. Incorporating the features of a transient price impact and an exogenous factor perceived by large traders can also capture a more realistic situation, as illustrated by our simulation-based numerical experiments.

\subsection{Microstructural effect on market price} \label{Subsection_2.3}

Some studies illuminate the effect of microstructural factors on the market price of a financial asset. In particular, how price dynamics of a financial asset are determined through the order book has attracted widespread interest among academic researchers and practitioners. We may recognize the mid-price as the price of a financial asset without any price impact caused by (large) order submissions. The mid-price is defined as the mean of best-bid and best-ask: $M \coloneqq (P^{b} + P^{a})/2$, where $P^{b}$ and $P^{a}$ are the best-bid and best-ask, respectively. Another feature that may be of interest to practitioners is the weighted mid-price defined as $W \coloneqq w P^{a} + (1 - w) P^{b}$. The weight $w$ is the \textit{order book imbalance} defined by the total volume at the best bid $Q^{b}$ and the total volume at the best ask $Q^{a}$ as $w \coloneqq Q^{b}/(Q^{b} + Q^{a})$.

Both types of ``mid-price'' make sense in terms of being easily obtained from market data, although existing studies have shown their shortcomings. For example, the mid-price does not depend on the volume at the best bid and best ask. Also, the weighted mid-price updates following every imbalance change (Gatheral and Oomen~\cite{GJORC10}, Robert and Rosenbaum~\cite{RCYRM12}). Stoikov~\cite{SS18} then defines the notion of \textit{micro-price}. The micro-price $P^{micro}$ incorporates the effect of mid-price $M$, order book imbalance $I$, and the bid-ask spread $S \coloneqq P^{b} - P^{a}$ into the components of the underlying price: in mathematical form, we can write the dependence as
\begin{align}
P^{micro} \coloneqq M + g (I, S), \label{Eq:1}
\end{align}
using a function $g$, which can be empirically estimated. Motivated by this spirit, we will formulate what we call the \textit{fundamental price} such that the Markovian environment affects the fundamental price and market price.

Theoretical studies have analyzed a trade execution strategy that exploits microstructural factors (e.g., G\^{a}rleanu and Pedersen~\cite{GNPLH13, GNPLH16}, Cartea et~al.~\cite{CADRJS18}, Neuman and Vo{\ss}~\cite{NEVM22}, Forde et~al.~\cite{FMSBLSB22}, Fouque et~al.~\cite{FJPJSSYF22}). Among them, Cartea and Jaimungal~\cite{CAJS16MFE} and Lehalle and Neuman~\cite{LCANE19} are closely related to our analysis. They investigate the influence of Markovian microstructure signal on a single large trader's optimal execution strategy. Our model differs from these two studies in the following aspects. First, we focus on an interaction between two large traders. To this end, our problem is formulated as a Markov game model. Second, our formulation of the Markovian environment enables us to examine how multiple (two) large traders tackle strategic and environmental uncertainty, interact with each other, and execute orders.

\section{Execution game model} \label{Section_3}

\subsection{Market} \label{Subsection_3.1}

In a discrete-time framework $t \in \{1, \ldots, T, T + 1\}$ (where $T \in \mathbb{Z}_{++} \coloneqq \{1, 2, \ldots\}$), we assume that two risk-averse large traders, denoted by $i \in \{1, 2\}$, must purchase $\mathfrak{Q}^{i}$ $(\in \mathbb{R})$ volume of one financial asset by the time $T + 1$. In the sequel, $q^{i}_{t}$ $(\in \mathbb{R})$ stands for the large amount of orders submitted by the large trader $i \in \{1, 2\}$ at time $t \in \{1, \ldots, T\}$.
\begin{remark}[] \label{Remark_3.1}
The following setup allows us to analyze large traders' buying and selling problems.
For each large trader $i \in \{1, 2\}$, the positive $q^{i}_{t}$ for $t \in \{1, \ldots, T\}$ stands for the acquisition and negative $q^{i}_{t}$ the liquidation of the financial asset.
\end{remark}
\noindent We define $\overline{Q}^{i}_{t}$ as the number of shares remained to purchase (hereafter \textit{remained execution volume}) by large trader $i \in \{1, 2\}$ at time $t \in \{1, \ldots, T, T + 1\}$. So $\{\overline{Q}^{i}_{t}\}_{t \in \{1, \ldots, T\}}$ satisfies
\begin{align}
\overline{Q}_{t + 1}^{i} = \overline{Q}^{i}_{t} - q^{i}_{t}, \label{heldquantity}
\end{align}
with the initial and terminal conditions for each large trader $i \in \{1, 2\}$:$\overline{Q}_{1}^{i} = \mathfrak{Q}^{i} \in \mathbb{R}$; $\overline{Q}_{T + 1}^{i} = 0$. In the rest of this paper, the buy-trade and sell-trade of a large trader are supposed to induce the same instantaneous linear price impact. This assumption is justified by some empirical studies (e.g., Cartea and Jaimungal~\cite{CAJS16MFE, CAJS16SIAM}).

The market price (or quoted price) of the financial asset at time $t \in \{1, \dots, T, T + 1\}$ is given by $P_{t}$. Then, the execution price of the asset becomes $\widehat{P}_{t}$ since the large traders submit a large number of orders, influencing the financial asset price at which they execute the transaction. In the rest of this paper, we assume that submitting one unit of (large) order at time $t \in \{1, \ldots, T\}$ causes the instantaneous price impact denoted as $\lambda_{t} \in \mathbb{R}_{++} \coloneqq (0, \infty)$.

We here define the \textit{residual effect of past price impacts} caused by both large traders at time $t \in \{1, \ldots, T\}$. The residual effect of past price impacts represents the discounted sum of past temporary price impacts and is a key determinant of the market price. Existing theoretical and empirical studies highlight the significance of the transient nature of price impacts (e.g., Bouchaud et~al.~\cite{BJPGYPMWM04}, Gatheral~\cite{GJ10}, Obizhaeva and Wang~\cite{OAAWJ13}). We formulate the residual effect of past price impacts by the following exponential function $G \colon \mathbb{R} \rightarrow \mathbb{R}_{++}$:
\begin{align}
G (t) \coloneqq \mathrm{e}^{- \rho t}, \label{decaykernel}
\end{align}
where $\rho \in \mathbb{R}_{+} \coloneqq [0, \infty)$ stands for the deterministic resilience speed.\footnote{We proceed with the following analysis without assuming the time-dependency of the resilience speed to simplify the notation.} The dynamics of the residual effect of past price impacts, represented by $\{ R_{t} \}_{t \in \{1, \ldots, T\}}$, are defined as follows:
\begin{align}
R_{t + 1} &\coloneqq \sum_{k = 1}^{t} \alpha_{k} \lambda_{k} (q_{k}^{1} + q_{k}^{2}) \mathrm{e}^{- \rho \left( (t + 1) - k \right)} \nonumber\\
&= \mathrm{e}^{-\rho} \sum_{k = 1}^{t - 1} \alpha_{k} \lambda_{k} (q_{k}^{1} + q_{k}^{2}) \mathrm{e}^{- \rho (t - k)} + a_{t} \lambda_{t} \left( q^{1}_{t} + q^{2}_{t} \right) \mathrm{e}^{- \rho} \nonumber\\
&= \mathrm{e}^{- \rho} \left[ R_{t} + \alpha_{t} \lambda_{t} (q^{1}_{t} + q^{2}_{t}) \right], \quad t = 1, \dots, T, \label{resilience}
\end{align}
where $\alpha_{t} \in [0, 1]$ represents the linear price impact coefficient representing the temporary price impact. We assume that $R_{0} = 0$, implying that there is no residual effect at the beginning of the trading window.

In our model, the \textit{fundamental price} of the financial asset plays an essential role when defining the dynamics of the market price of the financial asset. We first define the dynamics of the fundamental price at time $t \in \{1, \ldots, T\}$, denoted by $P^{f}_{t}$, and then obtain the dynamics of the market price. Since the residual effect of the past price impacts dissipates over the trading window, we define $P_{t} - R_{t}$ as the fundamental price of the financial asset, that is,
\begin{align}
P^{f}_{t} \coloneqq P_{t} - R_{t}.
\end{align}
\begin{remark}[] \label{Remark_3.2}
In the field of market microstructure, the model of the fundamental price in this paper is different from seminal papers that are inspired by Kyle~\cite{KAS85}.
\end{remark}

We assume that there are three factors in the financial market that affect the fundamental price (and thus the market price). The first factor is the public news/information about the economic situation between $t$ and $t + 1$, defined as a sequence of independent random variables $\{\epsilon_{t}\}_{t \in \{1, \ldots, T\}}$. Each $\epsilon_{t}$ for $t \in \{1, \ldots, T\}$ follows a normal distribution with mean $\mu^{\epsilon}_{t} \in \mathbb{R}$ and variance $(\sigma^{\epsilon}_{t})^{2} \in \mathbb{R}_{++}$:
\begin{align}
\epsilon_{t} \sim N \big( \mu^{\epsilon}_{t}, (\sigma^{\epsilon}_{t})^{2} \big), \quad t = 1, \dots, T.
\end{align}
In the sequel of this paper, we assume that $\mu^{\epsilon}_{t} = 0$ for all $t \in \{1, \ldots, T\}$.

We assume that the linear \textit{permanent price impact}, represented by 
\begin{align}
\beta_{t} \lambda_{t} (q^{1}_{t} + q^{2}_{t}), 
\end{align}
where $\beta_{t} \in [0, 1]$ is the second factor affecting the fundamental price (and market price). This assumption stems from the following idea: the permanent price impact still affects the price at the next execution time (by definition) as \textit{it affects the fundamentals}.

The third factor that affects the fundamental price is what we call the \textit{Markovian environment}. The Markovian environment, denoted by $\{\mathcal{I}_{t}\}_{t \in \{1, \ldots, T\}}$, describes an exogenous factor that influences the fundamental price (and market price). We assume that the probabilistic law of Markovian environment has a \textit{Markovian dependence} as follows:
\begin{align}
\begin{split}
\mathcal{I}_{0} &= 0; \\
\mathcal{I}_{t + 1} \vert_{\mathcal{I}_{t}} &\sim N \big( a_{t + 1}^{\mathcal{I}} - b_{t + 1}^{\mathcal{I}} \mathcal{I}_{t}, (\sigma_{t + 1}^{\mathcal{I}})^{2} \big), \quad t = 0, \ldots, T - 1.
\end{split}
\label{STAVwithMD1}
\end{align}
Note that $\{a^{\mathcal{I}}_{t}\}_{t \in \{1, \ldots, T\}}$, $\{b^{\mathcal{I}}_{t}\}_{t \in \{1, \ldots, T\}}$, and $\{\sigma^{\mathcal{I}}_{t}\}_{t \in \{ 1, \ldots, T\}}$ are all deterministic functions of time. In general, we can assume that the system equation for the Markovian environment is defined as $\mathcal{I}_{t + 1} \coloneqq k_{t} (\mathcal{I}_{t}, \omega_{t})$ through a (Borel-measurable) function $k_{t + 1} \colon \mathbb{R} \times \mathbb{R} \rightarrow \mathbb{R}$. Here $\{ \omega_{t} \}_{t \in \{1, \ldots, T\}}$ is an i.i.d.~sequence, where $\omega_{t} \sim N (0, 1)$ for all $t \in \{1, \ldots, T\}$. We focus on an important case given by
\begin{align}
\begin{split}
\mathcal{I}_{0} &= 0; \\
\mathcal{I}_{t + 1} &= (a_{t + 1}^{\mathcal{I}} - b_{t + 1}^{\mathcal{I}} \mathcal{I}_{t}) + \sigma_{t + 1}^{\mathcal{I}} \omega_{t + 1}, \quad t = 0, \ldots, T - 1.
\end{split}
\label{STAVwithMD2}
\end{align}
\begin{remark}[Implication of Markovian environment] \label{Remark_3.3}
The interpretation of a Markovian environment is various and needs to be carefully mentioned. For example, we can consider the Markovian environment as the price impact caused by \textit{aggregate orders of small traders} (or \textit{other traders' order-flow}). Theoretical studies investigate the effect of other traders' order-flow that has a Markovian dependence and show that the order-flow directly affects the optimal execution strategy for a single large trader (e.g., Cartea and Jaimungal~\cite{CAJS16MFE}, Fukasawa et~al.~\cite{FMOMSM21IJTAF}, Ohnishi and Shimoshimizu~\cite{OMSM22APFM}). Another way to interpret the Markovian environment is through the \textit{order book imbalance} (\textit{OBI}). Existing studies (e.g., Cont et~al.~\cite{CRKASS14}, Stoikov~\cite{SS18}) highlight the importance of incorporating OBI into the formulation of market price dynamics. Cartea and Jaimungal~\cite{CAJS16MFE}, Cartea et~al.~\cite{CADRJS18} and Lehalle and Neuman~\cite{LCANE19} investigate an optimal execution strategy focusing on OBI (or the ``microstructural signal'') and show that the microstructural signal influences the optimal execution strategy. Adding to the two examples, we can interpret the Markovian environment in several ways as follows: trader's \textit{view of the market} (Cartea et~al.~\cite{CAJSKD16}); \textit{effect of other assets} that large traders do \textit{not trade} (Cartea et~al.~\cite{CAGLJS19}); \textit{difference between the weighted returns of two assets} in the same industry or with similar characteristics (Avellaneda and Lee~\cite{AMLJH10}).
\end{remark}
\begin{remark}[] \label{Remark_3.4}
We omit the detailed classification of the dynamics of the Markovian environment, $\{\mathcal{I}_{t}\}_{t \in \{1, \ldots, T\}}$, in this paper. For the details, see Fukasawa et~al.~\cite{FMOMSM21IJTAF}. The paper classifies the dynamics of Eq.~\eqref{STAVwithMD2} in terms of various conditions for $\{a^{\mathcal{I}}_{t}\}_{t \in \{1, \ldots, T\}}$ and $\{b^{\mathcal{I}}_{t}\}_{t \in \{1, \ldots, T\}}$.
\end{remark}
Here we make the following assumptions for the correlation between two stochastic processes, $\{ \mathcal{I}_{t} \}_{t \in \{1, \ldots, T\}}$ and $\{ \epsilon_{t} \}_{t \in \{1, \ldots, T\}}$.
\begin{assumption}[] \label{Assumption_3.1}
For each time $t \in \{1, \ldots, T\}$, $\{ \mathcal{I}_{t} \}_{t \in \{1, \ldots, T\}}$ and $\{ \epsilon_{t} \}_{t \in \{1, \ldots, T\}}$ are defined on a filtered probability space $(\Omega, \mathcal{F}, \{ \mathcal{F}_{t} \}_{t \in \{1, \ldots, T\}}, \mathbb{P})$ where $\mathcal{F}_{t} \coloneqq \sigma \left( (\omega_{s}, \epsilon_{s})_{s \in \{ 1, \ldots, t \}} \right)$. Then, for all $t \in \{1, \ldots, T\}$, $\mathcal{I}_{t}$ and $\epsilon_{t}$ are \textit{correlated} with correlation coefficient $\rho^{\mathcal{I}, \epsilon} \in (-1, 1)$ \textit{given} $\mathcal{I}_{t - 1}$. So we have
\begin{align}
\begin{pmatrix}
\mathcal{I}_{t + 1} \\
\epsilon_{t + 1}
\end{pmatrix}
\Big\vert_{\mathcal{I}_{t}}
\sim
N \left(
\begin{pmatrix}
a_{t + 1}^{\mathcal{I}} - b_{t + 1}^{\mathcal{I}} \mathcal{I}_{t} \\
\mu_{t + 1}^{\epsilon}
\end{pmatrix},
\begin{pmatrix}
(\sigma_{t + 1}^{\mathcal{I}})^{2} & \rho^{\mathcal{I}, \epsilon} \sigma_{t + 1}^{\mathcal{I}} \sigma_{t + 1}^{\epsilon} \\
\rho^{\mathcal{I}, \epsilon} \sigma_{t + 1}^{\mathcal{I}} \sigma^{\epsilon}_{t + 1} & (\sigma_{t + 1}^{\epsilon})^{2}
\end{pmatrix}
\right).
\end{align}
\end{assumption}

From the assumption that the public news/information, the permanent price impact, and the Markovian environment affect the fundamental price, we define the dynamics of the fundamental price $P^{f}_{t}$ $(= P_{t} - R_{t})$ as follows:
\begin{align}
P^{f}_{t + 1}
&= P_{t + 1} - R_{t + 1} \nonumber\\
&\coloneqq P^{f}_{t} + \beta_{t} \lambda_{t} (q^{1}_{t} + q^{2}_{t}) + \mathcal{I}_{t} + \epsilon_{t} \nonumber\\
&= P_{t} - R_{t} + \beta_{t} \lambda_{t} (q^{1}_{t} + q^{2}_{t}) + \mathcal{I}_{t} + \epsilon_{t}, \quad t = 1, \ldots, T. \label{fundprice}
\end{align}
\begin{remark}[] \label{Remark_3.5}
The dynamics of the fundamental price imply that the permanent price impact and the Markovian environment may give a non-zero trend to the fundamental price. For a more detailed discussion, see Ohnishi and Shimoshimizu~\cite{OMSM20QF}. Also, we can interpret the Markovian environment as an exogenous trend of price dynamics, which supports the assumption that $\mu^{\epsilon}_{t} = 0$ for all $t \in \{1, \ldots, T\}$.
\end{remark}

According to Eq.~\eqref{resilience} and~\eqref{fundprice}, the dynamics of market price are described as
\begin{align}
P_{t + 1}
&= P_{t} + \left( R_{t + 1} - R_{t} \right) + \beta_{t} \lambda_{t} (q^{1}_{t} + q^{2}_{t}) + \mathcal{I}_{t} + \epsilon_{t} \nonumber\\
&= P_{t} - (1 - \mathrm{e}^{- \rho}) R_{t} + (\alpha_{t} \mathrm{e}^{- \rho} + \beta_{t}) \lambda_{t} (q^{1}_{t} + q^{2}_{t}) + \mathcal{I}_{t} + \epsilon_{t}, \quad t = 1, \ldots, T.
\end{align}
\begin{remark}[] \label{Remark_3.6}
In this context,
\begin{align}
\beta_{t} \lambda_{t} (q^{1}_{t} + q^{2}_{t}); \quad
\alpha_{t} \lambda_{t} (q^{1}_{t} + q^{2}_{t}); \quad
\mathrm{e}^{-\rho} \alpha_{t} \lambda_{t} (q^{1}_{t} + q^{2}_{t}),
\end{align}
represent the permanent price impact, temporary price impact, and transient price impact at time $t \in \{1, \ldots, T\}$, respectively. Figure~\ref{Figure_1} depicts the relationships among the permanent, temporary, and transient price impacts. Moreover, if $\rho \rightarrow \infty$, the residual effect of past price impacts becomes zero since the condition $R_{1} = 0$ and
\begin{align}
\lim_{\rho \to \infty} R_{t + 1}
= \lim_{\rho \to \infty} \mathrm{e}^{- \rho} [R_{t} + \alpha_{t} \lambda_{t} (q^{1}_{t} + q^{2}_{t})]
= 0, \quad t = 1, \ldots, T,
\end{align}
holds for all $t \in \{1, \ldots, T\}$ from Eq.~\eqref{resilience}. Therefore, the dynamics of the market price becomes
\begin{align}
P_{t + 1} = P_{t} + \beta_{t} \lambda_{t} (q^{1}_{t} + q^{2}_{t}) + \mathcal{I}_{t} + \epsilon_{t}, \quad t = 1, \ldots, T,
\end{align}
implying that we have a permanent price impact model with the effect of a Markovian environment.
\end{remark}
We consider the following assumption in the rest of this paper.
\begin{assumption}[] \label{Assumption_3.2}
For $\alpha_{t} \in [0, 1]$, $\beta_{t} \in [0, 1]$, and $\rho \in [0, \infty)$, the relationship
\begin{align}
\alpha_{t} \mathrm{e}^{- \rho} + \beta_{t} < 1 \label{Ass:3.2}
\end{align}
holds for all $t \in \{1, \ldots, T\}$.
\end{assumption}
The implication for Eq.~\eqref{Ass:3.2} is that the friction of permanent and transient price impact at time $t \in \{1, \ldots, T\}$ is strictly less than the price impact caused by both large traders. This assumption is plausible from the perspective of limit order book dynamics (as shown in Figure~\ref{Figure_1}).
\begin{figure}[tbp]
\centering
\tikzset{every picture/.style={line width=0.75pt}} 
\begin{tikzpicture}[x=0.75pt,y=0.75pt,yscale=-0.9,xscale=0.9]
\draw (99,251) -- (538,251) ;
\draw [shift={(540,251)}, rotate = 180] [color={rgb, 255:red, 0; green, 0; blue, 0 } ][line width=0.75] (10.93,-3.29) .. controls (6.95,-1.4) and (3.31,-0.3) .. (0,0) .. controls (3.31,0.3) and (6.95,1.4) .. (10.93,3.29) ;
\draw [color={rgb, 255:red, 0; green, 0; blue, 0 } ,draw opacity=1 ][line width=1.5] (120.96,237) -- (119.87,121.98) -- (119.04,34) ;
\draw [shift={(119,30)}, rotate = 89.46] [fill={rgb, 255:red, 0; green, 0; blue, 0 } ,fill opacity=1 ][line width=0.08] [draw opacity=0] (11.61,-5.58) -- (0,0) -- (11.61,5.58) -- cycle ;
\draw [shift={(121,241)}, rotate = 269.46] [fill={rgb, 255:red, 0; green, 0; blue, 0 } ,fill opacity=1 ][line width=0.08] [draw opacity=0] (11.61,-5.58) -- (0,0) -- (11.61,5.58) -- cycle ;
\draw [color={rgb, 255:red, 0; green, 0; blue, 0 } ,draw opacity=1 ][line width=1.5] (119,30) -- (531,31) ;
\draw [color={rgb, 255:red, 0; green, 0; blue, 0 } ,draw opacity=1 ][line width=1.5] (531.97,166) -- (531.03,35) ;
\draw [shift={(531,31)}, rotate = 89.59] [fill={rgb, 255:red, 0; green, 0; blue, 0 } ,fill opacity=1 ][line width=0.08] [draw opacity=0] (11.61,-5.58) -- (0,0) -- (11.61,5.58) -- cycle ;
\draw [shift={(532,170)}, rotate = 269.59] [fill={rgb, 255:red, 0; green, 0; blue, 0 } ,fill opacity=1 ][line width=0.08] [draw opacity=0] (11.61,-5.58) -- (0,0) -- (11.61,5.58) -- cycle ;
\draw [color={rgb, 255:red, 0; green, 0; blue, 0 } ,draw opacity=1 ][line width=1.5] (532,236) -- (532,174) ;
\draw [shift={(532,170)}, rotate = 90] [fill={rgb, 255:red, 0; green, 0; blue, 0 } ,fill opacity=1 ][line width=0.08] [draw opacity=0] (11.61,-5.58) -- (0,0) -- (11.61,5.58) -- cycle ;
\draw [shift={(532,240)}, rotate = 270] [fill={rgb, 255:red, 0; green, 0; blue, 0 } ,fill opacity=1 ][line width=0.08] [draw opacity=0] (11.61,-5.58) -- (0,0) -- (11.61,5.58) -- cycle ;
\draw [color={rgb, 255:red, 0; green, 0; blue, 0 } ,draw opacity=1 ][line width=1.5] (120,170) -- (532,170) ;
\draw [color={rgb, 255:red, 0; green, 0; blue, 0 } ,draw opacity=1 ][line width=1.5] (121,241) -- (532,240) ;
\draw [line width=1.5] [dash pattern={on 5.63pt off 4.5pt}] (151,239) -- (166.69,37.99) ; 
\draw [shift={(167,34)}, rotate = 94.46] [fill={rgb, 255:red, 0; green, 0; blue, 0 } ][line width=0.08] [draw opacity=0] (11.61,-5.58) -- (0,0) -- (11.61,5.58) -- cycle ;
\draw [shift={(151,239)}, rotate = 274.46] [color={rgb, 255:red, 0; green, 0; blue, 0 } ][fill={rgb, 255:red, 0; green, 0; blue, 0 } ][line width=1.5] (0, 0) circle [x radius= 4.36, y radius= 4.36] ;
\draw (113,31) .. controls (108.33,31) and (106,33.33) .. (106,38) -- (106,126) .. controls (106,132.67) and (103.67,136) .. (99,136) .. controls (103.67,136) and (106,139.33) .. (106,146)(106,143) -- (106,233) .. controls (106,237.67) and (108.33,240) .. (113,240) ;
\draw (539,168) .. controls (543.67,167.97) and (545.98,165.62) .. (545.95,160.95) -- (545.57,109.45) .. controls (545.52,102.78) and (547.83,99.43) .. (552.5,99.4) .. controls (547.83,99.43) and (545.47,96.12) .. (545.42,89.45)(545.45,92.45) -- (545.05,37.95) .. controls (545.02,33.28) and (542.67,30.97) .. (538,31) ;
\draw (540,240) .. controls (544.67,240) and (547,237.67) .. (547,233) -- (547,215.5) .. controls (547,208.83) and (549.33,205.5) .. (554,205.5) .. controls (549.33,205.5) and (547,202.17) .. (547,195.5)(547,198.5) -- (547,178) .. controls (547,173.33) and (544.67,171) .. (540,171) ;
\draw [line width=1.5] (168,29) .. controls (168.99,68.6) and (251.33,121.92) .. (328.66,122.01) ;
\draw [shift={(331,122)}, rotate = 179.27] [color={rgb, 255:red, 0; green, 0; blue, 0 } ][line width=1.5] (14.21,-4.28) .. controls (9.04,-1.82) and (4.3,-0.39) .. (0,0) .. controls (4.3,0.39) and (9.04,1.82) .. (14.21,4.28) ;
\draw [shift={(168,29)}, rotate = 88.57] [color={rgb, 255:red, 0; green, 0; blue, 0 } ][fill={rgb, 255:red, 0; green, 0; blue, 0 } ][line width=1.5] (0, 0) circle [x radius= 4.36, y radius= 4.36] ;
\draw [line width=1.5] (331,237) -- (331,126) ; 
\draw [shift={(331,122)}, rotate = 90] [fill={rgb, 255:red, 0; green, 0; blue, 0 } ][line width=0.08] [draw opacity=0] (11.61,-5.58) -- (0,0) -- (11.61,5.58) -- cycle ;
\draw [shift={(331,241)}, rotate = 270] [fill={rgb, 255:red, 0; green, 0; blue, 0 } ][line width=0.08] [draw opacity=0] (11.61,-5.58) -- (0,0) -- (11.61,5.58) -- cycle ;
\draw (330.56,79) -- (330.56,88.18) -- (334,90.22) -- (327.12,94.3) -- (334,98.38) -- (327.12,102.46) -- (334,106.54) -- (327.12,110.61) -- (334,114.69) -- (327.12,118.77) -- (330.56,120.81) -- (330.56,129.99) ;
\draw [line width=1.5] (330.56,88.18) -- (330.56,81.18) ; 
\draw [shift={(330.56,77.18)}, rotate = 90] [fill={rgb, 255:red, 0; green, 0; blue, 0 } ][line width=0.08] [draw opacity=0] (11.61,-5.58) -- (0,0) -- (11.61,5.58) -- cycle ;
\draw [line width=1.5] [dash pattern={on 5.63pt off 4.5pt}] (330.7,71) -- (334.52,33.97) ;
\draw [shift={(335,30)}, rotate = 96] [fill={rgb, 255:red, 0; green, 0; blue, 0 } ][line width=0.08] [draw opacity=0] (11.61,-5.58) -- (0,0) -- (11.61,5.58) -- cycle ;
\draw [shift={(330.6,72)}, rotate = 270] [color={rgb, 255:red, 0; green, 0; blue, 0 } ][fill={rgb, 255:red, 0; green, 0; blue, 0 } ][line width=1.5] (0, 0) circle [x radius= 4.36, y radius= 4.36] ;
\draw [line width=1.5] [dash pattern={on 5.63pt off 4.5pt}] (331,122) .. controls (385.6,123.95) and (448.76,126.85) .. (486.17,125.14) ;
\draw [shift={(489,125)}, rotate = 180] [color={rgb, 255:red, 0; green, 0; blue, 0 } ][line width=1.5] (14.21,-4.28) .. controls (9.04,-1.82) and (4.3,-0.39) .. (0,0) .. controls (4.3,0.39) and (9.04,1.82) .. (14.21,4.28) ;
\draw (325,127) .. controls (320.33,127.11) and (318.06,129.5) .. (318.17,134.17) -- (318.24,136.72) .. controls (318.4,143.39) and (316.15,146.78) .. (311.48,146.89) .. controls (316.15,146.78) and (318.56,150.05) .. (318.72,156.71)(318.65,153.71) -- (318.83,161.17) .. controls (318.94,165.84) and (321.33,168.11) .. (326,168) ;
\draw (337,120) .. controls (341.67,120) and (344,117.67) .. (344,113) -- (344,111.18) .. controls (344,104.51) and (346.33,101.18) .. (351,101.18) .. controls (346.33,101.18) and (344,97.85) .. (344,91.18)(344,94.18) -- (344,88) .. controls (344,83.33) and (341.67,81) .. (337,81) ;
\draw [color={rgb, 255:red, 0; green, 0; blue, 0 } ,draw opacity=1 ][line width=1.5] (331,169) -- (331,126) ; 
\draw [shift={(331,122)}, rotate = 90] [fill={rgb, 255:red, 0; green, 0; blue, 0 } ,fill opacity=1 ][line width=0.08] [draw opacity=0] (11.61,-5.58) -- (0,0) -- (11.61,5.58) -- cycle ;
\draw [shift={(331,173)}, rotate = 270] [fill={rgb, 255:red, 0; green, 0; blue, 0 } ,fill opacity=1 ][line width=0.08] [draw opacity=0] (11.61,-5.58) -- (0,0) -- (11.61,5.58) -- cycle ;
\draw (337,236) .. controls (341.67,235.96) and (343.98,233.61) .. (343.94,228.94) -- (343.76,210.06) .. controls (343.7,203.39) and (346,200.04) .. (350.67,199.99) .. controls (346,200.04) and (343.64,196.73) .. (343.58,190.06)(343.61,193.06) -- (343.06,131.94) .. controls (343.01,127.27) and (340.66,124.96) .. (335.99,125) ;
\draw (360,209) -- (423,209) ; 
\draw (542,253.4) node [anchor=north west][inner sep=0.75pt] [font=\small] [align=left] {time};
\draw (147,253.4) node [anchor=north west][inner sep=0.75pt] [font=\small] {$t$};
\draw (317.5,254.4) node [anchor=north west][inner sep=0.75pt] [font=\small] {$t + 1$};
\draw (353,188.4) node [anchor=north west][inner sep=0.75pt] [font=\footnotesize] {$(\alpha_{t} \mathrm{e}^{-\rho } +\beta_{t}) \lambda_{t} (q^{1}_{t} +q^{2}_{t})$};
\draw (176,136.4) node [anchor=north west][inner sep=0.75pt] [font=\footnotesize] {(iv) $\alpha_{t} \mathrm{e}^{-\rho} \lambda_{t} (q^{1}_{t} +q^{2}_{t})$};
\draw (557,193.4) node [anchor=north west][inner sep=0.75pt] [font=\footnotesize] {$\text{(ii) } \beta_{t} \lambda_{t} (q^{1}_{t} +q^{2}_{t})$};
\draw (556,90.4) node [anchor=north west][inner sep=0.75pt] [font=\footnotesize] {$\text{(iii) } \alpha_{t} \lambda_{t} (q^{1}_{t} +q^{2}_{t})$};
\draw (353,93.62) node [anchor=north west][inner sep=0.75pt] [font=\footnotesize] {$\text{(v) } \mathcal{I}_{t} +\epsilon _{t}$};
\draw (160,222) node [anchor=north west][inner sep=0.75pt] [font=\footnotesize] {$P_{t}$};
\draw (143,34) node [anchor=north west][inner sep=0.75pt] [font=\footnotesize] {$\widehat{P}_{t}$};
\draw (295,54.4) node [anchor=north west][inner sep=0.5pt] [font=\footnotesize] {$P_{t + 1}$};
\draw (5.0,125.4) node [anchor=north west][inner sep=0.75pt] [font=\footnotesize] {$\text{(i) } \lambda _{t} (q^{1}_{t} +q^{2}_{t})$};
\draw (375,213.4) node [anchor=north west][inner sep=0.75pt] [font=\footnotesize] {$< 1$ (\textbf{Assumption~\ref{Assumption_3.2}})};
\draw (81,12) node [anchor=north west][inner sep=0.75pt] [align=left, font=\small] {price};
\end{tikzpicture}
\caption{Graphical image of price dynamics when $\alpha_{t} + \beta_{t} = 1$ and Assumption \ref{Assumption_3.2} holds.}
\label{Figure_1}
\end{figure}
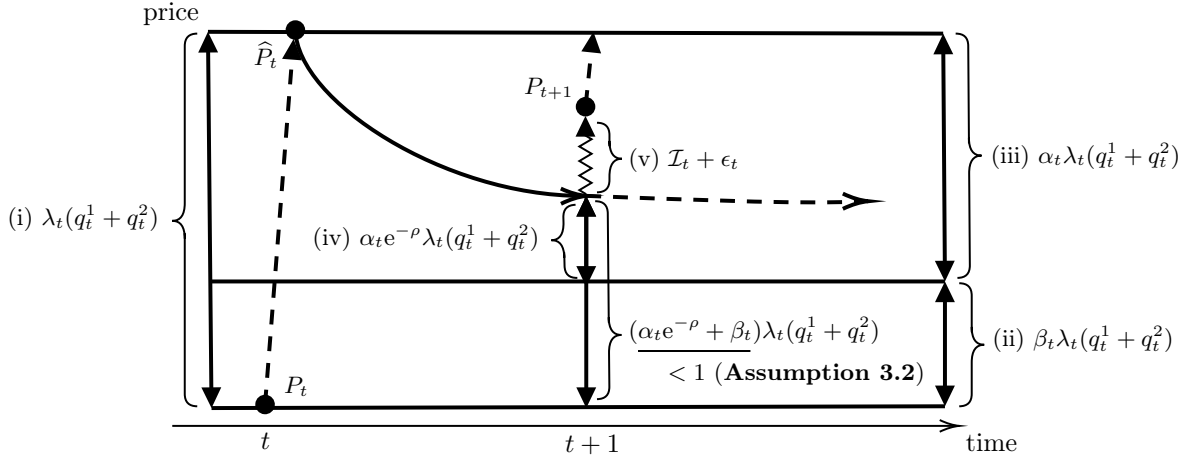

From the definition of the execution price, the wealth process for each large trader $i \in \{ 1, 2 \}$, denoted by $\{W^{i}_{t}\}_{t \in \{1, \ldots, T\}}$, evolves as follows:
\begin{align}
W_{t + 1}^{i}
= W^{i}_{t} - \widehat{P}_{t} q^{i}_{t}
= W^{i}_{t} - \left\{ P_{t} + \lambda_{t} \left( q^{1}_{t} + q^{2}_{t} \right) \right\} q^{i}_{t}, \quad t = 1, \ldots, T. \label{wealthprocess}
\end{align}

\subsection{Formulation as a Markov game} \label{Subsection_3.2}

In the above discrete-time window, we first define the state of the decision process at time $t \in \{1, \ldots, T, T + 1\}$ as $7$-tuple and denote it as
\begin{align}
\bm{s}_{t} = \big( W^{1}_{t}, W^{2}_{t}, P_{t}, \overline{Q}^{1}_{t}, \overline{Q}^{2}_{t}, R_{t}, \mathcal{I}_{t - 1} \big)
\in \mathbb{R} \times \mathbb{R} \times \mathbb{R} \times \mathbb{R} \times \mathbb{R} \times \mathbb{R} \times \mathbb{R} \eqqcolon \mathbb{S}.
\end{align}
For $t \in \{1, \ldots, T\}$, an allowable action chosen at state $\bm{s}_{t}$ is an execution volume $q^{i}_{t} \in \mathbb{R} \eqqcolon \mathbb{A}^{i}$ so that the set $\mathbb{A}^{i}$ of admissible actions is independent of the current state $\bm{s}_{t} \in \mathbb{S}$.

When an action $q^{i}_{t}$ is chosen in a state $\bm{s}_{t}$ at time $t \in \{1, \ldots, T\}$, a transition to a next state
\begin{align}
\bm{s}_{t + 1} = \big( W_{t + 1}^{1}, W_{t + 1}^{2}, P_{t + 1}, \overline{Q}_{t + 1}^{1}, \overline{Q}_{t + 1}^{2}, R_{t + 1}, \mathcal{I}_{t} \big) \in \mathbb{S}
\end{align}
occurs according to the law of motion which we have precisely described in Section~\ref{Subsection_3.1}. The transition is symbolically described by a Borel measurable system dynamics function $\bm{h}_{t}$ ($\colon \mathbb{S} \times \mathbb{A}^{i} \times \mathbb{A}^{j} \times (\mathbb{R} \times \mathbb{R}) \longrightarrow \mathbb{S}$) as
\begin{align}
\bm{s}_{t + 1} = \bm{h}_{t} \big( \bm{s}_{t}, q^{1}_{t}, q^{2}_{t}, (\omega_{t}, \epsilon_{t}) \big), \quad t = 1, \ldots, T.
\end{align}

We assume that each large trader $i \in \{ 1, 2 \}$ has a Constant Absolute Risk Aversion (CARA)-type Von Neumann-Morgenstern (vN-M) (or negative exponential) utility function with the absolute risk aversion parameter $\gamma^{i} > 0$. A utility payoff (or reward) arises only from a terminal state $\bm{s}_{T + 1} \in \mathbb{S}$ at the end of the trading window as
\begin{align}
g_{T + 1}^{i} (\bm{s}_{T + 1}) \coloneqq
\begin{cases}
- \exp \big\{ - \gamma^{i} W_{T + 1}^{i} \big\} & \text{if} \ \overline{Q}_{T + 1}^{i} = 0; \\
- \infty & \text{if} \ \overline{Q}_{T + 1}^{i} \neq 0. \label{utilitypayoff}
\end{cases}
\end{align}
The term $- \infty$ indicates a hard constraint enforcing each large trader to execute all of the remaining volumes at time $T$, that is, $q^{i}_{T} = \overline{Q}^{i}_{T}$ for $i \in \{1, 2\}$. The types of large traders could be defined by
\begin{align}
(W_{1}^{i}, \mathfrak{Q}^{i}, \gamma^{i}), \quad i = 1, 2,
\label{ck}
\end{align}
and are assumed to be their \textit{common knowledge}.
\begin{remark} \label{Remark_3.7}
In the real financial markets, large traders may have little access to the counterpart's information expressed by \eqref{ck}. The purpose of our analysis is to investigate how two large traders execute their orders under the existence of strategic and environmental uncertainties. Thus, we formulate the Markov game as a dynamic game of complete information. The above (hypothesized) definition and assumption expressed in \eqref{ck} are essential so that the solution concept of a Nash equilibrium in a non-cooperative game is (rationally or ideally) applicable in this model. To formulate a more general model as a dynamic game of incomplete information requires further intricate analysis and is left for our future research.
\end{remark}

If we define a history-independent one-stage decision rule $f^{i}_{t}$ for large trader $i \in \{1, 2\}$ at time $t \in \{1, \ldots, T\}$ by a Borel measurable map from a state $\bm{s}_{t} \in \mathbb{S} = \mathbb{R}^{7}$ to an action:
\begin{align}
q^{i}_{t} = f^{i}_{t} (\bm{s}_{t}) \in \mathbb{A}^{i},
\end{align}
then a Markov execution strategy for large trader $i$, $\pi^{i}$, is defined as a sequence of one-stage decision rules:
\begin{align}
\pi^{i} \coloneqq (f^{i}_{1}, \ldots, f^{i}_{t}, \ldots, f^{i}_{T}). \label{pi}
\end{align}
We denote the set of all Markov execution strategies as $\Pi^{i}_{\mathrm{M}}$. Further, for $t \in \{1, \ldots, T\}$, we define the sub-execution strategy after time $t$ of a Markov execution strategy $\pi^{i} \in \Pi^{i}_{\mathrm{M}}$ as
\begin{align}
\pi^{i}_{t} \coloneqq (f^{i}_{t}, \ldots, f^{i}_{T}), \label{pit}
\end{align}
and the entire set of $\pi^{i}_{t}$ as $\Pi^{i}_{\mathrm{M}, t}$.

By definition \eqref{utilitypayoff}, the value function under a pair of Markov execution strategies $(\pi^{1}, \pi^{2}) \in \Pi^{1}_{M} \times \Pi^{2}_{M}$ becomes an expected utility payoff arising from the terminal wealth $W_{T + 1}^{i}$ of each large trader $i \in \{1, 2\}$:
\begin{align}
V^{i}_{1} (\pi^{1}, \pi^{2}) \big[ \bm{s}_{1} \big]
&\coloneqq \mathbb{E}_{1}^{(\pi^{1}, \pi^{2})} \Big[ g_{T + 1}^{i} (\bm{s}_{T + 1}) \Big\vert \bm{s}_{1} \Big] \nonumber\\
&= \mathbb{E}_{1}^{\pi} \Big[ - \exp \big\{ -\gamma^{i} W_{T + 1}^{i} \big\} \cdot \mathds{1}_{\{\overline{Q}_{T + 1}^{i} = 0\}} + (- \infty) \cdot \mathds{1}_{\{\overline{Q}_{T + 1}^{i} \neq 0\}} \Big\vert \bm{s}_{1} \Big]. \label{V1pi}
\end{align}
Then, for $t \in \{1, \ldots, T, T + 1\}$ and $\bm{s}_{t} \in \mathbb{S}$, we further let
\begin{align}
V^{i}_{t} (\pi^{1}_{t}, \pi^{2}_{t}) \big[ \bm{s}_{t} \big]
&\coloneqq \mathbb{E}^{(\pi^{1}_{t}, \pi^{2}_{t})}_{t} \Big[ g_{T + 1}^{i} (\bm{s}_{T + 1}) \Big\vert \bm{s}_{t} \Big] \nonumber\\
&= \mathbb{E}^{(\pi^{1}_{t}, \pi^{2}_{t})}_{t} \Big[ - \exp \big\{ - \gamma^{i} W_{T + 1}^{i} \big\} \cdot \mathds{1}_{\{ \overline{Q}_{T + 1}^{i} = 0 \}} + (- \infty) \cdot \mathds{1}_{\{\overline{Q}_{T + 1}^{i} \neq 0\}} \Big\vert \bm{s}_{t} \Big], \label{Vtpi}
\end{align}
be the expected utility payoff at time $t$ under the strategy $(\pi^{1}_{t}, \pi^{2}_{t}) \in \Pi^{1}_{M} \times \Pi^{2}_{M}$. Note that each expression of the conditional expectations, $\mathbb{E}^{(\pi^{1}, \pi^{2})}_{1}$ in Eq.~\eqref{V1pi} and $\mathbb{E}^{(\pi^{1}_{t}, \pi^{2}_{t})}_{t}$ in Eq.~\eqref{Vtpi}, implies the dependence of the probability laws on the strategy profiles, $(\pi^{1}, \pi^{2})$ and $(\pi^{1}_{t}, \pi^{2}_{t})$, respectively. Also, $\mathds{1}_{A}$ stands for the indicator function of a measurable set (or an event) $A \in \mathcal{F}$.

We here seek an execution strategy for two large traders at a Markov perfect equilibrium. First, the definition of a Nash equilibrium in this model becomes as follows.
\begin{definition}[Nash equilibrium] \label{Definition_3.1}
$(\pi^{1\ast}, \pi^{2\ast}) \in \Pi^{1}_{\mathrm{M}} \times \Pi^{2}_{\mathrm{M}}$ is a \textit{Nash equilibrium} starting from a fixed initial state $\bm{s}_{1} \in \mathbb{S}$ if and only if
\begin{align}
V^{1}_{1} (\pi^{1\ast}, \pi^{2\ast}) \big[ \bm{s}_{1} \big] &\geq V^{1}_{1} (\pi^{1}, \pi^{2\ast}) \big[ \bm{s}_{1} \big], \quad \forall \pi^{1} \in \Pi^{1}_{\mathrm{M}}; \\
V^{2}_{1} (\pi^{1\ast}, \pi^{2\ast}) \big[ \bm{s}_{1} \big] &\geq V^{2}_{1} (\pi^{1\ast}, \pi^{2}) \big[ \bm{s}_{1} \big], \quad \forall \pi^{2} \in \Pi^{2}_{\mathrm{M}}.
\end{align}
\end{definition}

We can define a refinement of the Nash equilibrium of this model as the notion of a Markov perfect equilibrium:
\begin{definition}[Markov perfect equilibrium] \label{Definition_3.2}
$(\pi^{1\ast}, \pi^{2\ast}) \in \Pi^{1}_{\mathrm{M}} \times \Pi^{2}_{\mathrm{M}}$ is a \textit{Markov perfect equilibrium} if and only if
\begin{align}
V^{1}_{t} (\pi^{1\ast}_{t}, \pi^{2\ast}_{t}) \big[ \bm{s}_{t} \big] &\geq V^{1}_{t} (\pi^{1}_{t}, \pi^{2\ast}_{t}) \big[ \bm{s}_{t} \big], \quad \forall \pi^{1}_{t} \in \Pi^{1}_{\mathrm{M}, t}, \quad \forall \bm{s}_{t} \in \mathbb{S}, \quad \forall t = 1, \ldots, T;\\
V^{2}_{t} (\pi^{1\ast}_{t}, \pi^{2\ast}_{t}) \big[ \bm{s}_{t} \big] &\geq V^{2}_{t} (\pi^{1\ast}_{t}, \pi^{2}_{t}) \big[ \bm{s}_{t} \big], \quad \forall \pi^{2}_{t} \in \Pi^{2}_{\mathrm{M}, t}, \quad \forall \bm{s}_{t} \in \mathbb{S}, \quad \forall t = 1, \ldots, T.
\end{align}
\end{definition}

Based on the following \textit{one stage (or step, shot) deviation principle}, we obtain an execution strategy at a Markov perfect equilibrium by the backward induction method of dynamic programming.
\begin{align}
V^{1}_{t} (\pi^{1\ast}_{t}, \pi^{2\ast}_{t}) \big[ \bm{s}_{t} \big]
&= \sup_{q^{1}_{t} \in \mathbb{A}^{i}} \mathbb{E} 
\Big[ V^{1}_{t + 1} (\pi^{1\ast}_{t + 1}, \pi^{2\ast}_{t + 1}) \big[ \bm{h}_{t} ( \bm{s}_{t}, (q^{1}_{t}, f^{2\ast}_{t} (\bm{s}_{t}) ), (\omega_{t}, \epsilon_{t}) ) \big] \Big\vert \bm{s}_{t} \Big] \nonumber \\
&= \mathbb{E} 
\Big[ V^{1}_{t + 1} (\pi^{1\ast}_{t + 1}, \pi^{2\ast}_{t + 1}) \big[ \bm{h}_{t} (\bm{s}_{t}, (f^{1\ast}_{t} (\bm{s}_{t}), f^{2\ast}_{t} (\bm{s}_{t})), (\omega_{t}, \epsilon_{t})) \big] \Big\vert \bm{s}_{t} \Big]; \\
V^{2}_{t}(\pi^{1\ast}_{t}, \pi^{2\ast}_{t}) \big[ \bm{s}_{t} \big]
&= \sup_{q^{2}_{t} \in \mathbb{A}^{i}} \mathbb{E} 
\Big[ V^{2}_{t + 1} (\pi^{1\ast}_{t + 1}, \pi^{2\ast}_{t + 1}) \big[ \bm{h}_{t} (\bm{s}_{t}, (f^{1\ast}_{t} (\bm{s}_{t}), q^{2}_{t}), (\omega_{t}, \epsilon_{t})) \big] \Big\vert \bm{s}_{t} \Big] \nonumber \\
&= \mathbb{E} 
\Big[ V^{2}_{t + 1} (\pi^{1\ast}_{t + 1}, \pi^{2\ast}_{t + 1}) \big[ \bm{h}_{t} (\bm{s}_{t}, (f^{1\ast}_{t} (\bm{s}_{t}), f^{2\ast}_{t} (\bm{s}_{t})), (\omega_{t}, \epsilon_{t})) \big] \Big\vert \bm{s}_{t} \Big].
\end{align}

\subsection{Execution strategy at a Markov perfect equilibrium} \label{Subsection 3.3}

\begin{theorem}[Execution strategy at a Markov perfect equilibrium] \label{Theorem_3.1}
\mbox{}
There exists a Markov perfect equilibrium at which the following properties hold for each large trader $i \in \{1, 2\}$:
\begin{enumerate}
\item The execution volume at the Markov perfect equilibrium for large trader $i \in \{1, 2\}$ at time $t \in \{1, \ldots, T\}$, denoted as $q^{i\ast}_{t}$, becomes an \textit{affine function} of the Markovian environment \textit{at time $t - 1$}, the remaining execution volume of each large trader$i, j \in \{1, 2\}$ ($i \neq j$), and the residual effect of past price impacts:
\begin{align}
q^{i\ast}_{t}
&= f_{t} (W^{i}_{t}, W^{j}_{t}, P_{t}, \overline{Q}^{i}_{t}, \overline{Q}^{j}_{t}, R_{t}, \mathcal{I}_{t - 1}) \nonumber\\
&= a^{i}_{t} + b^{i}_{t} \overline{Q}^{i}_{t} + c^{i}_{t} \overline{Q}^{j}_{t} + d^{i}_{t} R_{t} + e^{i}_{t} \mathcal{I}_{t - 1}, \quad t = 1, \ldots, T,
\end{align}
where $a^{i}_{t}, b^{i}_{t}, c^{i}_{t}, d^{i}_{t}, e^{i}_{t}$ for $t \in \{1, \ldots, T\}$ are all deterministic functions of time $t$ which are dependent on the problem parameters and can be computed backwardly in time $t$.
\item The value function $V^{i}_{t} (\pi^{1}_{t}, \pi^{2}_{t}) \big[ \bm{s}_{t} \big]$ at time $t \in \{1, \ldots, T\}$ for each large trader $i \in \{1, 2\}$ is represented as a functional form as
\begin{align}
&V^{i}_{t} (\pi^{1}_{t}, \pi^{2}_{t}) \big[ W^{1}_{t}, W^{2}_{t}, P_{t}, \overline{Q}^{1}_{t}, \overline{Q}^{2}_{t}, R_{t}, \mathcal{I}_{t - 1} \big] \nonumber\\
&= - \exp \Big\{ - \gamma^{i} \Big[ W^{i}_{t} - P_{t} \overline{Q}^{i}_{t} + G^{1i}_{t} (\overline{Q}^{i}_{t})^{2} + G^{2i}_{t} \overline{Q}^{i}_{t} + H^{1i}_{t} \overline{Q}^{i}_{t} R_{t} \nonumber\\
&\quad + H^{2i}_{t} R^{2}_{t} + H^{3i}_{t} R_{t} + J^{1i}_{t} \overline{Q}^{i}_{t} \overline{Q}^{j}_{t} + J^{2i}_{t} \overline{Q}^{j}_{t} R_{t} + J^{3i}_{t} (\overline{Q}^{j}_{t})^{2} +
J^{4i}_{t} \overline{Q}^{j}_{t} \nonumber\\
&\quad + L^{1i}_{t} \overline{Q}^{i}_{t} \mathcal{I}_{t - 1} + L^{2i}_{t} R_{t} \mathcal{I}_{t - 1} + L^{3i}_{t} \overline{Q}^{j}_{t} \mathcal{I}_{t - 1}
+ L^{4i}_{t} \mathcal{I}_{t - 1}^{2} + L^{5i}_{t} \mathcal{I}_{t - 1} + Z^{i}_{t} \Big] \Big\},
\end{align}
where $G^{1i}_{t}, G^{2i}_{t}, H^{1i}_{t}, H^{2i}_{t}, H^{3i}_{t}, J^{1i}_{t}, J^{2i}_{t}, J^{3i}_{t}, J^{4i}_{t}, L^{1i}_{t}, L^{2i}_{t}, L^{3i}_{t}, L^{4i}_{t}, L^{5i}_{t}, Z^{i}_{t}$ for $t \in \{1, \ldots, T\}$ are deterministic functions of time $t$ which are dependent on the problem parameters and can be computed backwardly in time $t$.
\end{enumerate}
\end{theorem}
\begin{proof}
See Appendix~\ref{Section_A}.
\end{proof}

As Theorem~\ref{Theorem_3.1} shows, the execution volume $q^{i\ast}_{t}$ at the Markov perfect equilibrium for $t \in \{1, \ldots, T\}$ depends on the state $\bm{s}_{t} \in \mathbb{S}$ of the decision process through the Markovian environment at the previous time $\mathcal{I}_{t - 1}$, the remaining execution volume of each large trader $\overline{Q}^{i}_{t}$ for $i \in \{1, 2\}$, and the cumulative residual effect $R_{t}$, and not through the wealth of each large trader $W^{i}_{t}$ for $i \in \{1, 2\}$, or market price $P_{t}$. Furthermore, by the definition of the Markovian environment, the execution volume $q^{i\ast}_{t}$ for $t \in \{1, \ldots, T\}$ includes a nondeterministic term $\mathcal{I}_{t - 1}$, and thus becomes dynamic and non-deterministic. Hereafter, we call the execution strategy at the Markov perfect equilibrium as the \textit{equilibrium execution strategy}.

The following statement immediately follows from Theorem~\ref{Theorem_3.1}.
\begin{corollary}[Deterministic equilibrium execution strategy] \label{Corollary_3.1}
If the Markovian environment is \textit{deterministic} in time, so is the equilibrium execution strategy for each large trader.
\end{corollary}

Theorem~\ref{Theorem_3.1} and Corollary~\ref{Corollary_3.1} are one of our contributions to the field of a market impact game. The existing studies, such as Schied and Zhang~\cite{SAZT19}, Luo and Schied~\cite{LXSA20}, Cordoni and Lillo~\cite{CFLF22}, and Cordoni and Lillo~\cite{CFLF23}, reveal that an equilibrium execution strategy is deterministic when minimizing the expected execution cost and considering a mean-variance optimization. As shown in our model, however, the equilibrium execution strategies for risk-averse large traders are usually non-deterministic when they are obtained by backward induction methods of dynamic programming. It is mainly when the Markovian environment is deterministic that the equilibrium execution strategy also becomes deterministic.
\begin{remark}[] \label{Remark_3.8}
If only temporary and permanent price impacts influence the fundamental price, the equilibrium execution volume for each large trader $i \in \{1, 2\}$ at time $t \in \{1, \ldots, T\}$ becomes
\begin{align}
q^{i\ast}_{t} = a^{i}_{t} + b^{i}_{t} \overline{Q}^{i}_{t} + c^{i}_{t} \overline{Q}^{j}_{t} + d^{i}_{t} \mathcal{I}_{t - 1}.
\end{align}
In this case, the Markovian environment still affects the equilibrium execution strategy.
\end{remark}
\begin{remark}[The number of large traders] \label{Remrak_3.9}
We can extend the above model as an $n$ ($> 2$) large traders' execution problem. In that case, however, the difference equation derived in obtaining the equilibrium execution strategy becomes rather complicated. In addition, the extension regarding the number of large traders seems not to reveal any further intriguing results. Thus, we will keep focusing on two large traders' execution problems.
\end{remark}

\section{Simulation-based numerical experiments} \label{Section_4}

Theorem~\ref{Theorem_3.1} characterizes how a Markovian environment affects each large trader's equilibrium execution strategy. The effect of the Markovian environment on the equilibrium execution strategy is identified with environmental uncertainty, whereas the existence of the other large trader indicates strategic uncertainty. An important question is how each large trader facing such environmental and strategic uncertainties behaves in a financial market. To answer this question, we examine comparative statics concerning parameters that determine the Markovian environment through a simulation-based analysis. For the comparative statics with other parameters, see Ohnishi and Shimoshimizu~\cite{OMSM20QF} and Fukasawa et~al.~\cite{FMOMSM21IJTAF}.

This section examines the following three cases: a case where the dynamics of the Markovian environment consist of independent random variables, a case where the Markovian environment follows a random walk, and a case where the Markovian environment follows a general stochastic process. We assume the time homogeneity of the time-dependent parameters for simplicity. The benchmark values for parameters are shown in Table~\ref{Table_1}.
\begin{table}
\centering
\caption{Benchmark values for parameters.}
\begin{tabular}{lccccccccccc}
\toprule
Parameters & $\sigma^{\epsilon}_{t}$ & $\sigma^{\mathcal{I}}_{t}$ & $\rho^{\mathcal{I}, \epsilon}$ & $\alpha_{t}$ & $\beta_{t}$ & $\lambda_{t}$ & $a^{\mathcal{I}}_{t}$ & $b^{\mathcal{I}}_{t}$ & $\rho$ & $\gamma^{i}$ & $T$ \\
\midrule
Benchmark values & 0.02 & 0.01 & 0 & 0.5 & 0.5 & 0.001 & 0 & 0 & 0.1 & 0.001 & 10 \\
\bottomrule
\end{tabular}
\label{Table_1}
\end{table}
\begin{remark}[Implications of parameters] \label{Remark_4.1}
We set $\alpha_{t} = \beta_{t} = 0.5$ as the benchmark values. These parameter values imply that the instantaneous price impact caused by the large traders at time $t \in \{1, \ldots, T\}$ is just half decomposed into temporary and permanent price impacts. The parameter $\sigma^{\mathcal{I}}_{t}$ denotes how high the environmental uncertainty is.
\end{remark}

Our setup for numerical experiments bears a close resemblance to that of Fukasawa et~al.~\cite{FMOMSM21IJTAF} and Ohnishi and Shimoshimizu~\cite{OMSM22APFM}. We run the simulation for $N = 10,000$ sample pathes with generating $T \times N = 100,000$ standard normal random numbers for $\omega_{t}$ for $t \in \{1, \ldots, T\}$ and obtain a realized equilibrium execution volume $q^{i}_{t} (k)$ for $t \in \{1, \ldots, T\}$ and $k \in \{1, \ldots, N\}$.\footnote{We also draw the box-and-whisker plot in the following figures. The bold line in the center of the boxplot shows the median of the sample of the equilibrium execution volume. The top end of the box represents the third quartile, and the bottom end of the box represents the first quartile. The upper and lower whiskers are the largest and smallest data points in the range of ($1\text{st quartile} - 1.5 \times (3\text{rd quartile} - 1\text{st quartile})$) and above $(3\text{rd quartile} + 1.5 \times (3\text{rd quartile} - 1\text{st quartile}))$ and below, respectively.} For each $k \in \{1, \ldots, N\}$, $q^{i}_{t} (k)$ represents the $k$th sample path for $t \in \{1, \ldots, T\}$. In the rest of this paper, the figures illustrate the sample mean of $q^{i\ast}_{t} (k)$ for large trader $i \in \{1, 2\}$:
\begin{align}
\widehat{\mathbb{E} \left[ q^{i\ast}_{t} \right]} \coloneqq \frac{1}{N} \sum_{k = 1}^{N} q^{i}_{t} (k).
\end{align}
Since the Markovian environment conditionally follows a normal distribution, for the set of the sample equilibrium execution volume at each time, the sample mean becomes close to the median.

\subsection{Sequence of independent random variables} \label{Subsection_4.1}

We first illustrate the case when the Markovian environment consists of independent random variables. The setting of benchmark values corresponds to this case because when $a^{\mathcal{I}} = b^{\mathcal{I}} = 0$, $\mathcal{I}_{t}$ satisfies
\begin{align}
\mathcal{I}_{t} = \sigma^{\mathcal{I}} \omega_{t}, \label{eq:sigmaomega}
\end{align}
and $\{ \omega_{t} \}_{t \in \{1, \ldots, T\}}$ is an i.i.d.~random sequence. In this case, we have the following result as an immediate consequence of Theorem~\ref{Theorem_3.1}.
\begin{proposition}[] \label{Proposition_4.1}
If $\{\mathcal{I}_{t}\}_{t \in \{1, \ldots, T\}}$ is a sequence of independent random variables described as \eqref{eq:sigmaomega}, the equilibrium execution volume of each large trader at time $t \in \{1, \ldots, T\}$ becomes
\begin{align}
q^{i\ast}_{t} = a^{i}_{t} + b^{i}_{t} \overline{Q}^{i}_{t} + c^{i}_{t} \overline{Q}^{j}_{t} + d^{i}_{t} R_{t}.
\end{align}
Therefore, the Markovian environment has no direct effect on the equilibrium execution strategy.
\end{proposition}
Proposition~\ref{Proposition_4.1} infers that if $\{\mathcal{I}_{t}\}_{t \in \{1, \ldots, T\}}$ is a sequence of independent normal random variables, the equilibrium execution strategy becomes deterministic. This setting enables us to focus on how environmental uncertainty affects the equilibrium execution strategy.

\subsubsection{Symmetric large traders} \label{Subsubsection_4.1.1}

Let us begin with the case of \textit{symmetric large traders}. Specifically, the initial inventory and the risk aversion parameter of each large trader $i, j \in \{1, 2\}$ $(i \neq j)$ are equal; $\mathfrak{Q}^{i} = \mathfrak{Q}^{j} = 100,000$ and $\gamma^{i} = \gamma^{j} = 0.001$. We then examine how the uncertainty arising from the Markovian environment influences the execution strategy of symmetric large traders.
\begin{figure}[tbp]
\centering
\includegraphics[height=5cm, width=7.5cm]{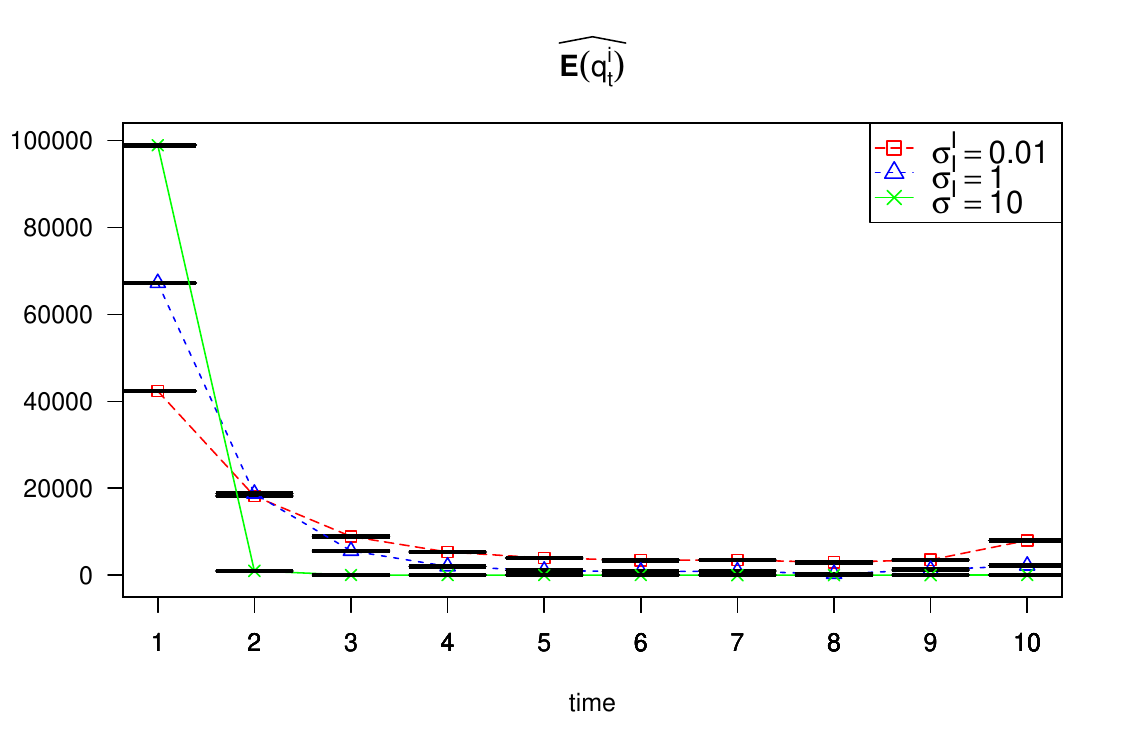}
\caption{Comparative statics for $\sigma^{\mathcal{I}} = 0.01, 1$, and $10$ (with $a^{\mathcal{I}} = b^{\mathcal{I}} = 0$ fixed).}
\label{Figure_2}
\end{figure}

Figure~\ref{Figure_2} illustrates the comparative statics for the equilibrium execution strategy of large traders for different values of $\sigma^{\mathcal{I}}$; $\sigma^{\mathcal{I}} = 0.01, 1$, and $10$. As expected, the larger $\sigma^{\mathcal{I}}$ facilitates faster execution for large traders since risk-averse large traders are inclined to avoid the risk of future price fluctuation. This result is intuitively understandable and is consistent with previous studies.

\subsubsection{Optimal execution and equilibrium execution} \label{Subsubsection_4.1.2}

Of fundamental interest among academic researchers and practitioners is how ``optimal execution'' and ``equilibrium execution'' differ. A growing body of literature has studied an optimal execution strategy for a single large trader. However, a financial market consists primarily of multiple large traders from practitioners' viewpoints. We next examine how a large trader's execution strategy differs depending on the other large trader's existence. Before proceeding to the comparison, we state the following proposition that confirms a relationship between optimal and equilibrium execution strategies. We omit the proof since it is apparent.
\begin{proposition}[] \label{Proposition_4.2}
If $\mathfrak{Q}^{j} = 0$ and $\gamma^{j} \rightarrow \infty$ for $j \in \{1, 2\}$, the equilibrium execution strategy of large trader $i$ ($\neq j$) becomes an optimal execution strategy.
\end{proposition}
\noindent Proposition~\ref{Proposition_4.2} suggests that the equilibrium execution strategy of large trader $i$ $(\neq j)$ is almost the same as an optimal execution strategy when $\mathfrak{Q}^{j} = 0$ and $\gamma^{j}$ is sufficiently large. We next analyze the situation in which the volume that a single large trader unwinds is the same as the total volume that two large traders unwind.
\begin{figure}[tbp]
\centering
\begin{minipage}[b]{0.45\linewidth}
\centering
\includegraphics[height=5cm, width=7.5cm]{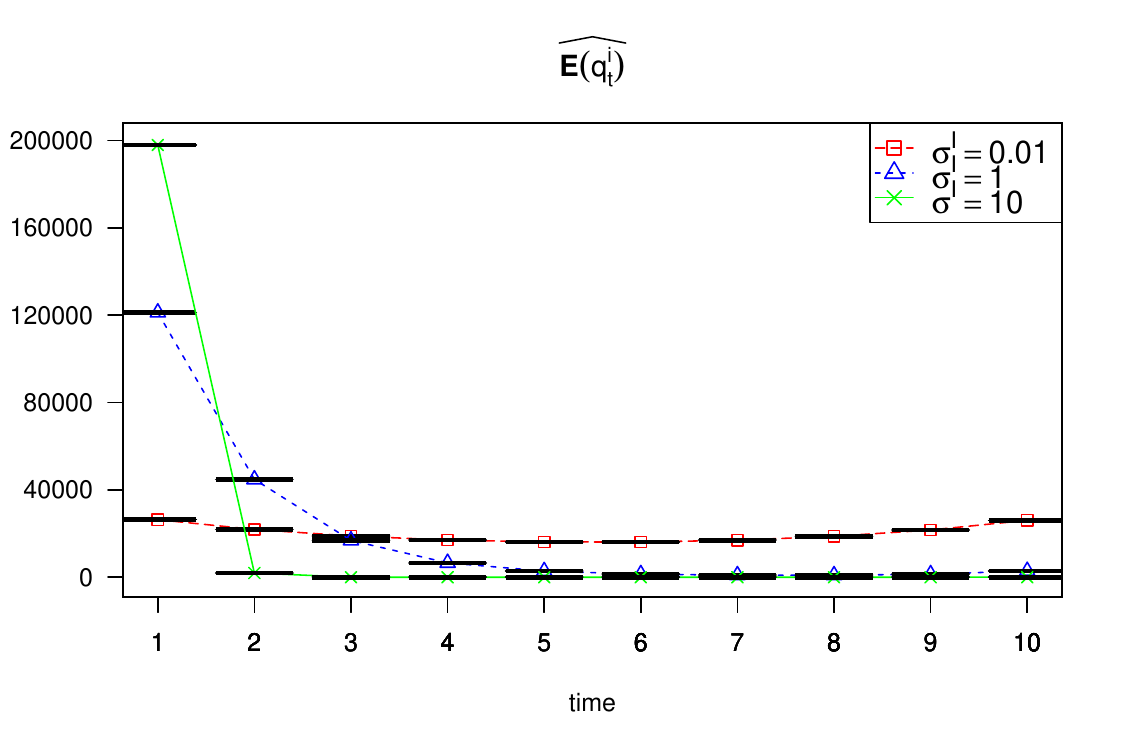}
\subcaption{Trader $i$'s optimal execution volume.}
\label{Figure_3a}
\end{minipage}
\begin{minipage}[b]{0.45\linewidth}
\centering
\includegraphics[height=5cm, width=7.5cm]{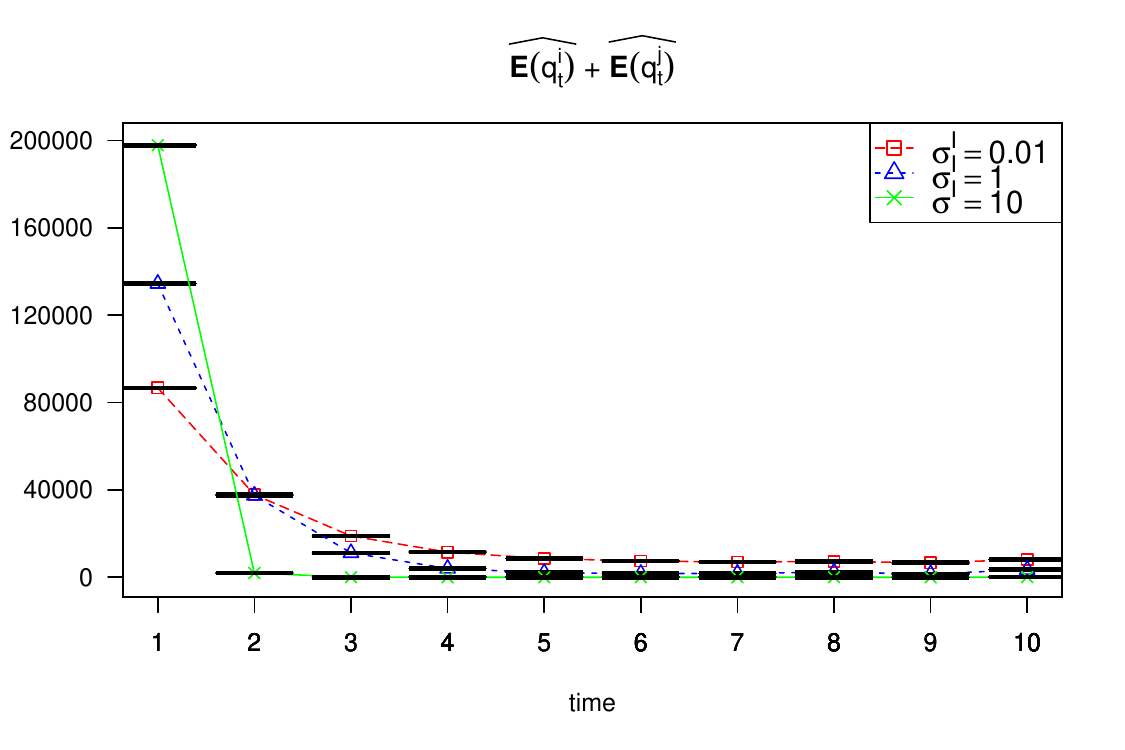}
\subcaption{Total volumes submitted by traders $i$ and $j$.}
\label{Figure_3b}
\end{minipage}
\caption{Optimal execution vs. equilibrium execution for $\sigma^{\mathcal{I}} = 0.01, 1$, and $10$ (with $a^{\mathcal{I}} = b^{\mathcal{I}} = 0$ fixed).}
\label{Figure_3}
\end{figure}

Figure~\ref{Figure_3a} describes the (essentially) single large trader case: $\mathfrak{Q}^{i} = 200,000$ and $\mathfrak{Q}^{j} = 0$ with $\gamma^{i} = 0.001$ and $\gamma^{j} = 1000$. Figure~\ref{Figure_3b} illustrates the equilibrium execution strategy of large traders $i$ and $j$ with the following initial holdings: $\mathfrak{Q}^{i} = \mathfrak{Q}^{j} = 100,000$ with $\gamma^{i} = \gamma^{j} = 0.001$. When $\sigma^{\mathcal{I}} = 0.01$ (i.e., less environmental uncertainty exists), the total volume executed by two large traders at time $t = 1$ is almost as twice as that by a single large trader. However, when $\sigma^{\mathcal{I}} = 10$ (i.e., higher environmental uncertainty exists), the total volume executed by two large traders at time $t = 1$ is almost the same as that by a single large trader. At time $t = 1$, the total volume executed in the case of a single large trader and the case of two large traders becomes similar as $\sigma^{\mathcal{I}}$ (i.e., environmental uncertainty) increases. These results offer the following insights for multiple large traders' behavior. When environmental uncertainty is low, the two large traders execute their orders faster than in the single large trader case in response to strategic uncertainty. On the contrary, large traders are sensitive to environmental uncertainty, so they execute their orders under high environmental uncertainty regardless of the existence of other large traders.

Our results can also provide the backbone for the effect of strategic uncertainty on a large trader's execution strategy. When the environmental uncertainty is low, the optimal execution strategy for a single large trader becomes close to an equally divided execution strategy. However, a single large trader's optimal execution strategy is not similar to the equilibrium execution strategy in a market where two large traders exist. This difference stems from the existence of strategic uncertainty. For the two large traders' case, one large trader is willing to behave as \textit{first-mover} (or \textit{leader}) since the buy-activity by the other large trader (opponent) will push up the future price of the financial asset.

\subsubsection{Asymmetric large traders} \label{Subsubsection_4.1.3}

In financial markets, large traders often face the situation that other large traders initially intend to execute different volumes. In such a situation, examining to what degree environmental uncertainty affects the execution strategy for each large trader is essential to analyze each large trader's execution strategy. From this viewpoint, we discuss the effect of $\sigma^{\mathcal{I}}$ on the large traders' equilibrium execution strategy with different initial inventories.

Our first focus is placed on the following scenario: one large trader $i$ initially intends to execute $\mathfrak{Q}^{i} = 100,000$ volumes of one financial asset, while the other large trader $j$ $(\neq i)$ does $\mathfrak{Q}^{j} = 0$ (no) volumes of that asset. We then investigate how the degree of environmental uncertainty affects each large trader's equilibrium execution strategy.
\begin{figure}[tbp]
\centering
\begin{minipage}[b]{0.45\linewidth}
\centering
\includegraphics[height=5cm, width=7.5cm]{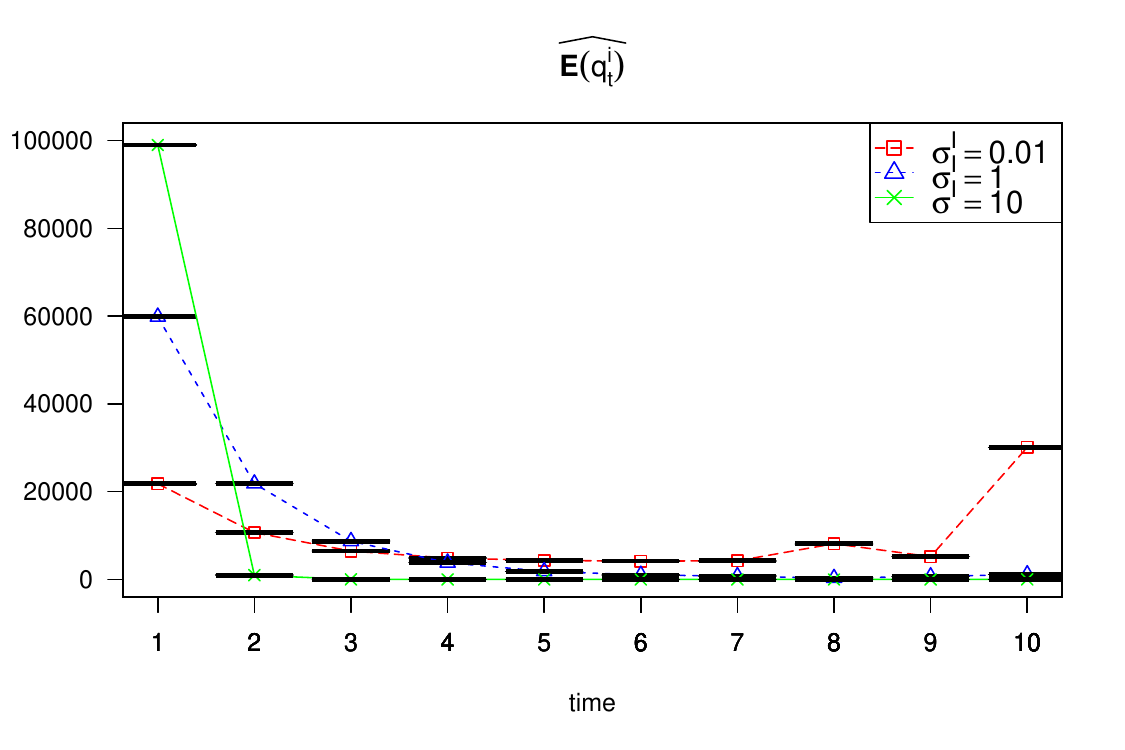}
\subcaption{Trader $i$'s equilibrium execution strategy.}
\label{Figure_4a}
\end{minipage}
\begin{minipage}[b]{0.45\linewidth}
\centering
\includegraphics[height=5cm, width=7.5cm]{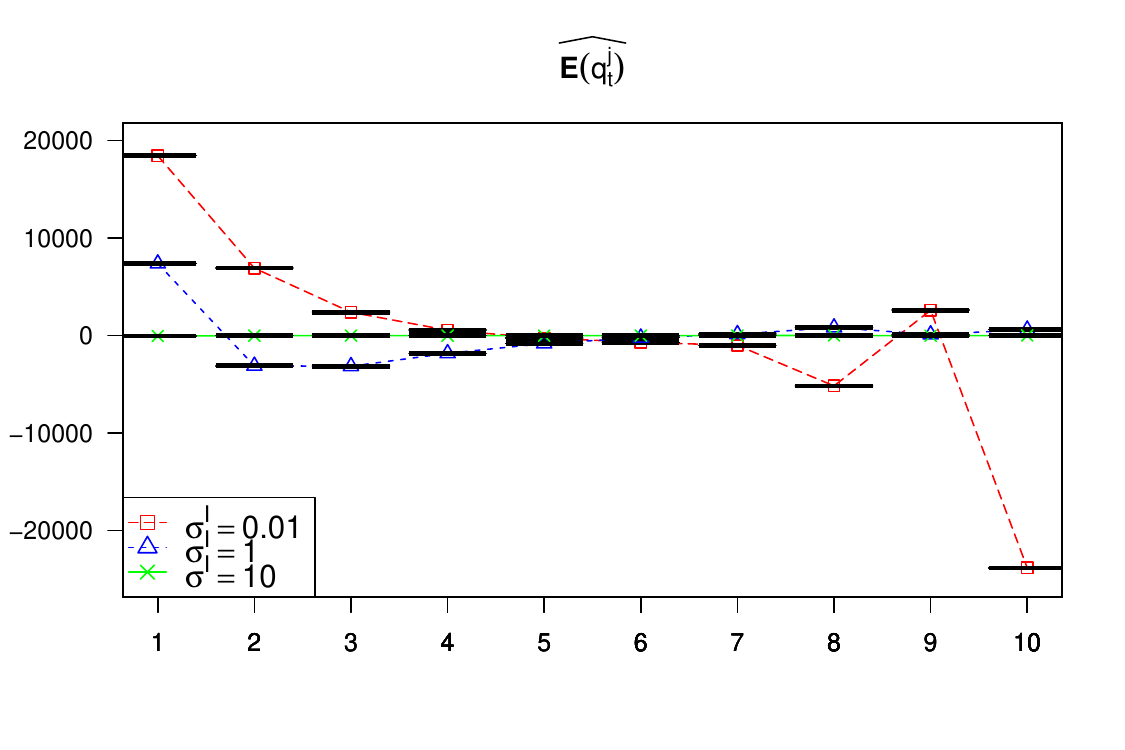}
\subcaption{Trader $j$'s equilibrium execution strategy.}
\label{Figure_4b}
\end{minipage}
\caption{Comparative statics for $\sigma^{\mathcal{I}} = 0.01, 1$, and $10$ (with $a^{\mathcal{I}} = b^{\mathcal{I}} = 0$ fixed and $\mathfrak{Q}^{i} = 100,000$ and $\mathfrak{Q}^{j} = 0$).}
\label{Figure_4}
\end{figure}
Figure~\ref{Figure_4} illustrates the equilibrium execution strategy of large traders $i$ and $j$. The large trader $j$ with no initial inventory executes fewer orders as $\sigma^{\mathcal{I}}$ becomes larger (i.e., environmental uncertainty becomes higher). We should note that this is the case for $\gamma^{j} = 0.001$, that is, when the large traders' degree of risk aversion is not risk-averse. This result thus confirms that large traders are rather sensitive to environmental uncertainty \textit{even if the degree of their risk aversion is not high}.
\begin{remark}[] \label{Remark_4.2}
If $\mathfrak{Q}^{j} = 0$ and $\sigma^{\mathcal{I}}$ is large enough, large trader $j$ will not execute any orders since the large trader is risk-averse. Thus, the equilibrium execution strategy of large trader $i$ $(\neq j)$ essentially becomes an optimal execution strategy.
\end{remark}

The next example describes the situation when each large trader $i, j \in \{1, 2\}$ $(i \neq j)$ has \textit{opposite} initial holdings: $\mathfrak{Q}^{i} = 100,000$ and $\mathfrak{Q}^{j} = -100,000$. In this situation, how does $\sigma^{\mathcal{I}}$ influence the equilibrium execution strategy for each large trader?
\begin{figure}[tbp]
\centering
\begin{minipage}[b]{0.45\linewidth}
\centering
\includegraphics[height=5cm, width=7.5cm]{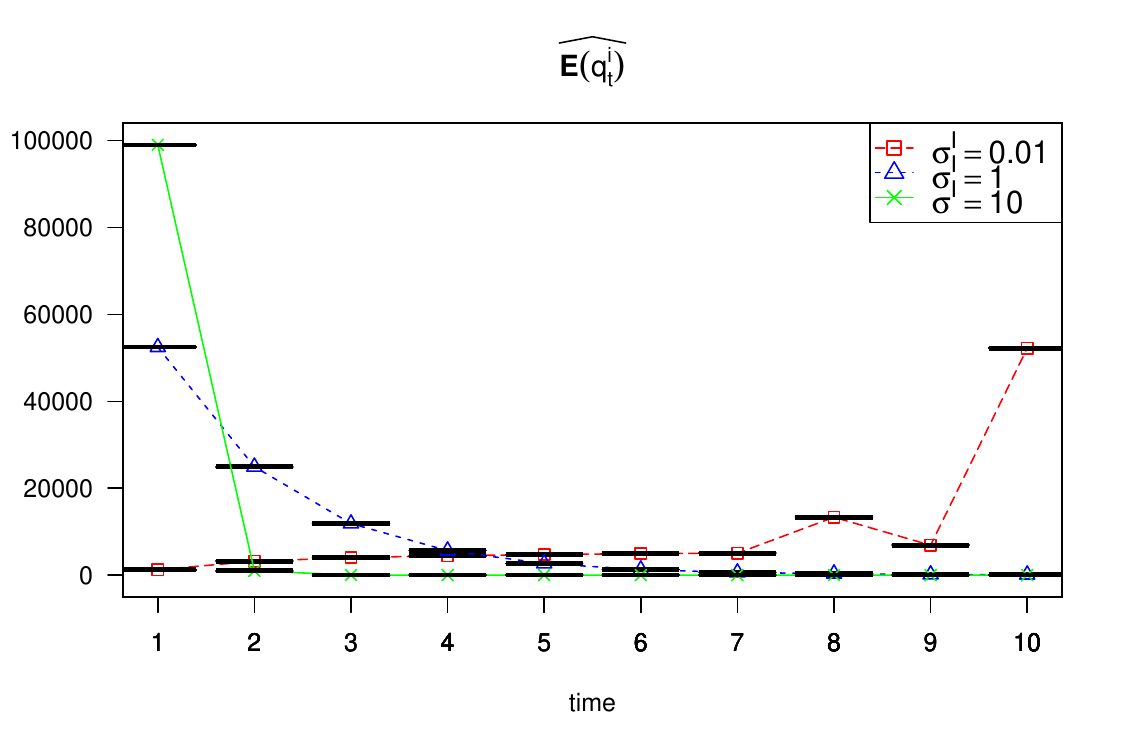}
\subcaption{Trader $i$ ($\mathfrak{Q}^{i} = 100,000$).}
\label{Figure_5a}
\end{minipage}
\begin{minipage}[b]{0.45\linewidth}
\centering
\includegraphics[height=5cm, width=7.5cm]{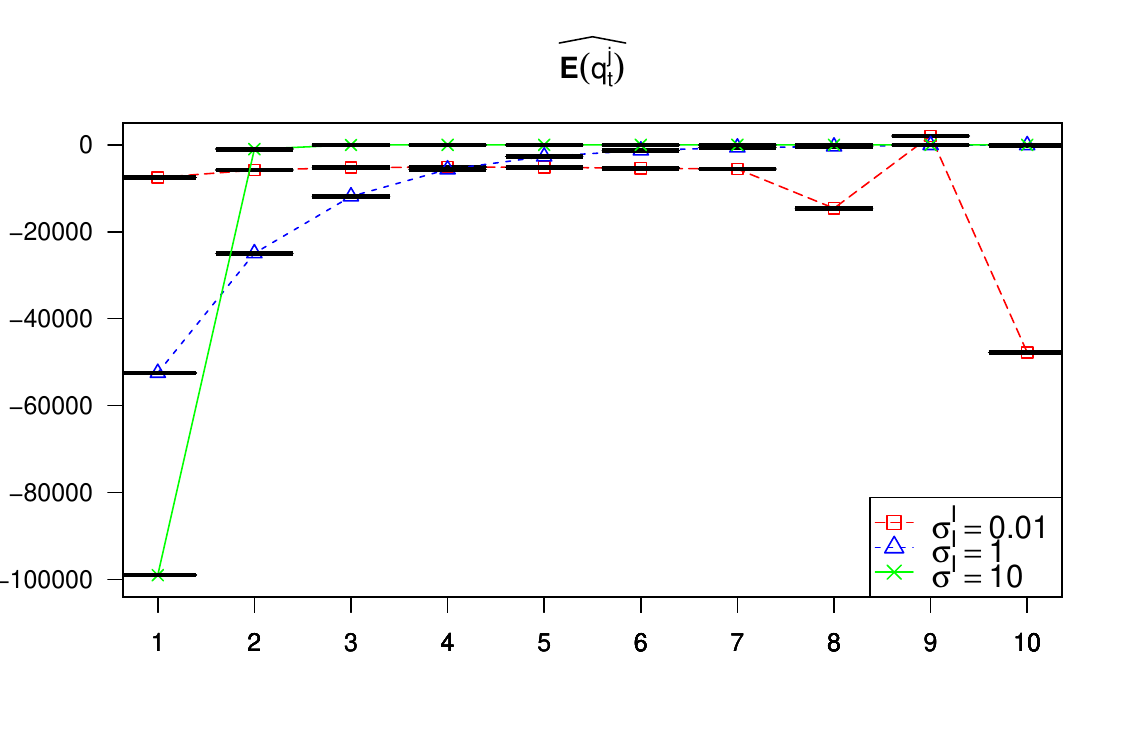}
\subcaption{Trader $j$ ($\mathfrak{Q}^{j} = -100,000$).}
\label{Figure_5b}
\end{minipage}
\caption{Equilibrium execution strategies with opposite initial inventories for $\sigma^{\mathcal{I}} = 0.01, 1$, and $10$ (with $a^{\mathcal{I}} = b^{\mathcal{I}} = 0$).}
\label{Figure_5}
\end{figure}
Figure~\ref{Figure_5} draws the equilibrium execution strategy for large traders $i$ and $j$ with the opposite initial holdings. When $\sigma^{\mathcal{I}}$ is low, both large traders inactively execute until maturity and unwind all the remaining positions at maturity. However, when $\sigma^{\mathcal{I}}$ is high, both large traders actively execute at the beginning and unwind almost all the positions in the first half of the trading window. The motivation for being \textit{second-mover} (or \textit{follower}) helps us reconcile the seemingly conflicting results: if large trader $i$ executes buy orders before large trader $j$, the opponent (large trader $j$) can sell orders with high prices, and vice versa.

\subsection{Special case: Random-walk} \label{Subsection_4.2}

When $b^{\mathcal{I}} = -1$, the Markovian environment is represented as a stochastic process with independent increments that follow a normal distribution.
\begin{definition} \label{Definition_4.1}
When $b^{\mathcal{I}} = -1$, the dynamics of the Markovian environment is rewritten as
\begin{align}
\Delta \mathcal{I}_{t + 1} \coloneqq \mathcal{I}_{t + 1} - \mathcal{I}_{t} = a^{\mathcal{I}} + \sigma^{\mathcal{I}} \omega_{t + 1}. \label{dy:RW}
\end{align}
In the case that $a^{\mathcal{I}} > 0$ ($a^{\mathcal{I}} < 0$, respectively), the stochastic process \eqref{dy:RW} refers to a \textit{positive-drifted} (\textit{negative-drifted}) \textit{random-walk}. If $a^{\mathcal{I}} = 0$, the process is called a \textit{symmetric random walk}.
\end{definition}
With the above definition in mind, this subsection analyzes how the drift term $a^{\mathcal{I}}$ influences the equilibrium execution strategy when the Markovian environment follows a stochastic process described by Eq.~\eqref{dy:RW}. This analysis enables us to illuminate the behavior of large traders under the existence of positive/negative (or no) trends in price dynamics. In particular, we focus on the following situation: large trader $i \in \{1, 2\}$, whom we call a \textit{buy-side} large trader, initially aims to acquire $100,000$ volumes of one financial asset, while large trader $j$ ($\neq i$), whom we call a \textit{sell-side} large trader, liquidates $100,000$ volumes (i.e., $\mathfrak{Q}^{i} = 100,000$ and $\mathfrak{Q}^{j} = -100,000$).
\begin{remark} \label{Remark_4.3}
Our definition of the buy/sell-side large trader is slightly different from the so-called ``buy/sell-side'' trader in practice. We regard the buy (sell, respectively)-side large trader as the one that intends or is obliged to buy (sell) a large amount of orders by the end of the trading window.
\end{remark}
\begin{figure}[tbp]
\centering
\begin{minipage}[b]{0.45\linewidth}
\centering
\includegraphics[height=5cm, width=7.5cm]{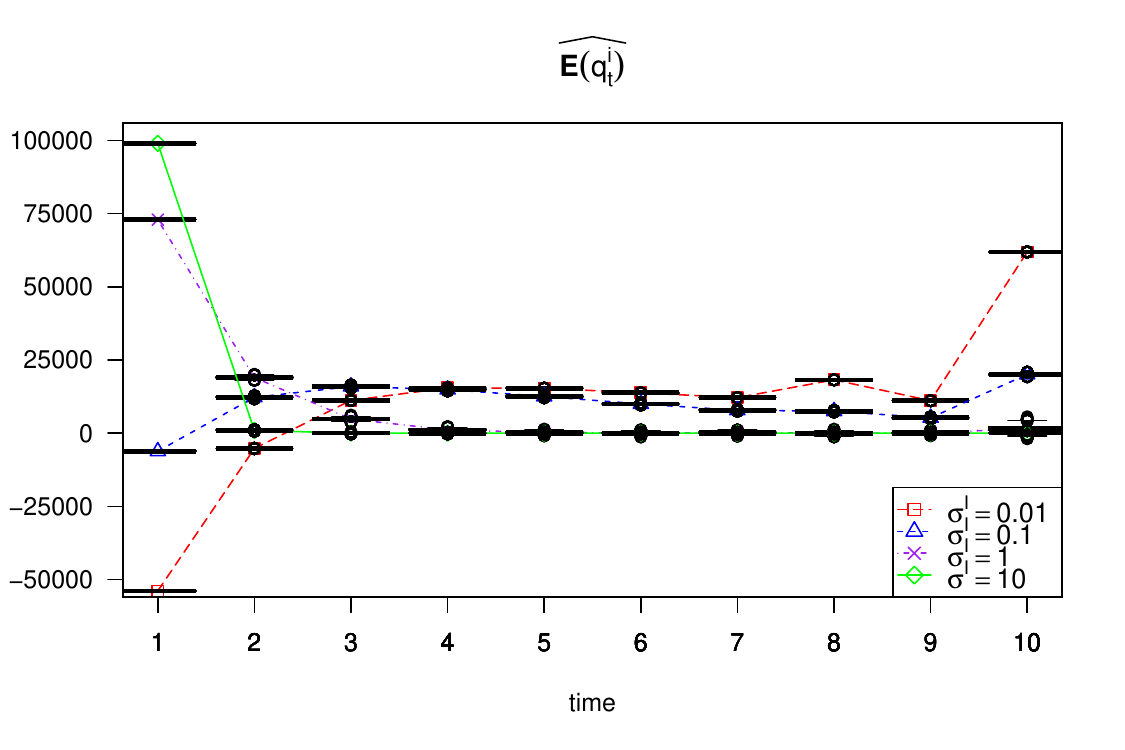}
\subcaption{Buy-side large trader $i$.}
\label{Figure_6a}
\end{minipage}
\begin{minipage}[b]{0.45\linewidth}
\centering
\includegraphics[height=5cm, width=7.5cm]{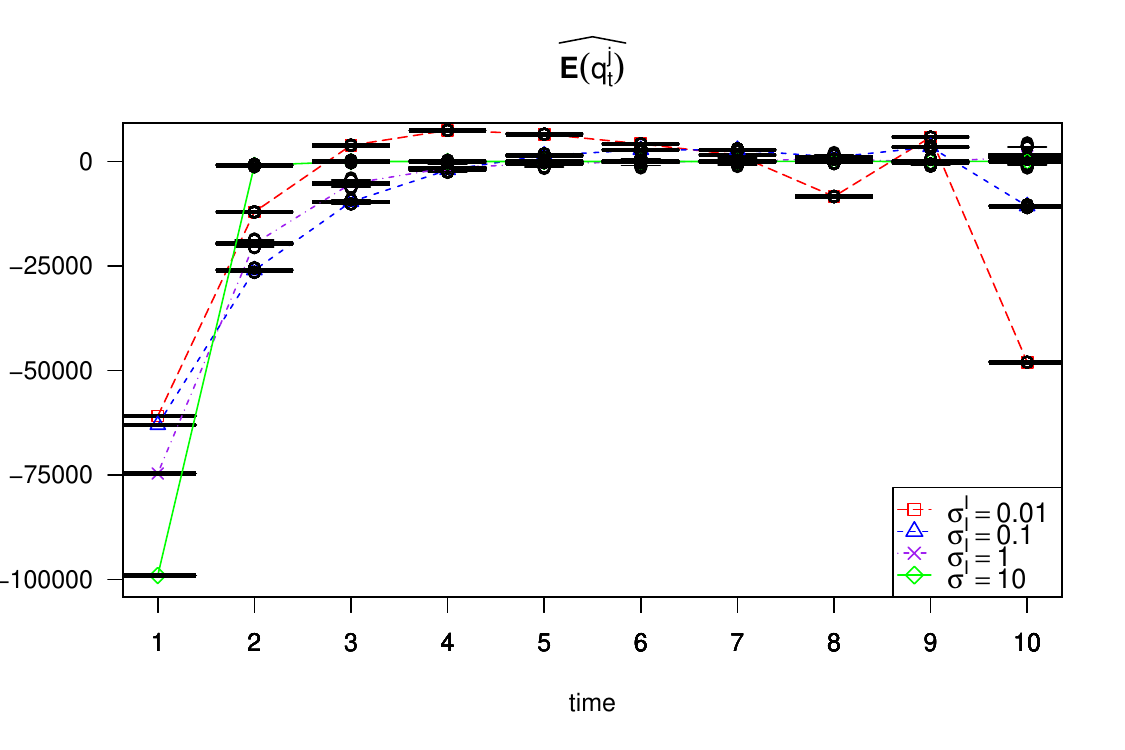}
\subcaption{Sell-side large trader $j$.}
\label{Figure_6b}
\end{minipage}
\caption{Effect of negative trend ($a^{\mathcal{I}} = - 0.5$) in price dynamics on the equilibrium execution strategy.}
\label{Figure_6}
\end{figure}
Figure~\ref{Figure_6} illustrates the equilibrium execution strategy of buy-side (sell-side, respectively) large trader $i$ ($j$) for the case that $a^{\mathcal{I}} = -0.5$ (i.e., a negative trend in price dynamics exists). Figure~\ref{Figure_6a} shows that when $\sigma^{\mathcal{I}}$ is $0.01$ (small), the buy-side large trader initially sells the financial asset and then gradually buys back the asset. The low value of $\sigma^{\mathcal{I}}$ indicates that the negative drift in the price dynamics seems to exist over the trading window. We can take this phenomenon as a key postulate that large traders want to buy (sell, respectively) a financial asset when the price is low (high). When $\sigma^{\mathcal{I}}$ is $1$ and $10$ (large), however, the buy-side large trader initially executes most of the orders since the large trader is risk-averse. In contrast, Figure~\ref{Figure_6b} infers that the strategy of the sell-side large trader $j$ to execute orders is seemingly counterintuitive. When $\sigma^{\mathcal{I}}$ is $0.01$ (small), the sell-side large trader liquidates the large part of positions at the initial and end of the trading window. The behavior of selling assets at maturity stems from strategic uncertainty, as demonstrated in Ohnishi and~Shimoshimizu~\cite{OMSM20QF}. On the contrary, the sell-side large trader sells orders at the beginning because he/she is willing to liquidate the holdings at a high price. Since there exists a negative trend in price dynamics, the sell-side large trader is keen to liquidate the holding at the beginning before the decrease of the financial asset price.

In addition, Figure~\ref{Figure_6} shows that the buy-side large trader's equilibrium execution strategy includes a \textit{round-trip trading} when $\sigma^{\mathcal{I}}$ is $0.01$ (small), while does not when $\sigma^{\mathcal{I}}$ is $1$ and $10$ (large). A round-trip trading is a trading strategy by which a buy-side (sell-side, respectively) trader repeatedly executes sell (buy) orders in a trading window. This result suggests that there exists a \textit{threshold} of $\sigma^{\mathcal{I}}$ at which the buy-side trader determines whether to depend on a round-trip trading or not. The implication is summarized as follows: when $\sigma^{\mathcal{I}}$ is low (high, respectively), the buy-side large trader is (un)able to expect some increase of his/her wealth from round-trip trading. The reason for this result is that low (high, respectively) $\sigma^{\mathcal{I}}$ implies an almost deterministic (not necessarily deterministic) negative trend in price dynamics.

\begin{figure}[tbp]
\centering
\begin{minipage}[b]{0.45\linewidth}
\centering
\includegraphics[height=5cm, width=7.5cm]{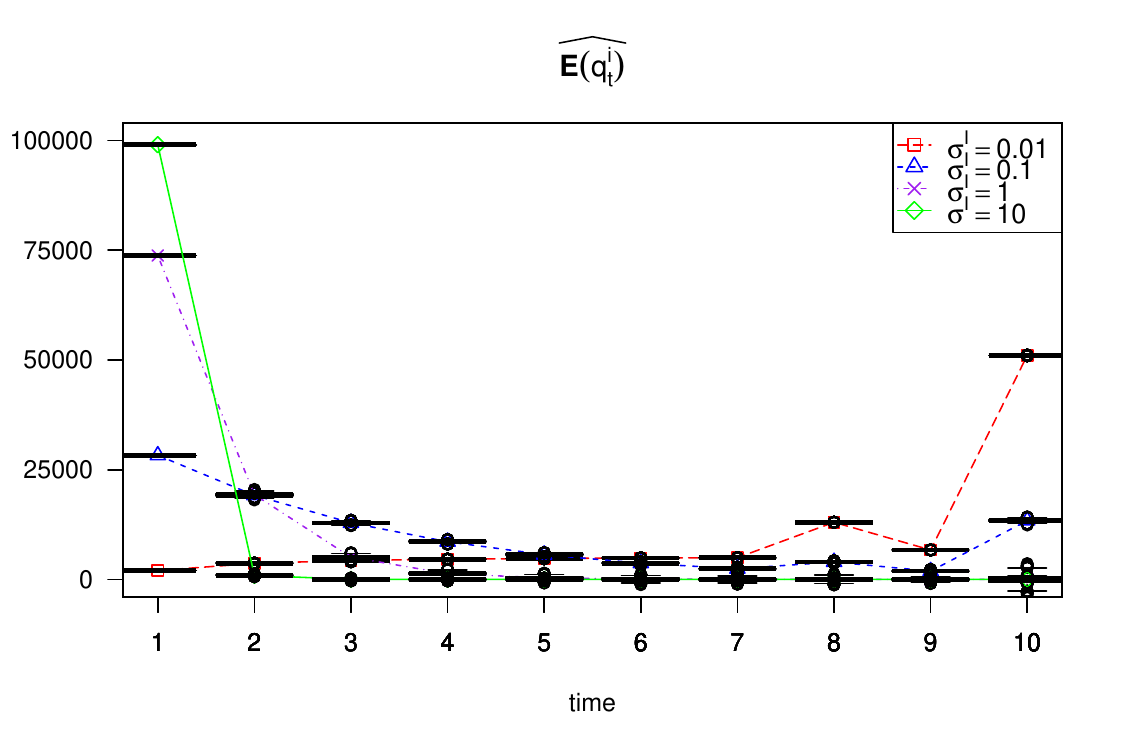}
\subcaption{Buy-side large trader $i$.}
\label{Figure_7a}
\end{minipage}
\begin{minipage}[b]{0.45\linewidth}
\centering
\includegraphics[height=5cm, width=7.5cm]{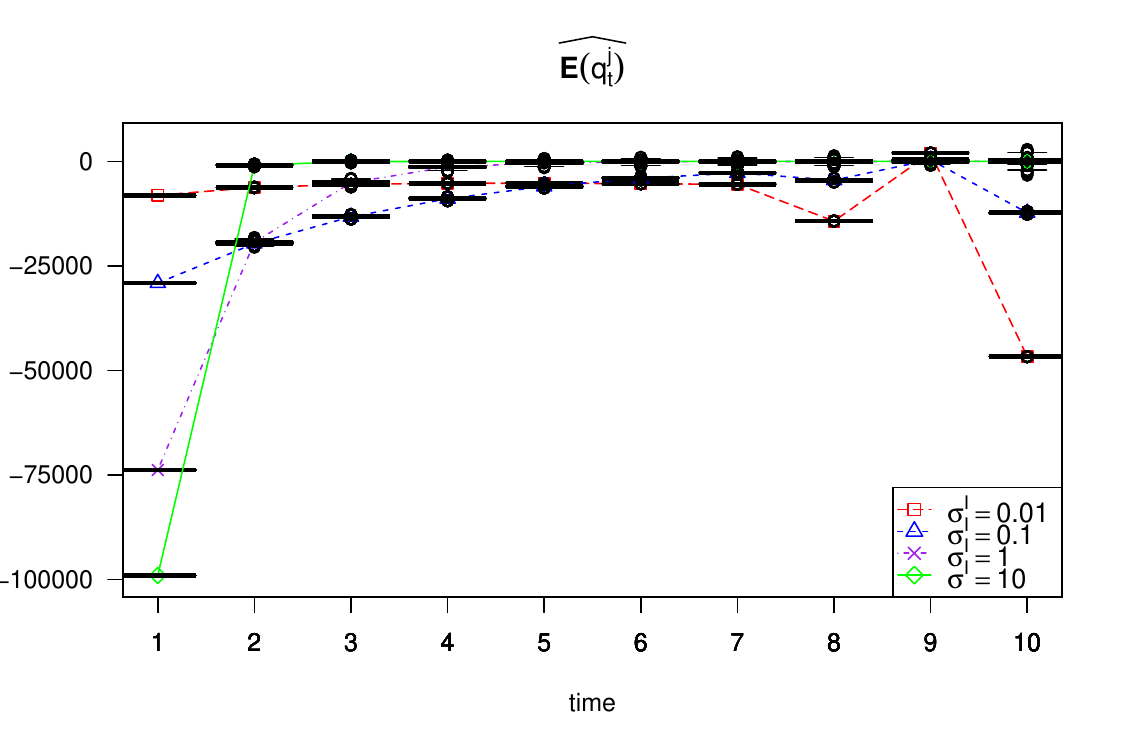}
\subcaption{Sell-side large trader $j$.}
\label{Figure_7b}
\end{minipage}
\caption{Effect of no trend ($a^{\mathcal{I}} = 0$) in price dynamics on equilibrium execution strategy.}
\label{Figure_7}
\end{figure}
Figure~\ref{Figure_7} shows the equilibrium execution strategy of buy-side (sell-side, respectively) large trader $i$ ($j$) for the case that $a^{\mathcal{I}} = 0$ (i.e., no-trends in price dynamics exist). The execution strategies are similar to those illustrated in Figure~\ref{Figure_5} since the difference only lies in the value of $b^{\mathcal{I}}$: $b^{\mathcal{I}} = 0$ in Figure~\ref{Figure_5} and $b^{\mathcal{I}} = -1$ in Figure~\ref{Figure_7}.

\begin{figure}[tbp]
\centering
\begin{minipage}[b]{0.45\linewidth}
\centering
\includegraphics[height=5cm, width=7.5cm]{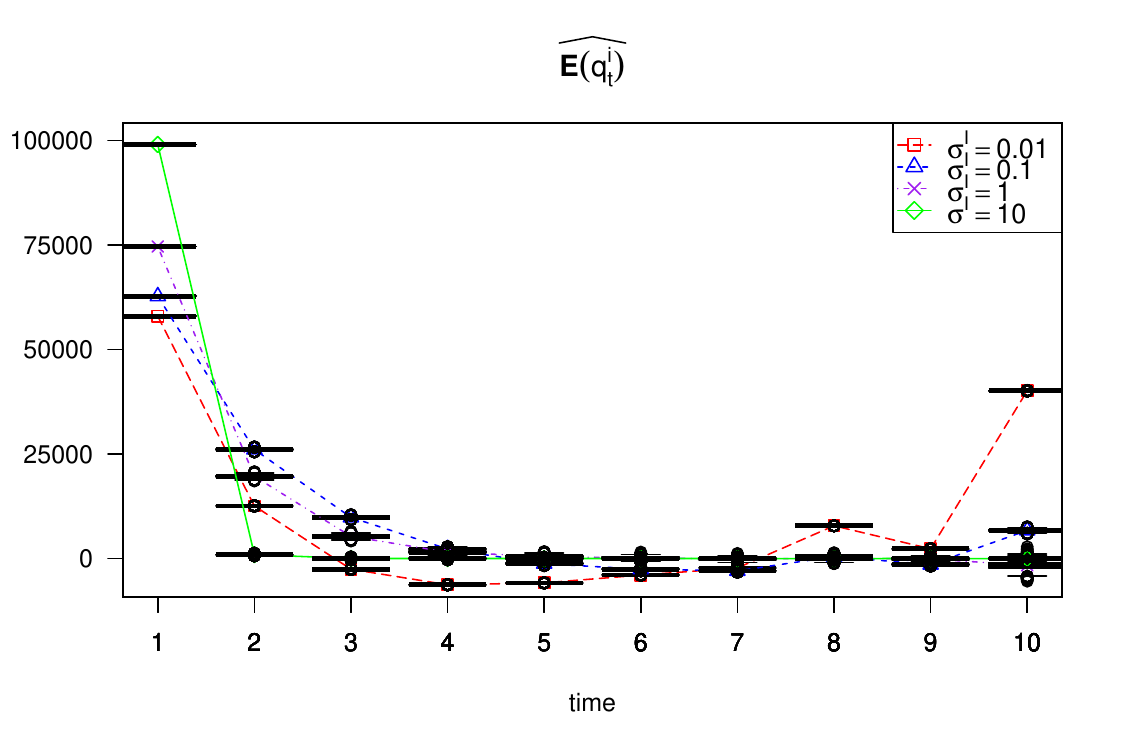}
\subcaption{Buy-side large trader $i$.}
\label{Figure_8a}
\end{minipage}
\begin{minipage}[b]{0.45\linewidth}
\centering
\includegraphics[height=5cm, width=7.5cm]{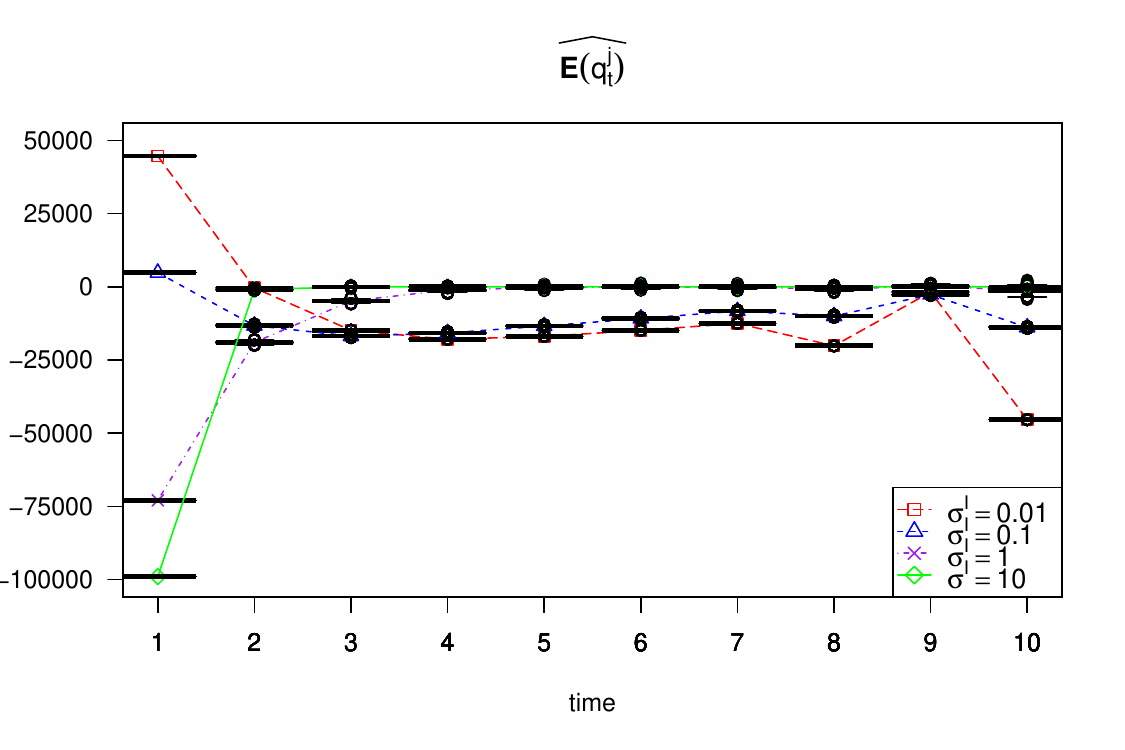}
\subcaption{Sell-side large trader $j$.}
\label{Figure_8b}
\end{minipage}
\caption{Effect of positive trend ($a^{\mathcal{I}} = 0.5$) in price dynamics on equilibrium execution strategy.}
\label{Figure_8}
\end{figure}
Figure~\ref{Figure_8} illustrates the equilibrium execution strategy of buy-side (sell-side, respectively) large trader $i$ ($j$) for the case that $a^{\mathcal{I}} = 0.5$ (i.e., a positive trend in price dynamics exists). As opposed to the case of the negative trend, when $\sigma^{\mathcal{I}}$ is $0.01$ (small), the sell-side large trader initially buys the financial asset and then gradually sells out the asset. In contrast, the strategy of the buy-side large trader $j$ to execute orders is again counterintuitive. When $\sigma^{\mathcal{I}}$ is $0.01$ (small), the buy-side large trader acquires the large part of positions at the initial and end of the trading window. In addition, the sell-side large trader's equilibrium execution strategy includes a \textit{round-trip trading} when $\sigma^{\mathcal{I}}$ is $0.01$ (small), but it does not include a round trip when $\sigma^{\mathcal{I}}$ is $1$ and $10$ (large). The logic for this result is similar to that explained in the case of the negative trend: low (high, respectively) $\sigma^{\mathcal{I}}$ implies an almost deterministic (not necessarily deterministic) positive trend in price dynamics.
\begin{remark}[] \label{Remark_4.4}
Figures~\ref{Figure_6}--\ref{Figure_8} show that the equilibrium execution strategy seems to be insensitive to the realization of the Markovian environment even if $\sigma^{\mathcal{I}}$ is not small. (That is, the equilibrium execution strategy becomes ``almost deterministic.'') The reason is that each large trader takes environmental uncertainty into account in advance and accelerates the execution.
\end{remark}
\begin{remark}[] \label{Remark_4.5}
Figures~\ref{Figure_6}--\ref{Figure_8} show that when there exists either a positive/negative trend, the equilibrium execution strategies of buy- and sell-side large traders become \textit{asymmetric}. However, the total volume traded by the large traders in the financial market, defined as
\begin{align}
TV_{t} \coloneqq \sum_{k = i, j} \lvert q^{k\ast}_{t} \rvert,
\end{align}
seems to form a U-shaped trading curve that is observed in intraday markets.
\end{remark}

\subsection{General case: Mean-reversion (diversion) process} \label{Subsection_4.3}

As a special case of the stochastic process $\{\mathcal{I}_{t}\}_{t \in \{1, \ldots, T\}}$, we can consider a mean-reversion process.

\begin{definition} \label{Definition_4.2}
When $b^{\mathcal{I}} \neq -1$, the stochastic process $\{\mathcal{I}_{t}\}_{t \in \{1, \ldots, T\}}$ can be rewritten as
\begin{align}
\Delta \mathcal{I}_{t + 1} \coloneqq \mathcal{I}_{t + 1} - \mathcal{I}_{t} = (1 + b^{\mathcal{I}}) \left( \frac{a^{\mathcal{I}}}{1 + b^{\mathcal{I}}} - \mathcal{I}_{t} \right) + \sigma^{\mathcal{I}} \omega_{t + 1}. \label{gcofMe}
\end{align}
For $b^{\mathcal{I}} \in (-1, 1)$, the term $(1 + b^{\mathcal{I}})$ in Eq.~\eqref{gcofMe} refers to the \textit{mean-reversion speed} and $a^{\mathcal{I}}/(1 + b^{\mathcal{I}})$ refers to the \textit{mean-reversion level}, and the stochastic process is called a \textit{mean-reversion process} (for large $T$).
\end{definition}
\noindent The graphical image of (expected) dynamics of $\{\mathcal{I}_{t}\}_{t \in \{1, \ldots, T\}}$ is illustrated in Figure~\ref{Figure_9}.
\begin{figure}[tbp]
\centering
\begin{tikzpicture}[yscale=0.7]
\draw[->,>=stealth,semithick] (-0.5,0)--(11.5,0) node[right]{\footnotesize{time}}; 
\draw[->,>=stealth,semithick] (0,-2.125)--(0, 2) node[left]{\footnotesize{$\mathbb{E} [\mathcal{I}_{t}]$}}; 
\draw (0,0) node[below left]{\footnotesize{$O$}}; 
\foreach \x/\thickness/\c/\linetype in {-0.5/semithick/black/solid, 1/thick/blue/dashed, 1.2/very thick/purple/densely dash dot}{
\foreach \i in {0,...,8}{
\coordinate [] (a\i) at (1.2*\i, {-0.5*1.75/2.2 + 1.75*(-\x)^\i/4.4});
};
\foreach \i in {0,...,7}{
\tikzmath
{
int \j;
\j = \i + 1.5;
}
\draw[-{Stealth},\thickness,\c,\linetype] (a\i) to[] (a\j);
};
};
\foreach \i in {1,...,8}{
\draw (1.2*\i,0) node[above]{\small{$\i$}};
\draw[-,semithick] (1.2*\i,0)--(1.2*\i, 0.1);
};
\foreach \x/\p/\c in {-0.5/below right/black, 1/above right/blue, 1.2/right/purple}{
\draw (1.2*8, {-0.5*2/2.2 + 2*(-\x)^8/4.4}) node[\p,\c]{\footnotesize{$b^{\mathcal{I}} = \x$}};
};
\end{tikzpicture}
\caption{Dynamics of $\mathbb{E} [\mathcal{I}_{t}]$ for different $b^{\mathcal{I}}$ ($= -0.5, 1$, and $1.2$) when $a^{\mathcal{I}} = -0.5$. Markovian environment follows a random walk on average for $b^{\mathcal{I}} = 1$, while diverges on average for $b^{\mathcal{I}} = 1.2$ ($> 1$).
}
\label{Figure_9}
\end{figure}
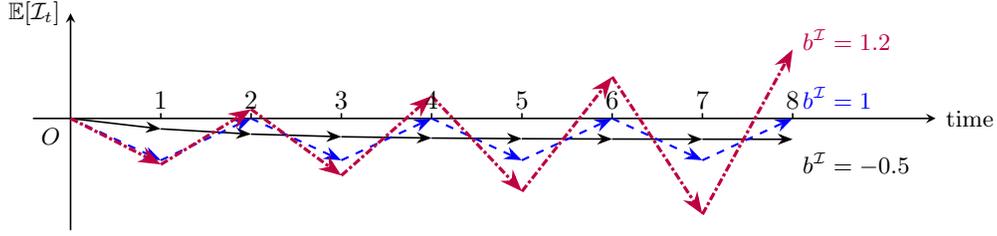

In what follows, we analyze how $a^{\mathcal{I}}$ and $b^{\mathcal{I}}$ influence the equilibrium execution strategy. This analysis allows us to shed light on how mean-reversion speed and mean-reversion level impact the equilibrium execution strategy. We also investigate the case that the Markovian environment \textit{does not} follow a mean-reversion process (i.e., $b^{\mathcal{I}} \notin (-1, 1)$). As in Section~\ref{Subsection_4.2}, our focus in this subsection lies on the \textit{opposite initial positions}: $\mathfrak{Q}^{i} = - \mathfrak{Q}^{j} = 100,000$.

\subsubsection{Effect of $a^{\mathcal{I}}$ and $b^{\mathcal{I}}$}
\label{Subsubsection_4.3.1}

We first investigate how $b^{\mathcal{I}}$ influences the equilibrium execution strategy. To this end, we fix $a^{\mathcal{I}}$ for negative, zero, and positive values. 

\begin{figure}[tbp]
\centering
\begin{minipage}[b]{0.45\linewidth}
 \centering
 \includegraphics[height=5cm, width=7.5cm]{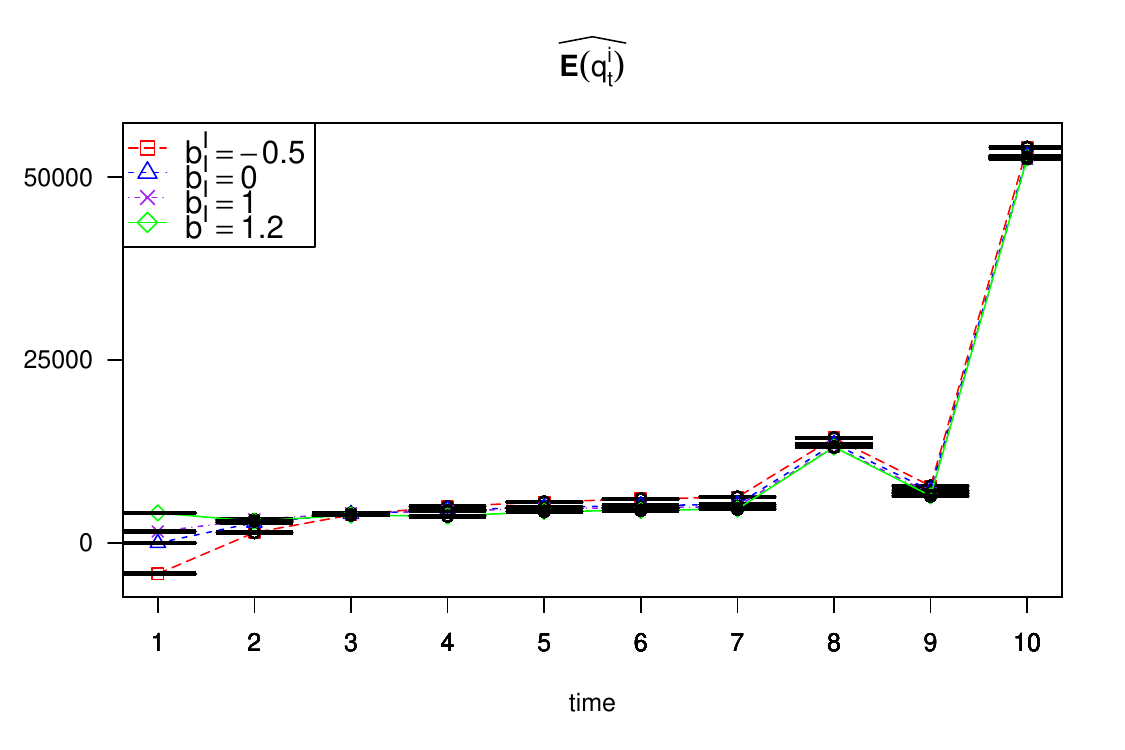}
 \subcaption{Buy-side large trader $i$.}
 \label{Figure_10a}
  \end{minipage}
\begin{minipage}[b]{0.45\linewidth}
    \centering
 \includegraphics[height=5cm, width=7.5cm]{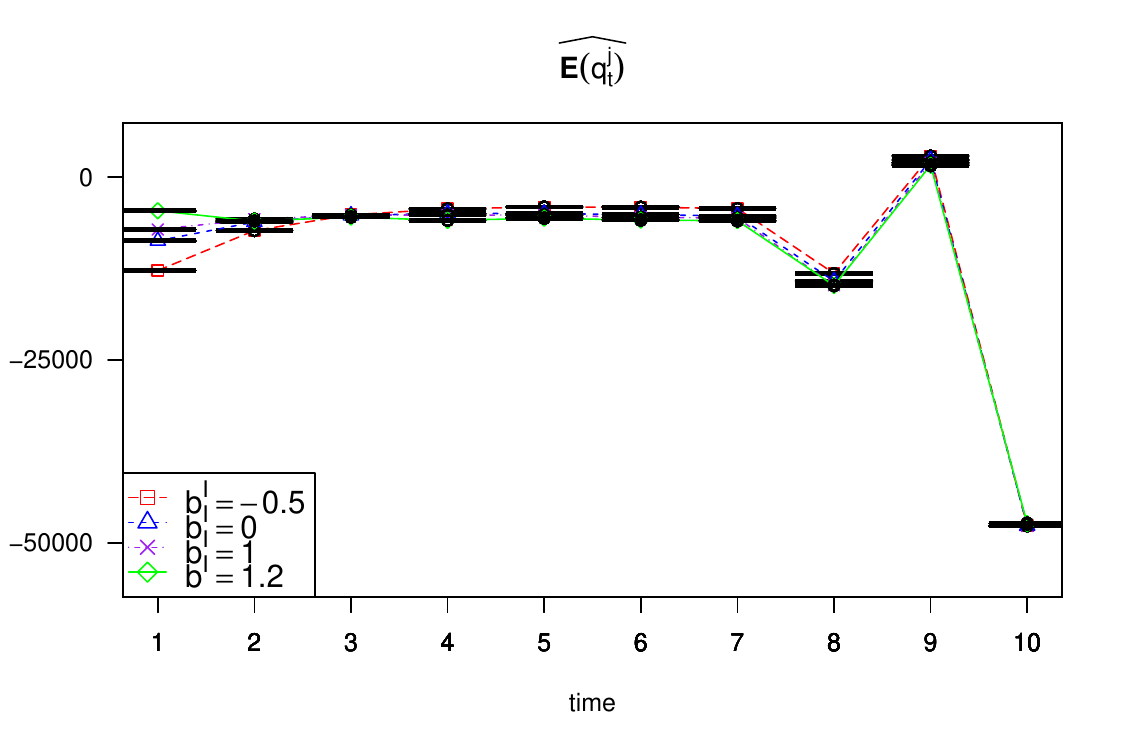}
    \subcaption{Sell-side large trader $j$.}
    \label{Figure_10b}
  \end{minipage}
\caption{Equilibrium execution strategy for $b^{\mathcal{I}} = -0.5, 0, 1, 1.2$ (with $a^{\mathcal{I}} = -0.5$ and $\sigma^{\mathcal{I}} = 0.01$ fixed).}
\label{Figure_10}
\end{figure}
Figure~\ref{Figure_10} depicts the equilibrium execution strategy for $b^{\mathcal{I}} = -0.5, 0, 1$, and $1.2$ when $a^{\mathcal{I}} = -0.5$ (negative). For negative $b^{\mathcal{I}}$ (i.e., $b^{\mathcal{I}} = -0.5$), the buy-side large trader initially executes sell orders. In the setting that $a^{\mathcal{I}} = -0.5$ and $b^{\mathcal{I}} = -0.5$, the Markovian environment gradually decreases to the mean-reversion level. The buy-side (as well as the sell-side) large trader then foresees the negative effect on the price impact, initially selling assets and gradually executing buy orders over the course of the trading window.

\begin{figure}[tbp]
\centering
\begin{minipage}[b]{0.45\linewidth}
 \centering
 \includegraphics[height=5cm, width=7.5cm]{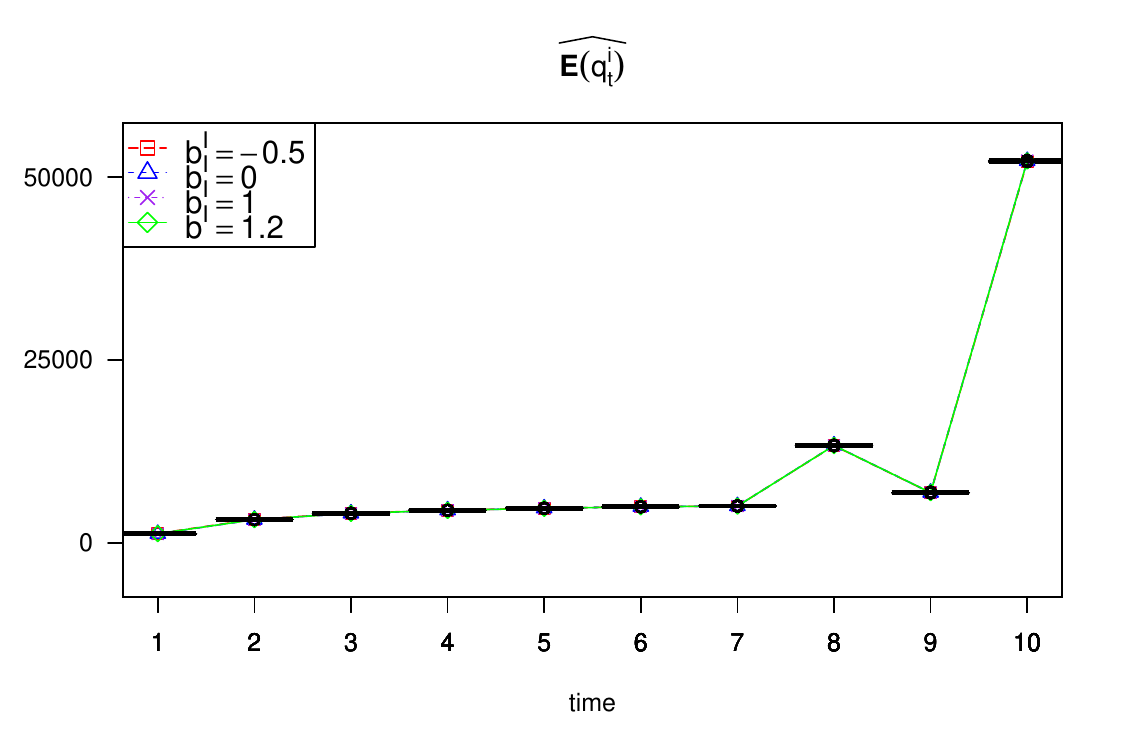}
 \subcaption{Buy-side large trader $i$.}
 \label{Figure_11a}
  \end{minipage}
\begin{minipage}[b]{0.45\linewidth}
    \centering
 \includegraphics[height=5cm, width=7.5cm]{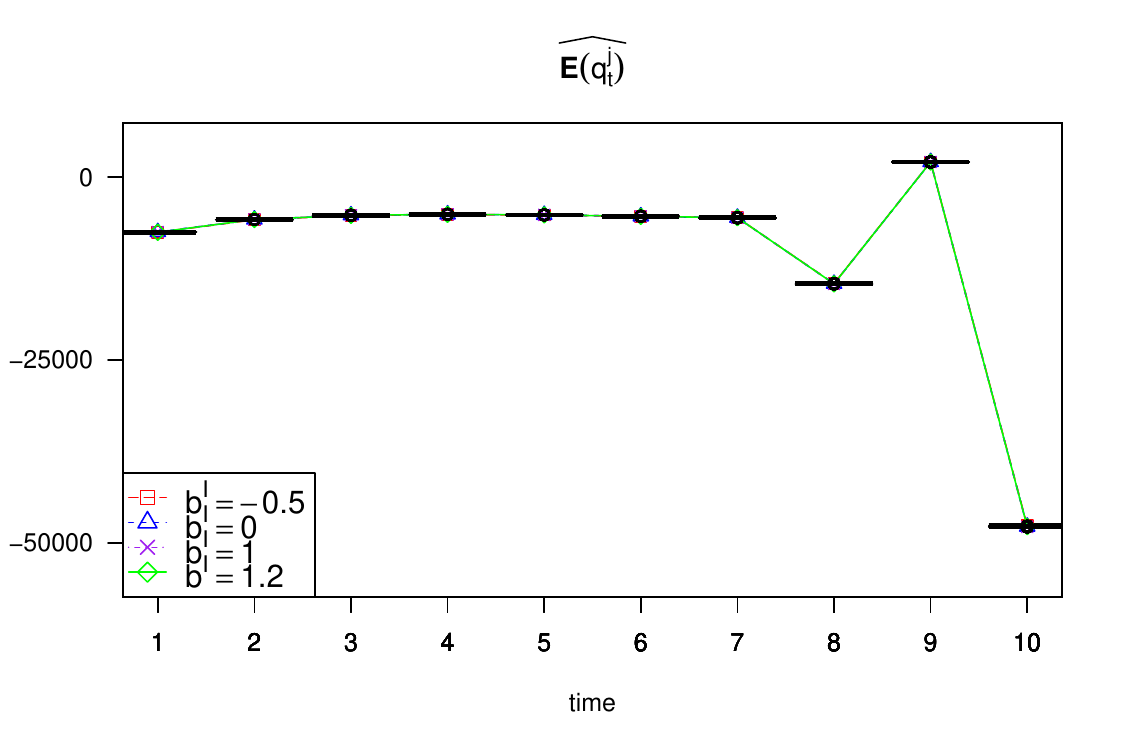}
    \subcaption{Sell-side large trader $j$.}
    \label{Figure_11b}
  \end{minipage}
\caption{Equilibrium execution strategy for $b^{\mathcal{I}} = -0.5, 0, 1, 1.2$ (with $a^{\mathcal{I}} = 0$ and $\sigma^{\mathcal{I}} = 0.01$ fixed).}
\label{Figure_11}
\end{figure}
Figure~\ref{Figure_11} illustrates the equilibrium execution strategy for $b^{\mathcal{I}} = -0.5, 0, 1$, and $1.2$ when $a^{\mathcal{I}} = 0$. The figure shows that the difference in $b^{\mathcal{I}}$ does not influence the equilibrium execution strategy when the Markovian environment has no drift effect. For $a^{\mathcal{I}} = 0$, the mean function of the Markovian environment becomes
\begin{align}
\mathbb{E} \left[ \mathcal{I}_{t + 1} \right] 
= a^{\mathcal{I}} + b^{\mathcal{I}} \mathbb{E} \left[ \mathcal{I}_{t} + \sigma^{\mathcal{I}} \omega_{t} \right]
= b^{\mathcal{I}} \mathbb{E} \left[ \mathcal{I}_{t} \right].
\end{align}
Combining the above relationship with $\mathcal{I}_{0} = 0$, we have $\mathbb{E} \left[ \mathcal{I}_{t} \right] = 0$ for all $t \in \{1, \ldots, T\}$. Therefore, there are no (expected) positive/negative trends in price dynamics, resulting in little changes in the equilibrium execution strategy for different values of $b^{\mathcal{I}}$.

\begin{figure}[tbp]
\centering
\begin{minipage}[b]{0.45\linewidth}
 \centering
 \includegraphics[height=5cm, width=7.5cm]{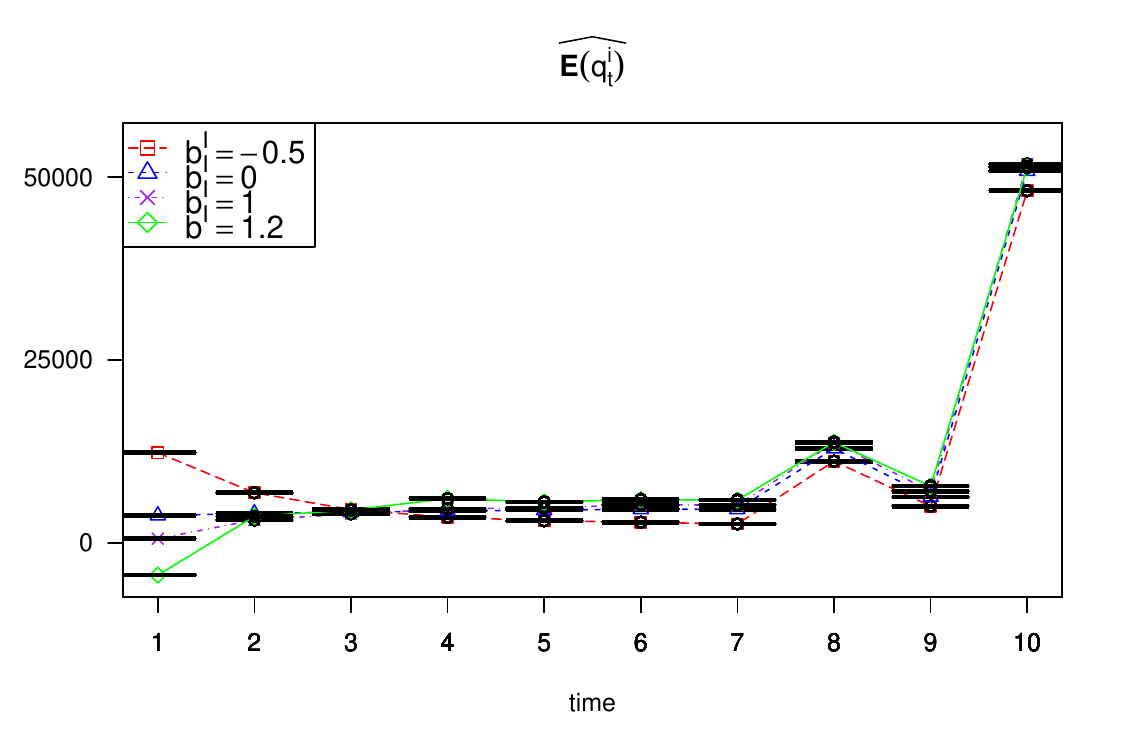}
 \subcaption{Buy-side large trader $i$.}
 \label{Figure_12a}
  \end{minipage}
\begin{minipage}[b]{0.45\linewidth}
    \centering
 \includegraphics[height=5cm, width=7.5cm]{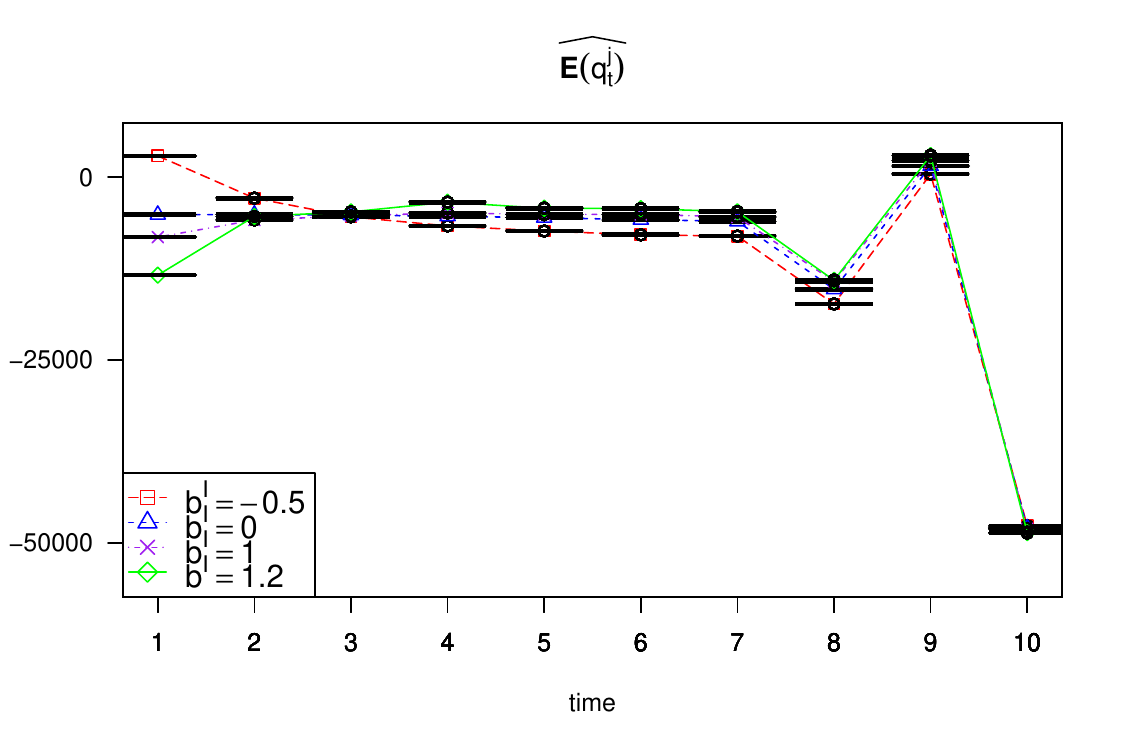}
    \subcaption{Sell-side large trader $j$.}
    \label{Figure_12b}
  \end{minipage}
\caption{Equilibrium execution strategy for $b^{\mathcal{I}} = -0.5, 0, 1, 1.2$ (with $a^{\mathcal{I}} = 1$ and $\sigma^{\mathcal{I}} = 0.01$ fixed).}
\label{Figure_12}
\end{figure}
Figure~\ref{Figure_12} shows the equilibrium execution strategy for $b^{\mathcal{I}} = -0.5, 0, 1$, and $1.2$ when $a^{\mathcal{I}} = 1$ (positive). Since the Markovian environment influences the price dynamics gradually and positively, the sell-side large trader initially executes the buy orders and sells out the assets over the course of the trading window.
\begin{remark}[Oscillation and divergence of Markovian environment] \label{Remark_4.6}
As shown in Figure~\ref{Figure_9}, when $\lvert b^{\mathcal{I}} \rvert > 1$, the Markovian environment occilates and diverges: the stochastic process $\{ \mathcal{I}_{t} \}_{t \in \{1, \ldots, T\}}$ becomes unstationary. The intuitive implication for why the equilibrium execution strategy differs depending on the values of $b^{\mathcal{I}}$ seems unclear from Figures~\ref{Figure_10}--\ref{Figure_12}. In this situation, the equilibrium execution strategy does not exhibit any oscillation until $t = 7$, while it oscillates as maturity approaches. The phenomenon of oscillation is also found in Schied and Zhang~\cite{SAZT19}. 
\end{remark}

We next examine the effect of $a^{\mathcal{I}}$ on the equilibrium execution strategy for negative, zero, and positive values of $b^{\mathcal{I}}$.

\begin{figure}[tbp]
\centering
\begin{minipage}[b]{0.45\linewidth}
 \centering
 \includegraphics[height=5cm, width=7.5cm]{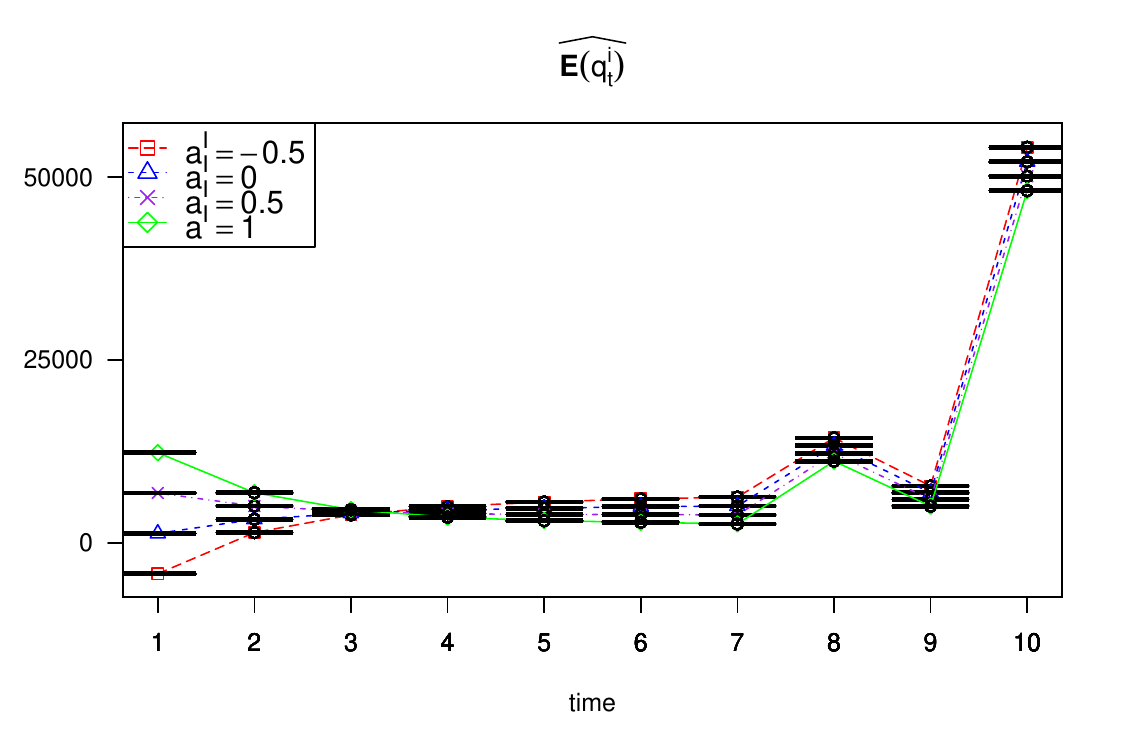}
 \subcaption{Buy-side large trader $i$.}
 \label{Figure_13a}
  \end{minipage}
\begin{minipage}[b]{0.45\linewidth}
    \centering
 \includegraphics[height=5cm, width=7.5cm]{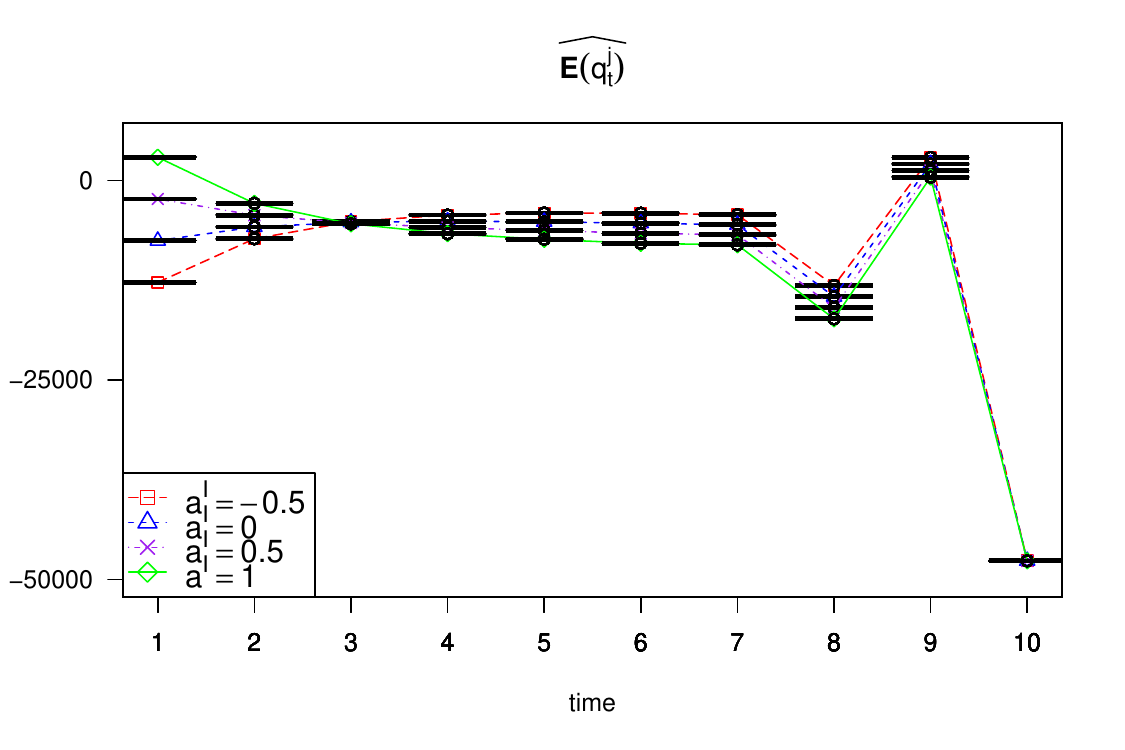}
    \subcaption{Sell-side large trader $j$.}
    \label{Figure_13b}
  \end{minipage}
\caption{Equilibrium execution strategy for $a^{\mathcal{I}} = -0.5, 0, 0.5, 1$ (with $b^{\mathcal{I}} = -0.5$ and $\sigma^{\mathcal{I}} = 0.01$ fixed).}
\label{Figure_13}
\end{figure}
\begin{figure}[tbp]
\centering
\begin{minipage}[b]{0.45\linewidth}
 \centering
 \includegraphics[height=5cm, width=7.5cm]{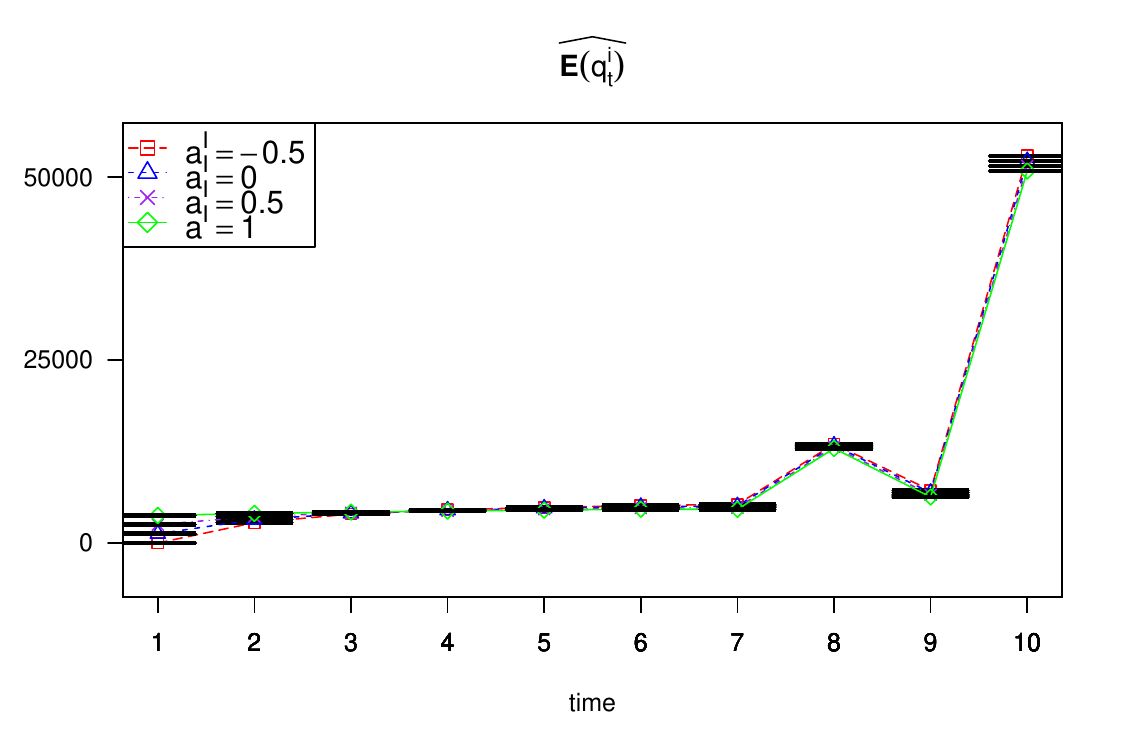}
 \subcaption{Buy-side large trader $i$.}
 \label{Figure_14a}
  \end{minipage}
\begin{minipage}[b]{0.45\linewidth}
    \centering
 \includegraphics[height=5cm, width=7.5cm]{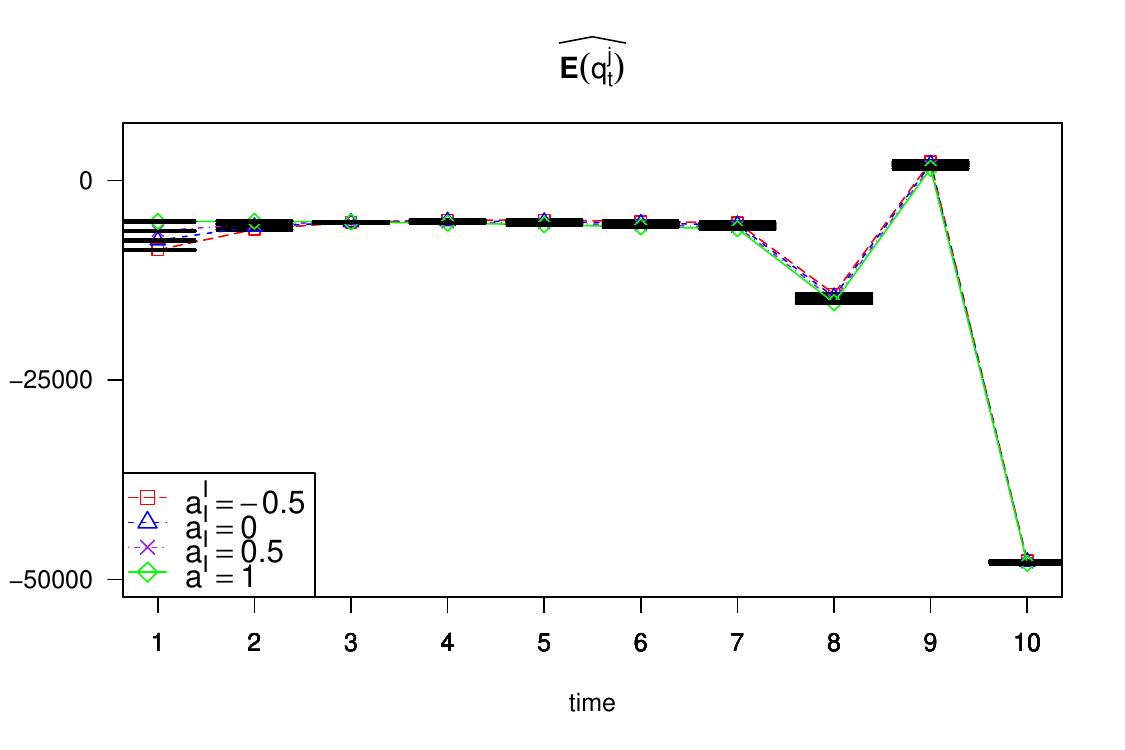}
    \subcaption{Sell-side large trader $j$.}
    \label{Figure_14b}
  \end{minipage}
\caption{Equilibrium execution strategy for $a^{\mathcal{I}} = -0.5, 0, 0.5, 1$ (with $b^{\mathcal{I}} = 0$ and $\sigma^{\mathcal{I}} = 0.01$ fixed).}
\label{Figure_14}
\end{figure}
\begin{figure}[tbp]
\centering
\begin{minipage}[b]{0.45\linewidth}
 \centering
 \includegraphics[height=5cm, width=7.5cm]{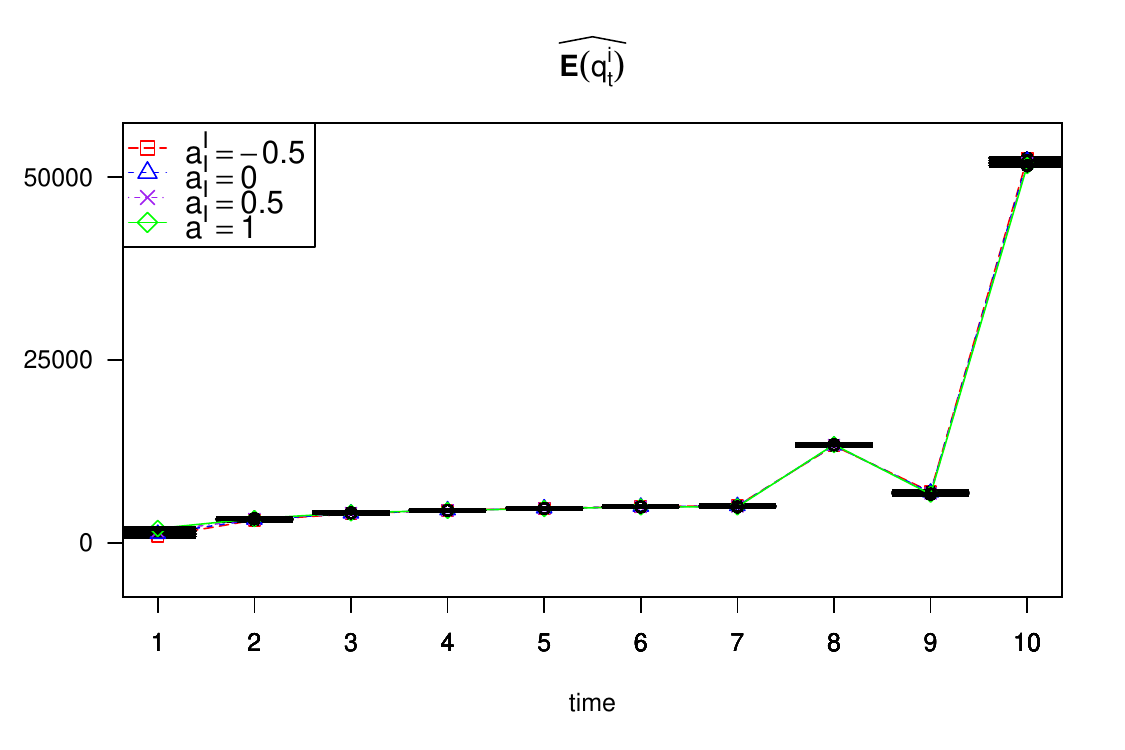}
 \subcaption{Buy-side large trader $i$.}
 \label{Figure_15a}
  \end{minipage}
\begin{minipage}[b]{0.45\linewidth}
    \centering
 \includegraphics[height=5cm, width=7.5cm]{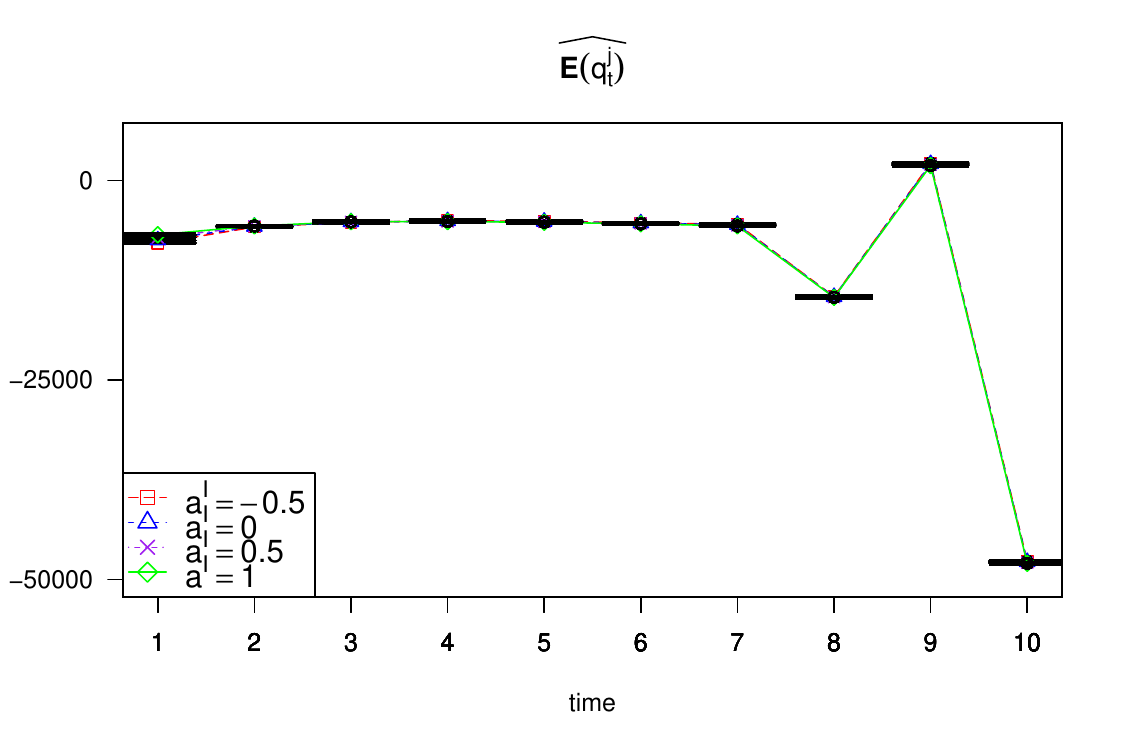}
    \subcaption{Sell-side large trader $j$.}
    \label{Figure_15b}
  \end{minipage}
\caption{Equilibrium execution strategy for different $a^{\mathcal{I}}$ (with $b^{\mathcal{I}} = 0.5$ and $\sigma^{\mathcal{I}} = 0.01$ fixed).}
\label{Figure_15}
\end{figure}
Figures~\ref{Figure_13}--\ref{Figure_15} illustrate the equilibrium execution strategy of large traders $i$ and $j$ ($i \neq j$) with the \textit{opposite} initial holdings ($\mathfrak{Q}^{i} = - \mathfrak{Q}^{i} = 100,000$) for different $a^{\mathcal{I}}$. As $b^{\mathcal{I}}$ increases, the effect of $a^{\mathcal{I}}$ on the equilibrium execution strategy decreases. Larger $b^{\mathcal{I}} \in (-1, 1)$ indicates faster mean-reversion speed and thus less change of $\mathcal{I}_{t}$, which in turn makes the price dynamics stable. Hence, the effect of the Markovian environment becomes little when $b^{\mathcal{I}} \in (-1, 1)$ is large.

\subsubsection{Effect of Mean-reversion speed} \label{Subsubsection_4.3.2}

We move on to the discussion on how the mean reversion speed (i.e., $1 + b^{\mathcal{I}}$) affects the equilibrium execution strategy with the mean reversion level (i.e., $a^{\mathcal{I}}/(1 + b^{\mathcal{I}})$) fixed.
\begin{figure}[tbp]
\centering
\begin{minipage}[b]{0.45\linewidth}
 \centering
 \includegraphics[height=5cm, width=7.5cm]{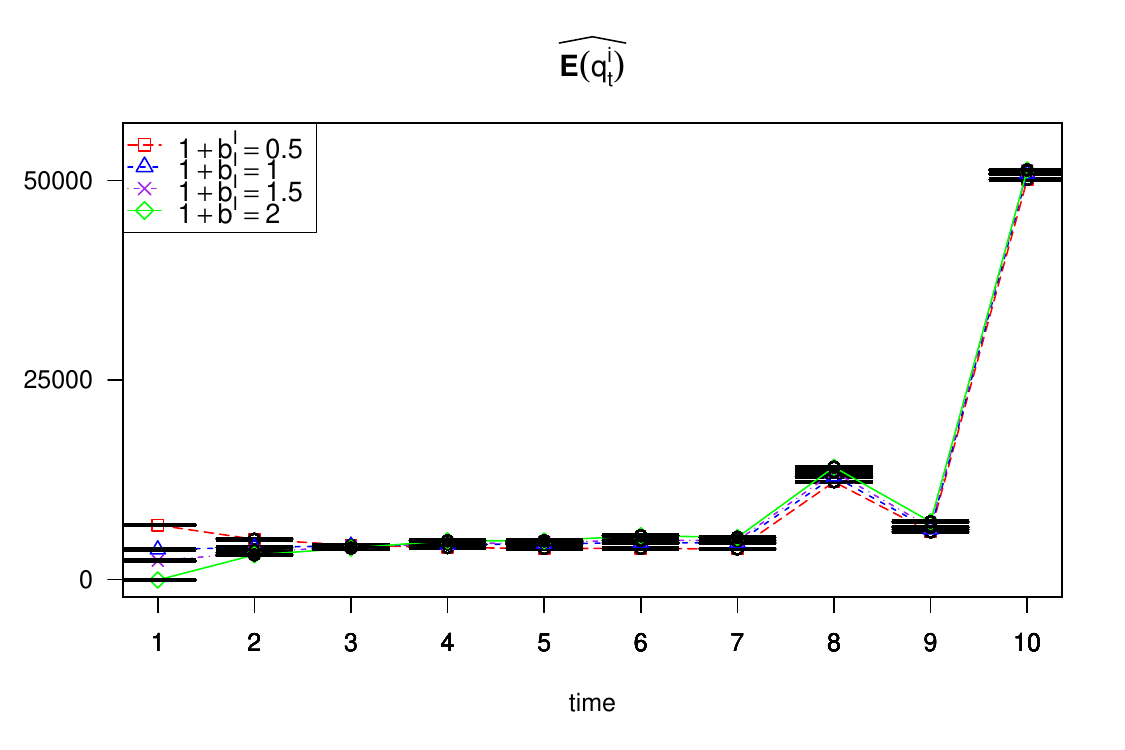}
 \subcaption{Buy-side large trader $i$.}
 \label{Figure_16a}
  \end{minipage}
\begin{minipage}[b]{0.45\linewidth}
    \centering
 \includegraphics[height=5cm, width=7.5cm]{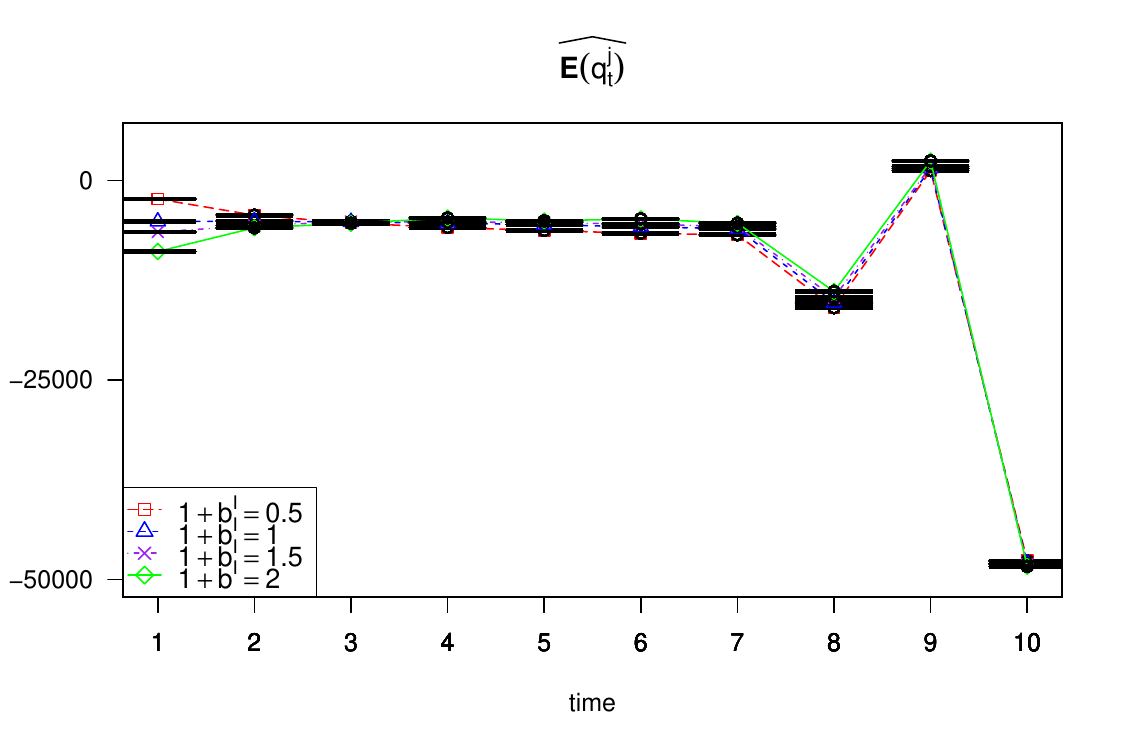}
    \subcaption{Sell-side large trader $j$.}
    \label{Figure_16b}
  \end{minipage}
\caption{Equilibrium execution strategy for ($a^{\mathcal{I}} =$) $1 + b^{\mathcal{I}} = 0.5, 1, 1.5,$ and $2$.}
\label{Figure_16}
\end{figure}
Figure~\ref{Figure_16} illustrates that lower $b^{\mathcal{I}}$ results in a few faster execution at the beginning, although the effect of mean reversion speed seems less unclear than in the cases that we have examined in Section~\ref{Section_4} so far.

\section{Equilibrium execution strategy with target close order}
\label{Section_5}

We finally analyze an execution game model with the closing price $P_{T + 1}$. The time framework is the same as the model in Section~\ref{Section_3}. However, we add an assumption that each large trader can execute his/her remaining execution volume at time $T + 1$ with the closing price. We further assume that the execution at time $T + 1$ imposes each large trader to pay the additive cost $\chi_{T + 1} \in \mathbb{R}_{++}$ per unit. Similar to Theorem~\ref{Theorem_3.1}, we have the following theorem. 

\begin{theorem}
\label{Theorem_5.1}
There exists a Markov perfect equilibrium at which the following properties hold for each large trader $i \in \{1, 2\}$:
\begin{enumerate} 
\item The execution volume at the Markov perfect equilibrium for large trader $i$ at time $t \in \{1, \ldots, T, T + 1\}$, denoted as $\tilde{q}^{i\ast}_{t}$, becomes an \textit{affine function} of the previous Markovian environment, the remaining execution volume of each large trader, and the cumulative residual effect of the past price impact. That is, 
\begin{align}
\tilde{q}^{i\ast}_{t}
= \tilde{a}^{i}_{t} + \tilde{b}^{i}_{t} \overline{Q}^{i}_{t} + \tilde{c}^{i\ast}_{t} \overline{Q}^{j}_{t} + \tilde{d}^{i\ast}_{t} R_{t} + \tilde{e}^{i\ast}_{t} \mathcal{I}_{t - 1}, \quad t = 1, \ldots, T, T + 1,
\end{align}
where $\tilde{a}^{i}_{t}, \tilde{b}^{i}_{t}, \tilde{c}^{i}_{t}, \tilde{d}^{i}_{t}, \tilde{e}^{i}_{t}$ for $t \in \{1, \ldots, T, T + 1\}$ are all deterministic functions of time $t$ which are dependent on the problem parameters and can be computed backwardly in time $t$.
\item The value function $V^{i}_{t} (\pi^{1}_{t}, \pi^{2}_{t}) \big[ \bm{s}_{t} \big]$ for each large trader $i \in \{ 1, 2 \}$ at time $t \in \{ 1, \ldots, T, T + 1 \}$ is represented as a functional form as follows: 
\begin{align}
&V^{i}_{t} (\pi^{1}_{t}, \pi^{2}_{t}) \big[ W^{1}_{t}, W^{2}_{t}, P_{t}, \overline{Q}^{1}_{t}, \overline{Q}^{2}_{t}, R_{t}, \mathcal{I}_{t - 1} \big] \nonumber\\
&= - \exp \Big\{ - \gamma^{i} \Big[ W^{i}_{t} - P_{t} \overline{Q}^{i}_{t} + \tilde{G}^{1i\ast}_{t} \overline{Q}^{i}_{t} + \tilde{G}^{2i\ast}_{t} (\overline{Q}^{i}_{t})^{2} + \tilde{H}^{1i\ast}_{t} \overline{Q}^{i}_{t} R_{t} + \tilde{H}^{2i\ast}_{t} R^{2}_{t} + \tilde{H}^{3i\ast}_{t} R_{t} \nonumber\\
&\quad + \tilde{J}^{1i\ast}_{t} \overline{Q}^{i}_{t} \overline{Q}^{j}_{t} + \tilde{J}^{2i\ast}_{t} \overline{Q}^{j}_{t} R_{t} + \tilde{J}^{3i}_{T} (\overline{Q}^{j}_{t})^{2} + \tilde{J}^{4i\ast}_{t} \overline{Q}^{j}_{t} \nonumber\\
&\quad + \tilde{L}^{1i\ast}_{t} \overline{Q}^{i}_{t} \mathcal{I}_{t - 1} + \tilde{L}^{2i\ast}_{t} R_{t} \mathcal{I}_{t - 1} + \tilde{L}^{3i\ast}_{t} \overline{Q}^{j}_{t} \mathcal{I}_{t - 1} + \tilde{L}^{4i\ast}_{t} \mathcal{I}_{t - 1}^{2} + \tilde{L}^{5i\ast}_{t} \mathcal{I}_{t - 1} + \tilde{Z}^{i\ast}_{t} \Big] \Big\}, 
\end{align} 
where $\tilde{G}^{1i\ast}_{t}, \tilde{G}^{2i\ast}_{t}, \tilde{H}^{1i\ast}_{t}, \tilde{H}^{2i\ast}_{t}, \tilde{H}^{3i\ast}_{t}, \tilde{J}^{1i\ast}_{t}, \tilde{J}^{2i\ast}_{t}, \tilde{J}^{3i\ast}_{t}, \tilde{J}^{4i\ast}_{t}, \tilde{L}^{1i\ast}_{t}, \tilde{L}^{2i\ast}_{t}, \tilde{L}^{3i\ast}_{t}, \tilde{L}^{4i\ast}_{t}, \tilde{L}^{5i\ast}_{t}, \tilde{Z}^{i\ast}_{t}$ for $t \in \{1,$ $\ldots, T, T + 1\}$ are deterministic functions of time $t$ which are dependent on the problem parameters, and can be computed backwardly in time $t$. 
\end{enumerate} 
\end{theorem} 
\begin{proof}
Similar to the proof of Theorem~\ref{Theorem_3.1}.
\end{proof}

\section{Conclusion} 
\label{Section_6}

This paper examines an execution game for two large traders under a transient price impact model. We incorporate the effect of a Markovian environment, described by an AR(1)-type random sequences, into the price dynamics of one financial asset. The existence of two (multiple) large traders indicates strategic uncertainty, and the Markovian environment indicates environmental uncertainty. We derive an execution strategy and its associated value function at a Markov perfect equilibrium and show that the Markovian environment directly affects the equilibrium execution strategy. Also, our numerical experiments demonstrate that the equilibrium execution strategy reflects various facets observed in financial markets.

One direction of future research is to consider an endogenous model for optimal or equilibrium execution problems. The submission of large orders by large traders may affect the subsequent orders posed by small traders in a financial market. Thus, endogenously incorporating the orders submitted by large traders into the modeling of aggregate orders posed by small traders deserves consideration. This model may enable us to investigate the interaction between large traders and small traders in detail.

\section*{Acknowledgements}

The authors would like to thank Dr. Seiya Kuno for closely examining our preliminary draft and the comments at the 3rd Autumn Meeting of Nippon Finance Association. Also, we would like to thank Dr. Sebastian Jaimungal for his helpful comments and suggestions at the Winter Workshop on Operations Research, Finance, and Mathematics, 2024. 

\section*{Disclosure statement} 

No potential conflict of interest was reported by the authors.

\section*{Funding} 

The authors were supported by Japan Society for the Promotion of Science (JSPS) under KAKENHI [Grant Numbers 21H04399, 21K13325, 22K01573].


\section*{Appendix}
\appendix

\section{Proof of Theorem~\ref{Theorem_3.1}}
\label{Section_A}

In this appendix, $\mathcal{S}^{n} (\mathbb{R})$ denotes the set of all $n \times n$ real-valued, symmetric, and non-singular matrices where $n \in \mathbb{Z}_{++}$. For an $n \times m$ real-valued matrix (or vector) ${\bf A}$ (where $m \in \mathbb{Z}_{++}$), ${\bf A}^{\top}$ represents the transpose of the matrix (or vector). Moreover, if a random variable $\bm{X}$ follows an $n$-dimensional normal distribution with mean $\pmb{\mu}_{\bm{X}} \in \mathbb{R}^{n}$ and covariance matrix $\pmb{\Sigma}_{\bm{X}} \in \mathcal{S}^{n} (\mathbb{R})$, we write $\bm{X} \sim N_{\mathbb{R}^{n}} \left( \pmb{\mu}_{\bm{X}}, \pmb{\Sigma}_{\bm{X}} \right)$. 

\subsection{Preliminaries}
\label{Subsection_A.1}

We first state the following well-known lemma that will be used in the proof of Theorem~\ref{Theorem_3.1}. Although the statement is straightforward, we note the result for this paper to be self-contained.

\begin{lemma}
\label{Lemma_A.1}
Define, for a set of random variables $X_{1}, X_{2}, \ldots, X_{n}$, $\mathbb{E} \left[ X_{i} \right] \coloneqq \mu_{i}$ and $\mathrm{Cov} \left[ X_{i}, X_{j} \right] \coloneqq \sigma_{ij}$. If an $\mathbb{R}^{n}$-valued random variable $\bm{X} \coloneqq (X_{1}, X_{2}, \ldots, X_{n})$ follows a normal distribution with mean $\bm{\mu}_{\bm{X}} \in \mathbb{R}^{n}$ and variance $\pmb{\Sigma}_{\bm{X}} \in \mathcal{S}^{n} (\mathbb{R})$, that is,
\begin{align}
\bm{X} \sim N_{\mathbb{R}^{n}} \left( \bm{\mu}_{\bm{X}}, \pmb{\Sigma}_{\bm{X}} \right),
\end{align}
where
\begin{align}
\bm{\mu}_{\bm{X}}
\coloneqq
\begin{pmatrix}
\mu_{1} \\
\vdots \\
\mu_{n}
\end{pmatrix}
\in \mathbb{R}^{n};
\quad
\pmb{\Sigma}_{\bm{X}}
\coloneqq
\begin{pmatrix}
	\sigma_{11} & \cdots & \sigma_{1n} \\
	\vdots       & \ddots & \vdots \\
	\sigma_{n1} & \cdots & \sigma_{nn}
\end{pmatrix}
\in \mathcal{S}^{n} (\mathbb{R}),
\end{align}
then the following sum of the random variables each of which is multiplied by a constant:
\begin{align}
\bm{c}^{\top} \bm{X} \coloneqq c_{1} X_{1} + c_{2} X_{2} + \cdots + c_{n} X_{n},
\end{align}
where $\bm{c} \coloneqq (c_{1} , \ldots, c_{n})^{\top} \in \mathbb{R}^{n}$, also follows a normal distribution as follows:
\begin{align}
c_{1} X_{1} + c_{2} X_{2} + \cdots + c_{n} X_{n} \sim N_{\mathbb{R}} \bigg( \sum_{i = 1}^{n} c_{i} \mu_{i}, \sum_{i, j = 1}^{n} c_{i} c_{j} \sigma_{ij} \bigg).
\end{align}
\end{lemma}

The next lemma is essential for the proof of Theorem~\ref{Theorem_3.1}.
\begin{lemma}[]
\label{Lemma_A.2}
Suppose that $(X, Y)^{\top} \sim N_{\mathbb{R}^{2}} \left( \pmb{\mu}, \pmb{\Sigma} \right)$, where
\begin{align}
\pmb{\mu}
\coloneqq
\begin{pmatrix}
	\mu_{X} \\
	\mu_{Y}
\end{pmatrix}
\in \mathbb{R}^{2};
\quad
\pmb{\Sigma}
\coloneqq
\begin{pmatrix}
	\sigma_{X}^{2} & \rho_{XY} \sigma_{X} \sigma_{Y} \\
	\rho_{XY} \sigma_{X} \sigma_{Y} & \sigma_{Y}^{2}
\end{pmatrix}
\in \mathcal{S}^{2} (\mathbb{R}),
\end{align}
and $\rho_{XY} \in (-1, 1)$ is the correlation coefficient between $X$ and $Y$. 
Let $a \in \mathbb{R}$ be such that 
\begin{align}
\pmb{\Sigma}^{\ast}
\coloneqq
\pmb{\Sigma}^{-1} -
\begin{pmatrix}
	2 a & 0 \\
	0   & 0
\end{pmatrix}
=
\begin{pmatrix}
	\widetilde{\sigma}_{11} - 2 a & \widetilde{\sigma}_{12} \\
	\widetilde{\sigma}_{21} & \widetilde{\sigma}_{22}
\end{pmatrix}
\in \mathcal{S}^{2} (\mathbb{R}),
\end{align}
where
\begin{align}
\pmb{\Sigma}^{-1}
\coloneqq
\begin{pmatrix}
	\widetilde{\sigma}_{11} & \widetilde{\sigma}_{12} \\
	\widetilde{\sigma}_{21} & \widetilde{\sigma}_{22}
\end{pmatrix}.
\end{align}
Then, for any $b, c \in \mathbb{R}$, we have
\begin{align}
\mathbb{E} \Big[ \exp \Big\{ a X^{2} + b X + c Y \Big\} \Big] =
\frac{\sqrt{(\pmb{\Sigma}^{\ast})^{-1}}}{\sqrt{\lvert \pmb{\Sigma} \rvert}}
\exp \Big\{ \frac{1}{2} (\pmb{\mu}^{\ast}) ^{\top} (\pmb{\Sigma}^{\ast})^{-1} \pmb{\mu}^{\ast} - \frac{1}{2} \pmb{\mu}^{\top} \pmb{\Sigma}^{-1} \pmb{\mu} \Big\}, \label{quadmgf}
\end{align}
where $\bm{s} \coloneqq (b, c)^{\top} \in \mathbb{R}^{2}$ and $\pmb{\mu}^{\ast} \coloneqq \pmb{\Sigma}^{-1} \pmb{\mu} + \bm{s}$, provided that $\pmb{\Sigma}^{\ast}$ is positive definite (i.e., $\pmb{\Sigma}^{\ast}$ is invertible).
\end{lemma}
\begin{proof}
Define $\bm{b} \coloneqq (b, c) \in \mathbb{R}^{2}$. Then, 
\begin{align}
&\mathbb{E} \Big[ \exp \Big\{ a X^{2} + b X + c Y \Big\} \Big] \\
&= \int_{\mathbb{R}^{2}} \exp \left\{ \bm{x}^{\top} \begin{pmatrix} a & 0 \\ 0 & 0 \end{pmatrix} \bm{x} + \bm{b}^{\top} \bm{x} \right\} \frac{1}{2 \pi \lvert \pmb{\Sigma} \rvert} \exp \left\{ - \frac{1}{2} ( \bm{x} - \pmb{\mu} )^{\top} \pmb{\Sigma}^{-1} ( \bm{x} - \pmb{\mu} ) \right\} \D \bm{x} \nonumber\\
&= \frac{1}{2 \pi \lvert \pmb{\Sigma} \rvert^{\frac{1}{2}}} 
\int_{\mathbb{R}^{2}} \exp \left\{ \bm{x}^{\top} \begin{pmatrix} a & 0 \\ 0 & 0 \end{pmatrix} \bm{x} + \bm{b}^{\top} \bm{x} - \frac{1}{2} \bm{x}^{\top} \begin{pmatrix} \widetilde{\sigma}_{11} & \widetilde{\sigma}_{12} \\ \widetilde{\sigma}_{21} & \widetilde{\sigma}_{22} \end{pmatrix} \bm{x} + \pmb{\mu}^{\top} \begin{pmatrix} \widetilde{\sigma}_{11} & \widetilde{\sigma}_{12} \\ \widetilde{\sigma}_{21} & \widetilde{\sigma}_{22} \end{pmatrix} \bm{x} \right. \nonumber\\
&\quad \left. - \frac{1}{2} \pmb{\mu}^{\top} \begin{pmatrix} \widetilde{\sigma}_{11} & \widetilde{\sigma}_{11} \\ \widetilde{\sigma}_{21} & \widetilde{\sigma}_{22} \end{pmatrix} \pmb{\mu} \right\} \D \bm{x} \nonumber\\
&= \frac{1}{2 \pi \lvert \pmb{\Sigma} \rvert^{\frac{1}{2}}} 
\int_{\mathbb{R}^{2}} \exp \Big\{ - \frac{1}{2} \bm{x}^{\top} \begin{pmatrix} \widetilde{\sigma}_{11} - 2 a & \widetilde{\sigma}_{12} \\ \widetilde{\sigma}_{21} & \widetilde{\sigma}_{22} \end{pmatrix} \bm{x} + \left[ \pmb{\mu}^{\top} \pmb{\Sigma}^{-1} + \bm{s}^{\top} \right] \bm{x} - \frac{1}{2} \pmb{\mu}^{\top} \pmb{\Sigma}^{-1} \pmb{\mu} \Big\} \D \bm{x} \nonumber\\
&= \frac{2 \pi \lvert ( \pmb{\Sigma}^{\ast} )^{-1} \rvert^{\frac{1}{2}}}{2 \pi \lvert \pmb{\Sigma} \rvert^{\frac{1}{2}}} \cdot 
\underbrace{\frac{1}{2 \pi \lvert (\pmb{\Sigma}^{\ast})^{-1} \rvert^{\frac{1}{2}}} 
\int_{\mathbb{R}^{2}} \exp \Big\{ - \frac{1}{2} \left( \bm{x} - ( \pmb{\Sigma}^{\ast} )^{-1} \pmb{\mu}^{\ast} \right)^{\top} \pmb{\Sigma}^{\ast} \left( \bm{x} - ( \pmb{\Sigma}^{\ast} )^{-1} \pmb{\mu}^{\ast} \right) \Big\} \D \bm{x}}_{= 1} \nonumber\\
&\quad \times \exp \Big\{ \frac{1}{2} (\pmb{\mu}^{\ast}) ^{\top} ( \pmb{\Sigma}^{\ast} )^{-1} \pmb{\mu}^{\ast} - \frac{1}{2} \pmb{\mu}^{\top} \pmb{\Sigma}^{-1} \pmb{\mu} \Big\}. \label{calculationofnorm}
\end{align}
Note that $\D \bm{x} \coloneqq \D x_{1} \D x_{2}$.
\end{proof}

Define
\begin{align}
\begin{pmatrix} \pi_{11} & \pi_{12} \\ \pi_{21} & \pi_{22} \end{pmatrix} \coloneqq \begin{pmatrix} \widetilde{\sigma}_{11} - 2a & \widetilde{\sigma}_{12} \\ \widetilde{\sigma}_{21} & \widetilde{\sigma}_{22} \end{pmatrix}^{-1} = (\pmb{\Sigma}^{\ast})^{-1}.
\end{align}
Then, rearranging Eq.~\eqref{quadmgf} results in
\begin{align}
\mathbb{E} \Big[ \exp \Big\{ a X^{2} + b X + c Y \Big\} \Big]
=
\frac{\sqrt{(\pmb{\Sigma}^{\ast})^{-1}}}{\sqrt{\lvert \pmb{\Sigma} \rvert}}
\exp \left\{ \frac{1}{2} \left[ \pi_{11} b^{2} + \pi_{22} c^{2} + 2 \pi_{12} b c + 2 \mu^{b} b + 2 \mu^{c} c + \mu^{a} \right] \right\},
\end{align}
where
\begin{align}
\mu^{a} &\coloneqq \pmb{\mu}^{\top} \pmb{\Sigma}^{-1} \pmb{\tilde{\Sigma}} \left( \pmb{\Sigma}^{\ast} \right)^{-1} \pmb{\Sigma}^{-1} \pmb{\mu} - \pmb{\mu}^{\top} \pmb{\Sigma}^{-1} \pmb{\mu}; \\
\mu^{b} &\coloneqq \left( \widetilde{\sigma}_{11} \pi_{11} + \widetilde{\sigma}_{12} \pi_{21} \right) \mu_{X} + \left( \widetilde{\sigma}_{21} \pi_{11} + \widetilde{\sigma}_{22} \pi_{21} \right) \mu_{Y}; \\
\mu^{c} &\coloneqq \left( \widetilde{\sigma}_{11} \pi_{12} + \widetilde{\sigma}_{12} \pi_{22} \right) \mu_{X} + \left( \widetilde{\sigma}_{21} \pi_{12} + \widetilde{\sigma}_{22} \pi_{22} \right) \mu_{Y}.
\end{align}
Note that $( \pmb{\Sigma}^{\ast} )^{-1}$ is symmetric.

\subsection{Proof of Theorem~\ref{Theorem_3.1}}
\label{Subsection_A.2}

We derive the execution volume $q^{i\ast}_{t}$ at the Markov perfect equilibrium for each large trader $i \in \{1, 2\}$ at time $t \in \{1, \ldots, T\}$ by backward induction method of dynamic programming from time $t = T$ via the following steps. 

\begin{flushleft}
\textbf{Step 1} From the assumption that each large trader must unwind all the remainder of his/her position at time $t = T$, we have
\end{flushleft}
\begin{align} 
\overline{Q}_{T + 1}^{i} = \overline{Q}^{i}_{T} - q^{i}_{T} = 0,
\end{align}
for $i \in \{1, 2\}$. Thus, $q^{i\ast}_{T} = \overline{Q}^{i}_{T}$ holds for $i \in \{1, 2\}$. Then, for $t = T$, the value function for each large trader $i, j \in \{1, 2\}$ $(i \neq j)$ is
\begin{align} 
V^{i}_{T} (\pi^{1\ast}_{T}, \pi^{2\ast}_{T}) \big[ \bm{s}_{T} \big]
&= \sup_{q^{i}_{T} \in \mathbb{R}} \mathbb{E} \Big[ g^{i}_{T + 1} (\bm{s}_{T + 1}) \Big\vert \bm{s}_{T} \Big] \nonumber\\
&= \sup_{q^{i}_{T} \in \mathbb{R}} \mathbb{E} \Big[ - \exp \left\{ - \gamma^{i} W^{i}_{T + 1} \right\} \Big\vert W^{1}_{T}, W^{2}_{T}, P_{T}, \overline{Q}^{1}_{T}, \overline{Q}^{2}_{T}, R_{T}, \mathcal{I}_{T - 1} \Big] \nonumber\\
&= \sup_{q^{i}_{T} \in \mathbb{R}} \mathbb{E} \Big[ - \exp \left\{ - \gamma^{i} \left[ W^{i}_{T} - \big[ P_{T} + \lambda_{T} (q^{i}_{T} + q^{j}_{T}) \big] q^{i}_{T} \right] \right\} \Big\vert W^{1}_{T}, W^{2}_{T}, P_{T}, \overline{Q}^{1}_{T}, \overline{Q}^{2}_{T}, R_{T}, \mathcal{I}_{T - 1} \Big] \nonumber\\ 
&= - \exp \left\{ - \gamma^{i} \big[ W^{i}_{T} - P_{T} \overline{Q}^{i}_{T} - \lambda_{T} (\overline{Q}^{i}_{T})^{2} - \lambda_{T} \overline{Q}^{i}_{T} \overline{Q}^{j}_{T} \big] \right\} \nonumber\\
&= - \exp \left\{ - \gamma^{i} \big[ W^{i}_{T} - P_{T} \overline{Q}^{i}_{T} + G^{1i}_{T} (\overline{Q}^{i}_{T})^{2} + J^{1i}_{T} \overline{Q}^{i}_{T} \overline{Q}^{j}_{T} \big] \right\},
\end{align}
where
\begin{align}
G^{1i}_{T} &\coloneqq - \lambda_{T} \ (< 0); \\
J^{1i}_{T} &\coloneqq - \lambda_{T} \ (< 0).
\end{align}

\begin{flushleft}
\textbf{Step 2} For $t = T - 1$, the value functions, $V^{i}_{T - 1} (\pi^{1\ast}_{T - 1}, \pi^{2\ast}_{T - 1}) \big[ \bm{s}_{T - 1} \big]$ for each large trader $i \in \{1, 2\}$, satisfy the following functional equations:
\end{flushleft}
\begin{align}
&V^{i}_{T - 1} (\pi^{1\ast}_{T - 1}, \pi^{2\ast}_{T - 1}) \big[ \bm{s}_{T - 1} \big] \nonumber\\
&= \sup_{q^{i}_{T - 1} \in \mathbb{R}} \mathbb{E} \Big[ V^{i}_{T} (\pi^{1 \ast}_{T}, \pi^{2 \ast}_{T})  \big[ \bm{s}_{T} \big] \Big\vert \bm{s}_{T - 1} \Big] \nonumber\\
&= \sup_{q^{i}_{T - 1} \in \mathbb{R}} \mathbb{E} \Big[ - \exp \Big\{ - \gamma^{i} \big[ W^{i}_{T} - P_{T} \overline{Q}^{i}_{T} + G^{1i}_{T} (\overline{Q}^{i}_{T})^{2} + J^{1i}_{T} \overline{Q}^{i}_{T} \overline{Q}^{j}_{T} \big] \Big\} \Big\vert \bm{s}_{T - 1} \Big] \nonumber\\
&= \sup_{q^{i}_{T - 1} \in \mathbb{R}} - \exp \Big\{ - \gamma^{i} \Big[ (- \lambda_{T - 1} + \alpha^{T - 1} \lambda_{T - 1} + G^{1i}_{T}) (q^{i}_{T - 1})^{2} + \big[ (- \alpha^{T - 1} \lambda_{T - 1} - 2 G^{1i}_{T}) \overline{Q}^{i}_{T - 1} \nonumber\\
&\quad + (- J^{1i}_{T}) \overline{Q}^{j}_{T - 1} + \{ - (1 - \mathrm{e}^{-\rho}) \} R_{T - 1} + (- \lambda_{T - 1} + \alpha^{T - 1} \lambda_{T - 1} + J^{1i}_{T}) q^{j}_{T - 1} \big] q^{i}_{T - 1} \nonumber\\
&\quad + W^{i}_{T - 1} - P_{T - 1} \overline{Q}^{i}_{T - 1} + G^{1i}_{T} (\overline{Q}^{i}_{T - 1})^{2} + (1 - \mathrm{e}^{- \rho}) \overline{Q}^{i}_{T - 1} R_{T - 1} + J^{1i}_{T} \overline{Q}^{i}_{T - 1} \overline{Q}^{j}_{T - 1} \nonumber\\
&\quad + (- \alpha^{T - 1} \lambda_{T - 1} - J^{1i}_{T}) \overline{Q}^{i}_{T - 1} q^{j}_{T - 1} \Big] \Big\} \nonumber\\
&\quad \times \mathbb{E} \Big[ \exp \Big\{ \gamma^{i} (\overline{Q}^{i}_{T - 1} - q^{i}_{T - 1}) ( \mathcal{I}_{T - 1} + \epsilon_{T - 1}) \Big\} \Big\vert \bm{s}_{T - 1} \Big],
\label{funcformatt=T-1}
\end{align}
where $\alpha^{T - 1} \coloneqq \alpha_{T - 1} \mathrm{e}^{- \rho} + \beta_{T - 1}$. As for the expectation term in \eqref{funcformatt=T-1}, we have 
\begin{align}
&\mathbb{E} \Big[ \exp \Big\{ \gamma^{i} \big( \overline{Q}^{i}_{T - 1} - q^{i}_{T - 1} \big) \big( \mathcal{I}_{T - 1} + \epsilon_{T - 1} \big) \Big\}  \Big\vert \bm{s}_{T - 1} \Big] \nonumber\\
&= \exp \left\{ \gamma^{i} \big( \overline{Q}^{i}_{T - 1} - q^{i}_{T - 1} \big) \big( a^{\mathcal{I}}_{T - 1} - b^{\mathcal{I}}_{T - 1} \mathcal{I}_{T - 2} 
\big) + \frac{1}{2} (\gamma^{i})^{2} \big( \overline{Q}^{i}_{T - 1} - q^{i}_{T - 1} \big)^{2} \Sigma^{\mathcal{I}, \epsilon}_{T - 1} \right\},
\label{mgfT-1}
\end{align}
where 
\begin{align}
\Sigma^{\mathcal{I}, \epsilon}_{T - 1}
\coloneqq \mathbb{V} \Big[ \mathcal{I}_{T - 1} + \epsilon_{T - 1} \Big\vert \bm{s}_{T - 1} \Big]
= (\sigma^{\mathcal{I}}_{T - 1})^{2} + (\sigma^{\epsilon}_{T - 1})^{2} + 2 \rho^{\mathcal{I}, \epsilon} \sigma^{\mathcal{I}}_{T - 1} \sigma^{\epsilon}_{T - 1},
\end{align}
according to Lemma~\ref{Lemma_A.1}. Thus, substituting Eq.~\eqref{mgfT-1} into Eq.~\eqref{funcformatt=T-1} and rearranging results in
\begin{align}
&V^{i}_{T - 1} (\pi^{1\ast}_{T - 1}, \pi^{2\ast}_{T - 1}) \big[ \bm{s}_{T - 1} \big] \nonumber\\
&= \sup_{q^{i}_{T - 1} \in \mathbb{R}} - \exp \Big\{ - \gamma^{i} \Big[ - A^{i}_{T - 1} (q^{i}_{T - 1})^{2} + \big[ B^{i}_{T - 1} \overline{Q}^{i}_{T - 1} + C^{i}_{T - 1} \overline{Q}^{j}_{T - 1} + D^{i}_{T - 1} R_{T - 1} + F^{i}_{T - 1} \mathcal{I}_{T - 2} + M^{i}_{T - 1} \nonumber\\
&\quad + N^{i}_{T - 1} q^{j}_{T - 1} \big] q^{i}_{T - 1} + W^{i}_{T - 1} - P_{T - 1} \overline{Q}^{i}_{T - 1} + \left( G^{1i}_{T} - \frac{1}{2} \gamma^{i} \Sigma^{\mathcal{I}, \epsilon}_{T - 1} \right) (\overline{Q}^{i}_{T - 1})^{2} + (- a^{\mathcal{I}}_{T - 1} 
) \overline{Q}^{i}_{T - 1} \nonumber\\
&\quad + (1 - \mathrm{e}^{-\rho}) R_{T - 1} \overline{Q}^{i}_{T - 1} + J^{1i}_{T} \overline{Q}^{i}_{T - 1} \overline{Q}^{j}_{T - 1} + b^{\mathcal{I}}_{T - 1} \overline{Q}^{i}_{T - 1} \mathcal{I}_{T - 2} + (- \alpha^{T - 1} \lambda_{T - 1} - J^{1i}_{T}) \overline{Q}^{i}_{T - 1} q^{j}_{T - 1} \Big] \Big\},
\label{vfatT-1befsup}
\end{align}
with the following relations: 
\begin{align}
A^{i}_{T - 1} &\coloneqq \lambda_{T - 1} - \alpha^{T - 1} \lambda_{T - 1} - G^{1i}_{T} + \frac{1}{2} \gamma^{i} \Sigma^{\mathcal{I}, \epsilon}_{T - 1} (> 0); \\
B^{i}_{T - 1} &\coloneqq - \alpha^{T - 1} \lambda_{T - 1} - 2 G^{1i}_{T} + \gamma^{i} \Sigma^{\mathcal{I}, \epsilon}_{T - 1}; \\
C^{i}_{T - 1} &\coloneqq - J^{1i}_{T}; \\
D^{i}_{T - 1} &\coloneqq - (1 - \mathrm{e}^{- \rho}); \\
F^{i}_{T - 1} &\coloneqq - b^{\mathcal{I}}_{T - 1}; \\
M^{i}_{T - 1} &\coloneqq a^{\mathcal{I}}_{T - 1}
; \\
N^{i}_{T - 1} &\coloneqq - \lambda_{T - 1} + \alpha^{T - 1} \lambda_{T - 1} + J^{1i}_{T}.
\end{align}
Note that, for all $B, C \in \mathbb{R}$ and all $\gamma, A \in \mathbb{R}_{++}$, two functions $c_{1} (x) \coloneqq - \exp \{- \gamma x\}$ and $c_{2} (x) \coloneqq - A x^{2} + B x + C$ are strictly concave functions, and therefore so is the composite function of the two, $K (x) \coloneqq c_{1} \circ c_{2} (x) = - \exp \big\{ - \gamma (- A x^{2} + B x + C) \big\}$. Thus, we obtain the execution volume attaining the supremum of Eq.~\eqref{vfatT-1befsup} by completing the square of the following function:
\begin{align}
&K^{i}_{T - 1} (q^{i}_{T - 1})
\coloneqq - A^{i}_{T - 1} (q^{i}_{T - 1})^{2} + \big[ B^{i}_{T - 1} \overline{Q}^{i}_{T - 1} + C^{i}_{T - 1} \overline{Q}^{j}_{T - 1} + D^{i}_{T - 1} R_{T - 1} + F^{i}_{T - 1} \mathcal{I}_{T - 2} + M^{i}_{T - 1} \nonumber\\
&\ + N^{i}_{T - 1} q^{j}_{T - 1} \big] q^{i}_{T - 1} + W^{i}_{T - 1} - P_{T - 1} \overline{Q}^{i}_{T - 1} + \left( G^{1i}_{T - 1} - \frac{1}{2} \gamma^{i} \Sigma^{\mathcal{I}, \epsilon}_{T - 1} \right) (\overline{Q}^{i}_{T - 1})^{2} + (- a^{\mathcal{I}}_{T - 1} 
) \overline{Q}^{i}_{T - 1} \nonumber\\
&\ + (1 - \mathrm{e}^{- \rho}) R_{T - 1} \overline{Q}^{i}_{T - 1} + J^{1i}_{T} \overline{Q}^{i}_{T - 1} \overline{Q}^{j}_{T - 1} + b^{\mathcal{I}}_{T - 1} \overline{Q}^{i}_{T - 1} \mathcal{I}_{t - 2} + (- \alpha^{T - 1} \lambda_{T - 1} - J^{1i}_{T}) \overline{Q}^{i}_{T - 1} q^{j}_{T - 1}.
\end{align}
The best response of large trader $i \in \{1, 2\}$ to the other large trader $j \neq i$, denoted by $BR^{i} (q^{j}_{T - 1})$, then becomes
\begin{align}
BR^{i} (q^{j}_{T - 1}) = \frac{1}{2 A^{i}_{T - 1}} \left( B^{i}_{T - 1} \overline{Q}^{i}_{T - 1} + C^{i}_{T - 1} \overline{Q}^{j}_{T - 1} + D^{i}_{T - 1} R_{T - 1} + F^{i}_{T - 1} \mathcal{I}_{T - 2} + M^{i}_{T - 1} + N^{i}_{T - 1} q^{j}_{T - 1} \right).
\end{align}
Thus, at the Markov perfect equilibrium, we have
\begin{align}
q^{i\ast}_{T - 1} = \frac{1}{2 A^{i}_{T - 1}} \left( B^{i}_{T - 1} \overline{Q}^{i}_{T - 1} + C^{i}_{T - 1} \overline{Q}^{j}_{T - 1} + D^{i}_{T - 1} R_{T - 1} + F^{i}_{T - 1} \mathcal{I}_{T - 2} + M^{i}_{T - 1} + N^{i}_{T - 1} q^{j\ast}_{T - 1} \right); \\
q^{j\ast}_{T - 1} = \frac{1}{2 A^{j}_{T - 1}} \left( B^{j}_{T - 1} \overline{Q}^{j}_{T - 1} + C^{j}_{T - 1} \overline{Q}^{i}_{T - 1} + D^{j}_{T - 1} R_{T - 1} + F^{j}_{T - 1} \mathcal{I}_{T - 2} + M^{j}_{T - 1} + N^{j}_{T - 1} q^{i\ast}_{T - 1} \right).
\end{align}
Solving the above simultaneous equations results in
\begin{align}
q^{i\ast}_{T - 1} 
&= B^{i\ast}_{T - 1} \overline{Q}^{i}_{T - 1} + C^{i\ast}_{T - 1} \overline{Q}^{j}_{T - 1} + D^{i\ast}_{T - 1} R_{T - 1} + F^{i\ast}_{T - 1} \mathcal{I}_{T - 2} + M^{i\ast}_{T - 1} \nonumber\\
&(\eqqcolon a^{i}_{T - 1} + b^{i}_{T - 1} \overline{Q}^{i}_{T - 1} + c^{i}_{T - 1} \overline{Q}^{j}_{T - 1} + d^{i}_{T - 1} R_{T - 1} + e^{i}_{T - 1} \mathcal{I}_{T - 2}),
\end{align}
where
\begin{align}
\zeta^{i}_{T - 1} &\coloneqq 2 A^{i}_{T - 1} - \frac{N^{i}_{T - 1} N^{j}_{T - 1}}{2 A^{j}_{T - 1}}; &
B^{i\ast}_{T - 1} &\coloneqq \frac{1}{\zeta^{i}_{T - 1}} \left( B^{i}_{T - 1} + \frac{N^{i}_{T - 1} C^{j}_{T - 1}}{2 A^{j}_{T - 1}} \right); \nonumber\\
C^{i\ast}_{T - 1} &\coloneqq \frac{1}{\zeta^{i}_{T - 1}} \left( C^{i}_{T - 1} + \frac{N^{i}_{T - 1} B^{j}_{T - 1}}{2 A^{j}_{T - 1}} \right); &
D^{i\ast}_{T - 1} &\coloneqq \frac{1}{\zeta^{i}_{T - 1}} \left( D^{i}_{T - 1} + \frac{N^{i}_{T - 1} D^{j}_{T - 1}}{2 A^{j}_{T - 1}} \right); \nonumber\\
F^{i\ast}_{T - 1} &\coloneqq \frac{1}{\zeta^{i}_{T - 1}} \left( F^{i}_{T - 1} + \frac{N^{i}_{T - 1} F^{j}_{T - 1}}{2 A^{j}_{T - 1}} \right); &
M^{i\ast}_{T - 1} &\coloneqq \frac{1}{\zeta^{i}_{T - 1}} \left( M^{i}_{T - 1} + \frac{N^{i}_{T - 1} M^{j}_{T - 1}}{2 A^{j}_{T - 1}} \right).
\end{align}
for each $i \in \{1, 2\}$. $q^{1\ast}_{T - 1}$ and $q^{2\ast}_{T - 1}$ are equilibrium execution volumes at the Markov perfect equilibrium for time $T - 1$. The value function for each large trader $i \in \{1, 2\}$ at the Markov Perfect equilibrium $(\pi^{1 \ast}, \pi^{2 \ast}) \in \Pi_{M}^{1} \times \Pi_{M}^{2}$ then becomes 
\begin{align}
&V^{i}_{T - 1} (\pi^{1\ast}_{T - 1}, \pi^{2\ast}_{T - 1}) \big[ \bm{s}_{T - 1} \big] \nonumber\\
&= - \exp \Big\{ - \gamma^{i} \Big[ W^{i}_{T - 1} - P_{T - 1} \overline{Q}^{i}_{T - 1} + \left( G^{1i}_{T} - \frac{1}{2} \gamma^{i} \Sigma^{\mathcal{I}, \epsilon}_{T - 1} \right) (\overline{Q}^{i}_{T - 1})^{2} + (- a^{\mathcal{I}}_{T - 1}) \overline{Q}^{i}_{T - 1} \nonumber\\
&\quad + (1 - \mathrm{e}^{- \rho}) R_{T - 1} \overline{Q}^{i}_{T - 1} + J^{1i}_{T} \overline{Q}^{i}_{T - 1} \overline{Q}^{j}_{T - 1} + b^{\mathcal{I}}_{T - 1} \overline{Q}^{i}_{T - 1} \mathcal{I}_{T - 2} + (- \alpha^{T - 1} \lambda_{T - 1} - J^{1i}_{T}) \overline{Q}^{i}_{T - 1} q^{j\ast}_{T - 1} \nonumber\\
&\quad + \frac{1}{4 A^{i}_{T - 1}} (B^{i\ast\ast}_{T - 1} \overline{Q}^{i}_{T - 1} + C^{i\ast\ast}_{T - 1} \overline{Q}^{j}_{T - 1} + D^{i\ast}_{T - 1} R_{t} + F^{i\ast\ast}_{T - 1} \mathcal{I}_{T - 2} + M^{i\ast\ast}_{T - 1})^{2} \Big] \Big\} \nonumber\\
&= - \exp \Big\{ - \gamma^{i} \Big[ W^{i}_{T - 1} - P_{T - 1} \overline{Q}^{i}_{T - 1} + G^{1i}_{T - 1} (\overline{Q}^{i}_{T - 1})^{2} + G^{2i}_{T - 1} \overline{Q}^{i}_{T - 1} + H^{1i}_{T - 1} \overline{Q}^{i}_{T - 1} R_{T - 1} \nonumber\\
&\quad + H^{2i}_{T - 1} R^{2}_{T - 1} + H^{3i}_{T - 1} R_{T - 1} + J^{1i}_{T - 1} \overline{Q}^{i}_{T - 1} \overline{Q}^{j}_{T - 1} + J^{2i}_{T - 1} \overline{Q}^{j}_{T - 1} R_{T - 1} + J^{3i}_{T - 1} (\overline{Q}^{j}_{T - 1})^{2} + J^{4i}_{T - 1} \overline{Q}^{j}_{T - 1} \nonumber\\
&\quad + L^{1i}_{T - 1} \overline{Q}^{i}_{T - 1} \mathcal{I}_{T - 2} + L^{2i}_{T - 1} R_{T - 1} \mathcal{I}_{T - 2} + L^{3i}_{T - 1} \overline{Q}^{j}_{T - 1} \mathcal{I}_{T - 2} + L^{4i}_{T - 1} \mathcal{I}^{2}_{T - 2} + L^{5i}_{T - 1} \mathcal{I}_{T - 2} + Z^{i}_{T - 1} \Big] \Big\},
\end{align} 
where 
\begin{align}
B^{i\ast\ast}_{T - 1} &\coloneqq B^{i}_{T - 1} + N^{i}_{T - 1} C^{j\ast}_{T - 1}; &
C^{i\ast\ast}_{T - 1} &\coloneqq C^{i}_{T - 1} + N^{i}_{T - 1} B^{j\ast}_{T - 1}; \nonumber\\
D^{i\ast\ast}_{T - 1} &\coloneqq D^{i}_{T - 1} + N^{i}_{T - 1} D^{j\ast}_{T - 1}; &
F^{i\ast\ast}_{T - 1} &\coloneqq F^{i}_{T - 1} + N^{i}_{T - 1} F^{j\ast}_{T - 1}; \nonumber\\
M^{i\ast\ast}_{T - 1} &\coloneqq M^{i}_{T - 1} + N^{i}_{T - 1} M^{j\ast}_{T - 1}, &
i, j &= 1, 2, \quad i \neq j,
\end{align}
and
\begin{align}
G^{1i}_{T - 1} &\coloneqq G^{1i}_{T} - \frac{1}{2} \gamma^{i} \Sigma^{\mathcal{I}, \epsilon}_{T - 1} + (- \alpha^{T - 1} \lambda_{T - 1} - J^{1i}_{T}) C^{j\ast}_{T - 1} + \frac{(B^{i\ast\ast}_{T - 1})^{2}}{4 A^{i}_{T - 1}}; \nonumber\\
G^{2i}_{T - 1} &\coloneqq - a^{\mathcal{I}}_{T - 1} + ( - \alpha^{T - 1} \lambda_{T - 1} - J^{1i}_{T} ) M^{j\ast}_{T - 1} + \frac{B^{i\ast\ast}_{T - 1} M^{i\ast\ast}_{T - 1}}{2 A^{i}_{T - 1}}; \nonumber\\
H^{1i}_{T - 1} &\coloneqq (1 - \mathrm{e}^{-\rho}) + (- \alpha^{T - 1} \lambda_{T - 1} - J^{1i}_{T}) D^{j\ast}_{T - 1} + \frac{B^{i\ast\ast}_{T - 1}D^{i\ast\ast}_{T - 1}}{2 A^{i}_{T - 1}}; \nonumber\\
H^{2i}_{T - 1} &\coloneqq \frac{(D^{i\ast\ast}_{T - 1})^{2}}{4 A^{i}_{T - 1}}; \quad
H^{3i}_{T - 1} \coloneqq \frac{D^{i\ast\ast}_{T - 1} M^{i\ast\ast}_{T - 1}}{2 A^{i}_{T - 1}}; \quad
J^{1i}_{T - 1} \coloneqq J^{1i}_{T} + (- \alpha^{T - 1} \lambda_{T - 1} - J^{1i}_{T}) B^{j\ast}_{T - 1} + \frac{B^{i\ast\ast}_{T - 1} C^{i\ast\ast}_{T - 1}}{2 A^{i}_{T - 1}}; \nonumber\\
J^{2i}_{T - 1} &\coloneqq \frac{C^{i\ast\ast}_{T - 1} D^{i\ast\ast}_{T - 1}}{2 A^{i}_{T - 1}}; \quad
J^{3i}_{T - 1} \coloneqq \frac{(C^{i\ast\ast}_{T - 1})^{2}}{4 A^{i}_{T - 1}}; \quad
J^{4i}_{T - 1} \coloneqq \frac{C^{i\ast\ast}_{T - 1} M^{i\ast\ast}_{T - 1}}{2 A^{i}_{T - 1}}; \nonumber\\
L^{1i}_{T - 1} &\coloneqq b^{\mathcal{I}}_{T - 1} + (- \alpha^{T - 1} \lambda_{T - 1} - J^{1i}_{T}) F^{j\ast}_{T - 1} + \frac{B^{i\ast\ast}_{T - 1} F^{i\ast\ast}_{T - 1}}{2 A^{i}_{T - 1}}; \quad
L^{2i}_{T - 1} \coloneqq \frac{D^{i\ast\ast}_{T - 1} F^{i\ast\ast}_{T - 1}}{2 A^{i}_{T - 1}}; \nonumber\\
L^{3i}_{T - 1} &\coloneqq \frac{C^{i\ast\ast}_{T - 1} F^{i\ast\ast}_{T - 1}}{2 A^{i}_{T - 1}}; \quad
L^{4i}_{T - 1} \coloneqq \frac{(F^{i\ast\ast}_{T - 1})^{2}}{4 A^{i}_{T - 1}}; \quad
L^{5i}_{T - 1} \coloneqq \frac{F^{i\ast\ast}_{T - 1} M^{i\ast\ast}_{T - 1}}{2 A^{i}_{T - 1}}; \quad
Z^{i}_{T - 1} \coloneqq \frac{(M^{i\ast\ast}_{T - 1})^{2}}{4 A^{i}_{T - 1}}.
\end{align}

\begin{flushleft}
\textbf{Step 3} From the above results, for $t + 1 \in \{T - 1, \ldots, 2\}$, we can assume that the optimal value function at time $t + 1$ has the following functional form: 
\end{flushleft}
\begin{align}
&V^{i}_{t + 1} (\pi^{1\ast}_{t + 1}, \pi^{2\ast}_{t + 1}) \big[ \bm{s}_{t + 1} \big] \nonumber\\
&= - \exp \Big\{ - \gamma^{i} \Big[ W^{i}_{t + 1} - P_{t + 1} \overline{Q}^{i}_{t + 1} + G^{1i}_{t + 1} (\overline{Q}^{i}_{t + 1})^{2} + G^{2i}_{t + 1} \overline{Q}^{i}_{t + 1} + H^{1i}_{t + 1} \overline{Q}^{i}_{t + 1} R_{t + 1} \nonumber\\
&\quad + H^{2i}_{t + 1} R_{t + 1}^{2} + H^{3i}_{t + 1} R_{t + 1} + J^{1i}_{t + 1} \overline{Q}^{i}_{t + 1} \overline{Q}^{j}_{t + 1} + J^{2i}_{t + 1} R_{t + 1} \overline{Q}^{j}_{t + 1} + J^{3i}_{t + 1} (\overline{Q}^{j}_{t + 1})^{2} + J^{4i}_{t + 1} \overline{Q}^{j}_{t + 1} \nonumber\\
&\quad + L^{1i}_{t + 1} \overline{Q}^{i}_{t + 1} \mathcal{I}_{t} + L^{2i}_{t + 1} R_{t + 1} \mathcal{I}_{t} + L^{3i}_{t + 1} \overline{Q}^{j}_{t + 1} \mathcal{I}_{t} + L^{4i}_{t + 1} \mathcal{I}_{t}^{2} + L^{5i}_{t + 1} \mathcal{I}_{t} + Z^{i}_{t + 1} \Big] \Big\}.
\end{align}
Then, at time $t$, we have
\begin{align}
&V^{i}_{t} (\pi^{1\ast}_{t}, \pi^{2\ast}_{t}) \big[ \bm{s}_{t} \big] \nonumber\\
&= \sup_{q^{i}_{t} \in \mathbb{R}} \mathbb{E} \Big[ V^{i}_{t + 1} (\pi^{1\ast}_{t + 1}, \pi^{2\ast}_{t + 1}) \big[ \bm{s}_{t + 1} \big] \Big\vert \bm{s}_{t} \Big] \nonumber\\
&= \sup_{q^{i}_{t} \in \mathbb{R}} - \exp \Big\{ - \gamma^{i} \Big[ - \left\{ (1 - \alpha^{t}) \lambda_{t} - G^{1i}_{t + 1} + \alpha_{t} \lambda_{t} \mathrm{e}^{- \rho} H^{1i}_{t + 1} + \alpha^{2}_{t} \lambda^{2}_{t} \mathrm{e}^{- 2 \rho} H^{2i}_{t + 1} \right\} (q^{i}_{t})^{2} \nonumber\\
&\quad + \big[ (- \alpha^{t} \lambda_{t} - 2 G^{1i}_{t + 1} + \alpha_{t} \lambda_{t} \mathrm{e}^{- \rho} H^{1i}_{t + 1}) \overline{Q}^{i}_{t} + \{ - (1 - \mathrm{e}^{- \rho}) - \mathrm{e}^{- \rho} H^{1i}_{t + 1}  + 2 \alpha_{t} \lambda_{t} \mathrm{e}^{- 2 \rho} H^{2i}_{t + 1} \} R_{t} \nonumber\\
&\quad + \big( - J^{1i}_{t + 1} + \alpha_{t} \lambda_{t} \mathrm{e}^{- \rho} J^{2i}_{t + 1} \big) \overline{Q}^{j}_{t} + \big\{ - (1 - \alpha^{t}) \lambda_{t} - \alpha_{t} \lambda_{t} \mathrm{e}^{- \rho} H^{1i}_{t + 1} + 2 \alpha^{2}_{t} \lambda^{2}_{t} \mathrm{e}^{- 2 \rho} H^{2i}_{t + 1} + J^{1i}_{t + 1} \nonumber\\
&\quad - \alpha_{t} \lambda_{t} \mathrm{e}^{- \rho} J^{2i}_{t + 1} \big\} q^{j}_{t} + (- G^{2i}_{t + 1} + \alpha_{t} \lambda_{t} \mathrm{e}^{- \rho} H _{t + 1}^{3i}) \big] q^{i}_{t} \nonumber\\
&\quad + W^{i}_{t} - P_{t} \overline{Q}^{i}_{t} + G^{1i}_{t + 1} (\overline{Q}^{i}_{t})^{2} + G^{2i}_{t + 1} \overline{Q}^{i}_{t} + \{(1 - \mathrm{e}^{- \rho}) + \mathrm{e}^{- \rho} H^{1i}_{t + 1}\} \overline{Q}^{i}_{t} R_{t} + \mathrm{e}^{- 2 \rho} H^{2i}_{t + 1} R^{2}_{t} + \mathrm{e}^{- \rho} H^{3i}_{t + 1} R_{t} \nonumber\\
&\quad + J^{1i}_{t + 1} \overline{Q}^{i}_{t} \overline{Q}^{j}_{t} + \mathrm{e}^{- \rho} J^{2i}_{t + 1} R_{t} \overline{Q}^{j}_{t} + J^{3i}_{t + 1} (\overline{Q}^{j}_{t})^{2} + J^{4i}_{t + 1} \overline{Q}^{j}_{t} + Z^{i}_{t + 1} \nonumber\\
&\quad + (\alpha^{2}_{t} \lambda^{2}_{t} \mathrm{e}^{- 2 \rho} H^{2i}_{t + 1} - \alpha_{t} \lambda_{t} \mathrm{e}^{- \rho} J^{2i}_{t + 1} + J^{3i}_{t + 1}) (q^{j}_{t})^{2} +\big[ (- \alpha^{t} \lambda_{t} + \alpha_{t} \lambda_{t} \mathrm{e}^{- \rho} H_{t + 1}^{1 i} - J^{1i}_{t + 1}) \overline{Q}^{i}_{t} \nonumber\\
&\quad + (2 \alpha_{t} \lambda_{t} \mathrm{e}^{- 2\rho} H^{2i}_{t + 1} - \mathrm{e}^{- \rho} J^{2i}_{t + 1}) R_{t} + (\alpha_{t} \lambda_{t} \mathrm{e}^{- \rho} J^{2i}_{t + 1} - 2 J_{t + 1}^{3 i}) \overline{Q}^{j}_{t} + (\alpha_{t} \lambda_{t} \mathrm{e}^{- \rho} H_{t + 1}^{3 i} - J_{t + 1}^{4 i}) \big] q^{j}_{t} \Big] \Big\} \nonumber\\
&\times \mathbb{E} \Big[ \exp \Big\{ - \gamma^{i} \Big[ L_{t + 1}^{4i} \mathcal{I}^{2}_{t} + \big[ (1 - L^{1i}_{t + 1} + \alpha_{t} \lambda_{t} \mathrm{e}^{- \rho}L^{2i}_{t + 1}) q^{i}_{t} + (- 1 + L^{1i}_{t + 1}) \overline{Q}^{i}_{t} + \mathrm{e}^{- \rho} L^{2i}_{t + 1} R_{t} + L_{t + 1}^{3i} \overline{Q}^{j}_{t} \nonumber\\
&\quad + L_{t + 1}^{5i} + (\alpha_{t} \lambda_{t} \mathrm{e}^{- \rho} L^{2i}_{t + 1} - L_{t + 1}^{3i}) q^{j}_{t} \big] \mathcal{I}_{t} - (\overline{Q}^{i}_{t} - q^{i}_{t}) \epsilon_{t} \Big\} \Big\vert \bm{s}_{t} \Big], \label{Vattwithexpect}
\end{align}
where $\alpha^{t} \coloneqq \alpha_{t} \mathrm{e}^{-\rho} + \beta_{t}$.

Define
\begin{align}
\theta^{i}_{t} &\coloneqq 1 - L^{1i}_{t + 1} + \alpha_{t} \lambda_{t} \mathrm{e}^{- \rho} L^{2i}_{t + 1}; \\
\delta^{i}_{t} &\coloneqq -1 + L^{1i}_{t + 1}; \\
\phi^{i}_{t} &\coloneqq \alpha_{t} \lambda_{t} \mathrm{e}^{- \rho} L^{2i}_{t + 1} - L_{t + 1}^{3i}.
\end{align}
Then, let
\begin{align}
a &\coloneqq - \gamma^{i} L_{t + 1}^{4i}; \\
b &\coloneqq - \gamma^{i} (\theta^{i}_{t} q^{i}_{t} + \delta^{i}_{t} \overline{Q}^{i}_{t} + \mathrm{e}^{- \rho} L^{2i}_{t + 1} R_{t} + L_{t + 1}^{3i} \overline{Q}^{j}_{t} + L_{t + 1}^{5i} + \phi^{i}_{t} q^{j}_{t}); \\
c &\coloneqq \gamma^{i} (\overline{Q}^{i}_{t} - q^{i}_{t})
\end{align}
for the paremeters in Lemma~\ref{Lemma_A.2}. By using the lemma and rearranging Eq.~\eqref{Vattwithexpect}, we have
\begin{align}
&V^{i}_{t} (\pi^{1\ast}_{t}, \pi^{2\ast}_{t}) \big[ \bm{s}_{t} \big] \nonumber\\
&= \sup_{q^{i}_{t} \in \mathbb{R}} - \exp \Big\{ - \gamma^{i} \Big[ - A^{i}_{t} (q^{i}_{t})^{2} + \big[ B^{i}_{t} \overline{Q}^{i}_{t} + C^{i}_{t} \overline{Q}^{j}_{t} + D^{i}_{t} R_{t} + F^{i}_{t} \mathcal{I}_{t - 1} + M^{i}_{t} + N^{i}_{t} q^{j}_{t} \big] q^{i}_{t} + W^{i}_{t} - P_{t} \overline{Q}^{i}_{t} \nonumber\\
&\quad + \left\{ G^{1i}_{t + 1} - \frac{1}{2} \gamma^{i} (\delta^{i}_{t})^{2} \pi^{i, 11}_{t} - \frac{1}{2} \gamma^{i} \pi^{i, 22}_{t} + \gamma^{i} \pi^{i, 12}_{t} \delta^{i}_{t} \right\} (\overline{Q}^{i}_{t})^{2} \nonumber\\
&\quad + \Big\{ G^{2i}_{t + 1} - \gamma^{i} \delta^{i}_{t} L_{t + 1}^{5i} \pi^{i, 11}_{t} + \gamma^{i} \pi^{i, 12}_{t} L_{t + 1}^{5i} + \left( \widetilde{\sigma}^{11}_{t} \pi^{i, 11}_{t} + \widetilde{\sigma}^{12}_{t} \pi^{i, 21}_{t} \right) a^{\mathcal{I}}_{t} \delta^{i}_{t} - \left( \widetilde{\sigma}^{11}_{t} \pi^{i, 12}_{t} + \widetilde{\sigma}^{12}_{t} \pi^{i, 22}_{t} \right) a^{\mathcal{I}}_{t} \Big\} \overline{Q}^{i}_{t}, \nonumber\\
&\quad + \left\{ (1 - \mathrm{e}^{- \rho}) + \mathrm{e}^{- \rho} H^{1i}_{t + 1} - \gamma^{i} \delta^{i}_{t} \mathrm{e}^{- \rho} L^{2i}_{t + 1} \pi^{i, 11}_{t} + \gamma^{i} \pi^{i, 12}_{t} \mathrm{e}^{- \rho} L^{2i}_{t + 1} \right\} \overline{Q}^{i}_{t} R_{t} \nonumber\\
&\quad + \left\{ \mathrm{e}^{- 2 \rho} H^{2i}_{t + 1} - \frac{1}{2} \gamma^{i} \mathrm{e}^{- 2 \rho} (L^{2i}_{t + 1})^{2} \pi^{i, 11}_{t} \right\} R_{t}^{2} \nonumber\\
&\quad + \left\{ \mathrm{e}^{- \rho} H^{3i}_{t + 1} - \gamma^{i} \mathrm{e}^{- \rho} L^{2i}_{t + 1} L_{t + 1}^{5i} \pi^{i, 11}_{t} + \left( \widetilde{\sigma}^{11}_{t} \pi^{i, 11}_{t} + \widetilde{\sigma}^{12}_{t} \pi^{i, 21}_{t} \right) a^{\mathcal{I}}_{t} \mathrm{e}^{- \rho} L^{2i}_{t + 1} \right\} R_{t} \nonumber\\
&\quad + \left\{ J^{1i}_{t + 1} - \gamma^{i} \delta^{i}_{t} L_{t + 1}^{3i} \pi^{i, 11}_{t} + \gamma^{i} \pi^{i, 12}_{t} L_{t + 1}^{3i} \right\} \overline{Q}^{i}_{t} \overline{Q}^{j}_{t} 
+ \left\{ \mathrm{e}^{- \rho} J^{2i}_{t + 1} - \gamma^{i} \mathrm{e}^{- \rho} L^{2i}_{t + 1} L_{T + 1}^{3i} \pi^{i, 11}_{t} \right\} R_{t} \overline{Q}^{j}_{t} \nonumber\\
&\quad + \left\{ J^{3i}_{t + 1} - \frac{1}{2} \gamma^{i} (L_{t + 1}^{3i})^{2} \pi^{i, 11}_{t} \right\} (\overline{Q}^{j}_{t})^{2} 
+ \left\{ J^{4i}_{t + 1} - \gamma^{i} L_{t + 1}^{3i} L_{t + 1}^{5i} \pi^{i, 11}_{t} + (\widetilde{\sigma}^{11}_{t} \pi^{i, 11}_{t} + \widetilde{\sigma}^{12}_{t} \pi^{i, 21}_{t}) a^{\mathcal{I}}_{t} L_{t + 1}^{3i} \right\} \overline{Q}^{j}_{t} \nonumber\\
&\quad + \left\{ - \left( \widetilde{\sigma}^{11}_{t} \pi^{i, 11}_{t} + \widetilde{\sigma}^{12}_{t} \pi^{i, 21}_{t} \right) b^{\mathcal{I}}_{t} \delta^{i}_{t} + \left( \widetilde{\sigma}^{11}_{t} \pi^{i, 12}_{t} + \widetilde{\sigma}^{12}_{t} \pi^{i, 22}_{t} \right) b^{\mathcal{I}}_{t} \right\} \overline{Q}^{i}_{t} \mathcal{I}_{t - 1} \nonumber\\
&\quad + \left\{ - \left( \widetilde{\sigma}^{11}_{t} \pi^{i, 11}_{t} + \widetilde{\sigma}^{12}_{t} \pi^{i, 21}_{t} \right) b^{\mathcal{I}}_{t} \mathrm{e}^{- \rho} L^{2i}_{t + 1} \right\} R_{t} \mathcal{I}_{t - 1} 
+ \left\{ - \left( \widetilde{\sigma}^{11}_{t} \pi^{i, 11}_{t} + \widetilde{\sigma}^{12}_{t} \pi^{i, 21}_{t} \right) b^{\mathcal{I}}_{t} L_{t + 1}^{3i} \right\} \overline{Q}^{j}_{t} \mathcal{I}_{t - 1} \nonumber\\
&\quad + \left\{ - \frac{1}{2 \gamma^{i}} \left\{ \left( \widetilde{\sigma}^{11}_{t} \right)^{2} \pi^{i, 11}_{t} + 2 \widetilde{\sigma}^{11}_{t} \widetilde{\sigma}^{12}_{t} \pi^{i, 12}_{t} + \widetilde{\sigma}^{12}_{t} \widetilde{\sigma}^{22}_{t} \pi^{i, 22}_{t} \right\} (b^{\mathcal{I}}_{t})^{2} - \frac{\widetilde{\sigma}^{11}_{t}}{2 \gamma^{i}} (b^{\mathcal{I}}_{t})^{2} \right\} \mathcal{I}_{t - 1}^{2} \nonumber\\
&\quad + \left\{ - \left( \widetilde{\sigma}^{11}_{t} \pi^{i, 11}_{t} + \widetilde{\sigma}^{12}_{t} \pi^{i, 21}_{t} \right) b^{\mathcal{I}}_{t} L_{t + 1}^{5i} + \frac{1}{\gamma^{i}} a^{\mathcal{I}}_{t} b^{\mathcal{I}}_{t} \Big\{ \left( \widetilde{\sigma}^{11}_{t} \right)^{2} \pi^{i, 11}_{t} + 2 \widetilde{\sigma}^{11}_{t} \widetilde{\sigma}^{12}_{t} \pi^{i, 12}_{t} + \widetilde{\sigma}^{12}_{t} \widetilde{\sigma}^{22}_{t} \pi^{i, 22}_{t} \Big\} + \frac{1}{\gamma^{i}} \tilde{\sigma}^{11}_{t} a^{\mathcal{I}}_{t} b^{\mathcal{I}}_{t} \right\} \mathcal{I}_{t - 1} \nonumber\\
&\quad + Z^{i}_{t + 1} + \left( \widetilde{\sigma}^{11}_{t} \pi^{i, 11}_{t} + \widetilde{\sigma}^{12}_{t} \pi^{i, 21}_{t} \right) a^{\mathcal{I}}_{t} L_{t + 1}^{5i} - \frac{1}{2 \gamma^{i}} (a^{\mathcal{I}}_{t})^{2} \Big\{ \left( \widetilde{\sigma}^{11}_{t} \right)^{2} \pi^{i, 11}_{t} + 2 \widetilde{\sigma}^{11}_{t} \widetilde{\sigma}^{12}_{t} \pi^{i, 12}_{t} + \widetilde{\sigma}^{12}_{t} \widetilde{\sigma}^{22}_{t} \pi^{i, 22}_{t} \Big\} \nonumber\\
&\quad - \frac{1}{2 \gamma^{i}} \widetilde{\sigma}^{11}_{t} (a^{\mathcal{I}}_{t})^{2} + x^{i}_{t} + X^{i}_{t} (q^{i}_{t})^{2} + \left[ Y^{1i}_{t} \overline{Q}^{i}_{t} + Y^{2i}_{t} R_{t} + Y^{3i}_{t} \overline{Q}^{j}_{t} + Y^{4i}_{t} \mathcal{I}_{t - 1} + Y_{t - 1}^{5i} \right] q^{j}_{t} \Big] \Big\},
\end{align}
where $\displaystyle{x^{i}_{t} \coloneqq -\frac{1}{\gamma^{i}} \log \frac{\lvert \left( \pmb{\Sigma}^{\ast} \right)^{-1} \rvert^{\frac{1}{2}}}{\lvert \pmb{\Sigma} \rvert^{\frac{1}{2}}}}$, and
\begin{align}
\begin{pmatrix} 
	\widetilde{\sigma}^{11}_{t} & \widetilde{\sigma}^{12}_{t} \\ 
	\widetilde{\sigma}^{21}_{t} & \widetilde{\sigma}^{22}_{t}
\end{pmatrix} 
\coloneqq 
\begin{pmatrix}
	\left( \sigma_{t + 1}^{\mathcal{I}} \right)^{2} & \rho^{\mathcal{I}, \epsilon} \sigma_{t + 1}^{\mathcal{I}} \sigma_{t + 1}^{\epsilon} \\
	\rho^{\mathcal{I}, \epsilon} \sigma_{t + 1}^{\mathcal{I}} \sigma^{\epsilon}_{t} & (\sigma^{\epsilon}_{t + 1})^{2}
\end{pmatrix}^{-1}
, \quad
\begin{pmatrix} 
	\pi^{i, 11}_{t} & \pi^{i, 12}_{t} \\
	\pi^{i, 21}_{t} & \pi^{i, 22}_{t} 
\end{pmatrix} 
\coloneqq 
\begin{pmatrix} 
	\widetilde{\sigma}^{11}_{t} + 2 \gamma^{i} L_{t + 1}^{4i} & \widetilde{\sigma}^{12}_{t} \\ 
	\widetilde{\sigma}^{21}_{t} & \widetilde{\sigma}^{22}_{t} 
\end{pmatrix}^{-1}, 
\end{align}
and
\begin{align}
A^{i}_{t} &\coloneqq \displaystyle{(1 - \alpha^{t}) \lambda_{t} - G^{1i}_{t + 1} + \alpha_{t} \lambda_{t} \mathrm{e}^{- \rho} H^{1i}_{t + 1} - \alpha^{2}_{t} \lambda^{2}_{t} \mathrm{e}^{- 2 \rho} H^{2i}_{t + 1} + \frac{1}{2} \gamma^{i} (\theta^{i}_{t})^{2} \pi^{i, 11}_{t} + \frac{1}{2} \gamma^{i} \pi^{i, 22}_{t} + \gamma^{i} \theta^{i}_{t} \pi^{i, 12}_{t}}; \nonumber\\
B^{i}_{t} &\coloneqq - \alpha^{t} \lambda_{t} - 2 G^{1i}_{t + 1} + \alpha_{t} \lambda_{t} \mathrm{e}^{- \rho} H^{1i}_{t + 1} - \gamma^{i} \theta^{i}_{t} \delta^{i}_{t} \pi^{i, 11}_{t} + \gamma^{i} \pi^{i, 22}_{t} + \gamma^{i} \theta^{i}_{t} \pi^{i, 12}_{t} - \gamma^{i} \delta^{i}_{t} \pi^{i, 12}_{t}; \nonumber\\
C^{i}_{t} &\coloneqq - J^{1i}_{t + 1} + \alpha_{t} \lambda_{t} \mathrm{e}^{- \rho} J^{2i}_{t + 1} - \gamma^{i} \theta^{i}_{t} L_{t + 1}^{3i} \pi^{i, 11}_{t} - \gamma^{i} L_{t + 1}^{3i} \pi^{i, 12}_{t}; \nonumber\\	
D^{i}_{t} &\coloneqq - (1 - \mathrm{e}^{- \rho}) - \mathrm{e}^{- \rho} H^{1i}_{t + 1} + 2 \alpha_{t} \lambda_{t} \mathrm{e}^{- 2\rho} H^{2i}_{t + 1} - \gamma^{i} \theta^{i}_{t} \mathrm{e}^{- \rho} L^{2i}_{t + 1} \pi^{i, 11}_{t} - \gamma^{i} \mathrm{e}^{- \rho} L^{2i}_{t + 1} \pi^{i, 12}_{t}; \nonumber\\
F^{i}_{t} &\coloneqq - \left( \widetilde{\sigma}^{11}_{t} \pi^{i, 11}_{t} + \widetilde{\sigma}^{12}_{t} \pi^{i, 21}_{t} \right) b^{\mathcal{I}}_{t} \theta^{i}_{t} - \left( \widetilde{\sigma}^{11}_{t} \pi^{i, 12}_{t} + \widetilde{\sigma}^{12}_{t} \pi^{i, 22}_{t} \right) b^{\mathcal{I}}_{t}; \nonumber\\
M^{i}_{t} &\coloneqq - G^{2i}_{t + 1} + \alpha_{t} \lambda_{t} \mathrm{e}^{- \rho} H^{3i}_{t + 1} - \gamma^{i} \theta^{i}_{t} L^{5i}_{t + 1} \pi^{i, 11}_{t} - \gamma^{i} L_{t + 1}^{5i} \pi^{i, 12}_{t} + a^{\mathcal{I}}_{t} \theta^{i}_{t} \big( \widetilde{\sigma}^{11}_{t} \pi^{i, 11}_{t} + \widetilde{\sigma}^{12}_{t} \pi^{i, 21}_{t} \big) \nonumber\\
&\qquad + a^{\mathcal{I}}_{t} \left( \widetilde{\sigma}^{11}_{t} \pi^{i, 12}_{t} + \widetilde{\sigma}^{12}_{t} \pi^{i, 22}_{t} \right); \nonumber\\
N^{i}_{t} &\coloneqq - (1 - \alpha^{t}) \lambda_{t} - \alpha_{t} \lambda_{t} \mathrm{e}^{- \rho} H^{1i}_{t + 1} + 2 \alpha^{2}_{t} \lambda^{2}_{t} \mathrm{e}^{- 2 \rho} H^{2i}_{t + 1} + J^{1i}_{t + 1} - \alpha_{t} \lambda_{t} \mathrm{e}^{- \rho} J^{2i}_{t + 1} - \gamma^{i} \theta^{i}_{t} \phi^{i}_{t} \pi^{i, 11}_{t} - \gamma^{i} \phi^{i}_{t} \pi^{i, 12}_{t},
\end{align}
and
\begin{align}
X^{i}_{t} &\coloneqq \alpha^{2}_{t} \lambda^{2}_{t} \mathrm{e}^{- \rho} H^{2i}_{t + 1} - \alpha_{t} \lambda_{t} \mathrm{e}^{- \rho} J^{2i}_{t + 1} + J^{3i}_{t + 1} - \frac{1}{2} \gamma^{i} (\phi^{i}_{t})^{2} \pi^{i, 11}_{t}; \nonumber\\
Y^{1i}_{t} &\coloneqq - \alpha^{t} \lambda_{t} + \alpha_{t} \lambda_{t} \mathrm{e}^{- \rho} H^{1i}_{t + 1} - J^{1i}_{t + 1} - \gamma^{i} \delta^{i}_{t} \phi^{i}_{t} \pi^{i, 11}_{t} + \gamma^{i} \phi^{i}_{t} \pi^{i, 12}_{t}; \nonumber\\
Y^{2i}_{t} &\coloneqq 2 \alpha_{t} \lambda_{t} \mathrm{e}^{- 2\rho} H^{2i}_{t + 1} - \mathrm{e}^{- \rho} J^{2i}_{t + 1} - \gamma^{i} \mathrm{e}^{- \rho} L^{2i}_{t + 1} \phi^{i}_{t} \pi^{i, 11}_{t}; \nonumber\\
Y^{3i}_{t} &\coloneqq \alpha_{t} \lambda_{t} \mathrm{e}^{- \rho} J^{2i}_{t + 1} - 2 J^{3i}_{t + 1} - \gamma^{i} L_{t + 1}^{3i} \phi^{i}_{t} \pi^{i, 11}_{t}; \nonumber\\
Y^{4i}_{t} &\coloneqq - \left( \widetilde{\sigma}^{11}_{t} \pi^{i, 11}_{t} + \widetilde{\sigma}^{12}_{t} \pi^{i, 21}_{t} \right) b^{\mathcal{I}}_{t} \phi^{i}_{t}; \nonumber\\
Y^{5i}_{t} &\coloneqq \alpha_{t} \lambda_{t} \mathrm{e}^{- \rho} H^{3i}_{t + 1} - J^{4i}_{t + 1} - \gamma^{i} L_{t + 1}^{5i} \phi^{i}_{t} \pi^{i, 11}_{t} + \left( \widetilde{\sigma}^{11}_{t} \pi^{i, 11}_{t} + \widetilde{\sigma}^{12}_{t} \pi^{i, 21}_{t} \right) a^{\mathcal{I}}_{t} \phi^{i}_{t}.
\end{align}
Then, the best response of large trader $i \in \{1, 2\}$ to the other large trader $j \neq i$ at time $t$, denoted by $BR^{i} (q^{j}_{t})$, becomes  
\begin{align}
BR^{i} (q^{j}_{t}) = \frac{1}{2 A^{i}_{t}} \left( B^{i}_{t} \overline{Q}^{i}_{t} + C^{i}_{t} \overline{Q}^{j}_{t} + D^{i}_{t} R_{t} + F^{i}_{t} \mathcal{I}_{t - 1} + M^{i}_{t} + N^{i}_{t} q^{j}_{t} \right).
\end{align}
Thus, at the Markov perfect equilibrium, we have
\begin{align}
q^{i\ast}_{t} &= \frac{1}{2 A^{i}_{t}} \left( B^{i}_{t} \overline{Q}^{i}_{t} + C^{i}_{t} \overline{Q}^{j}_{t} + D^{i}_{t} R_{t} + F^{i}_{t} \mathcal{I}_{t - 1} + M^{i}_{t} + N^{i}_{t} q^{j\ast}_{t} \right); \\
q^{j\ast}_{t} &= \frac{1}{2 A^{j}_{t}} \left( B^{j}_{t} \overline{Q}^{j}_{t} + C^{j}_{t} \overline{Q}^{i}_{t} + D^{j}_{t} R_{t} + F^{j}_{t} \mathcal{I}_{t - 1} + M^{j}_{t} + N^{j}_{t} q^{i\ast}_{t} \right).
\end{align}
Solving the above simultaneous equations results in
\begin{align}
q^{i\ast}_{t} 
&= B^{i\ast}_{t} \overline{Q}^{i}_{t} + C^{i\ast}_{t} \overline{Q}^{j}_{t} + D^{i\ast}_{t} R_{t} + F^{i\ast}_{t} \mathcal{I}_{t - 1} + M^{i\ast}_{t} \nonumber\\
&(\eqqcolon a^{i}_{t} + b^{i}_{t} \overline{Q}^{i}_{t} + c^{i}_{t} \overline{Q}^{j}_{t} + d^{i}_{t} R_{t} + e^{i}_{t} \mathcal{I}_{t - 1}),
\end{align}
where
\begin{align}
\zeta^{i}_{t} &\coloneqq 2 A^{i}_{t} - \frac{N^{i}_{t} N^{j}_{t}}{2 A^{j}_{t}}; &
B^{i\ast}_{t} &\coloneqq \frac{1}{\zeta^{i}_{t}} \left( B^{i}_{t} + \frac{N^{i}_{t} C^{j}_{t}}{2 A^{j}_{t}} \right); \nonumber\\
C^{i\ast}_{t} &\coloneqq \frac{1}{\zeta^{i}_{t}} \left( C^{i}_{t} + \frac{N^{i}_{t} B^{j}_{t}}{2 A^{j}_{t}} \right); &
D^{i\ast}_{t} &\coloneqq \frac{1}{\zeta^{i}_{t}} \left( D^{i}_{t} + \frac{N^{i}_{t} D^{j}_{t}}{2 A^{j}_{t}} \right); \nonumber\\
F^{i\ast}_{t} &\coloneqq \frac{1}{\zeta^{i}_{t}} \left( F^{i}_{t} + \frac{N^{i}_{t} F^{j}_{t}}{2 A^{j}_{t}} \right); &
M^{i\ast}_{t} &\coloneqq \frac{1}{\zeta^{i}_{t}} \left( M^{i}_{t} + \frac{N^{i}_{t} M^{j}_{t}}{2 A^{j}_{t}} \right),
\end{align}
for each $i \in \{1, 2\}$. $(q^{1\ast}_{t}, q^{2\ast}_{t})$ is the pair of execution volume at the Markov perfect equilibrium $(\pi^{1 \ast}, \pi^{2 \ast}) \in \Pi_{M}^{1} \times \Pi_{M}^{2}$ at time $t \in \{T - 2, \ldots, 1\}$. The value function for large trader $i \in \{1, 2\}$ at the Markov Perfect equilibrium then becomes 
\begin{align}
&V^{i}_{t} (\pi^{1\ast}_{t}, \pi^{2\ast}_{t}) \big[ W^{1}_{t}, W^{2}_{t}, P_{t}, \overline{Q}^{1}_{t}, \overline{Q}^{2}_{t}, R_{t}, \mathcal{I}_{t - 1} \big] \nonumber\\
&= - \exp \Big\{ - \gamma^{i} \Big[ W^{i}_{t} - P_{t} \overline{Q}^{i}_{t} + G^{1i}_{t} (\overline{Q}^{i}_{t})^{2} + G^{2i}_{t} \overline{Q}^{i}_{t} + H^{1i}_{t} \overline{Q}^{i}_{t} R_{t} \nonumber\\
&\quad + H^{2i}_{t} R_{t}^{2} + H^{3i}_{t} R_{t} + J^{1i}_{t} \overline{Q}^{i}_{t} \overline{Q}^{j}_{t} + J^{2i}_{t}  \overline{Q}^{j}_{t} R_{t} + J^{3i}_{t} (\overline{Q}^{j}_{t})^{2} + J^{4i}_{t} \overline{Q}^{j}_{t} \nonumber\\
&\quad + L^{1i}_{t} \overline{Q}^{i}_{t} \mathcal{I}_{t - 1} + L^{2i}_{t} R_{t} \mathcal{I}_{t - 1} + L^{3i}_{t} \overline{Q}^{j}_{t} \mathcal{I}_{t - 1} 
+ L^{4i}_{t} \mathcal{I}_{t - 1}^{2} + L^{5i}_{t} \mathcal{I}_{t - 1} + Z^{i}_{t} \Big] \Big\}, 
\end{align}
where 
\begin{align}
B^{i\ast\ast}_{t} &\coloneqq B^{i}_{t} + N^{i}_{t} C^{j\ast}_{t}; & C^{i\ast\ast}_{t} &\coloneqq C^{i}_{t} + N^{i}_{t} B^{j\ast}_{t}; \nonumber\\
D^{i\ast\ast}_{t} &\coloneqq D^{i}_{t} + N^{i}_{t} D^{j\ast}_{t}; & F^{i\ast\ast}_{t} &\coloneqq F^{i}_{t} + N^{i}_{t} F^{j\ast}_{t}; \nonumber\\
M^{i\ast\ast}_{t} &\coloneqq M^{i}_{t} + N^{i}_{t} M^{j\ast}_{t}, & i, j &= 1, 2, \quad i \neq j,
\end{align}
and
\begin{align} 
G^{1i}_{t} &\coloneqq G^{1i}_{t + 1} - \frac{1}{2} \gamma^{i} \delta^{i}_{t} \pi^{i, 11}_{t} - \frac{1}{2} \gamma^{i} \pi^{i, 22}_{t} + \gamma^{i} \delta^{i}_{t} \pi^{i, 12}_{t} + X^{i}_{t} (C^{j\ast}_{t})^{2} + Y^{1i}_{t} C^{j\ast}_{t} + \frac{(B^{i\ast\ast}_{t})^{2}}{4 A^{i}_{t}}; \nonumber\\ 
G^{2i}_{t} &\coloneqq G^{2i}_{t + 1} - \gamma^{i} \delta^{i}_{t} L_{t + 1}^{5i} \pi^{i, 11}_{t} + \gamma^{i} L_{t + 1}^{5i} \pi^{i, 12}_{t} + \left( \widetilde{\sigma}^{11}_{t} \pi^{i, 11}_{t} + \widetilde{\sigma}^{12}_{t} \pi^{i, 21}_{t} \right) a^{\mathcal{I}}_{t} \delta^{i}_{t} - \left( \widetilde{\sigma}^{11}_{t} \pi^{i, 12}_{t} + \widetilde{\sigma}^{12}_{t} \pi^{i, 22}_{t} \right) a^{\mathcal{I}}_{t} \nonumber\\ 
&\quad + 2 X^{i}_{t} C^{j\ast}_{t} M^{j\ast}_{t} + Y^{5i}_{t} C^{j\ast}_{t} + Y^{1i}_{t} M^{j\ast}_{t} + \frac{B^{i\ast\ast}_{t} M^{i\ast\ast}_{t}}{2 A^{i}_{t}}; \nonumber\\
H^{1i}_{t} &\coloneqq (1 - \mathrm{e}^{- \rho}) + \mathrm{e}^{- \rho} H^{1i}_{t + 1} - \gamma^{i} \delta^{i}_{t} \mathrm{e}^{- \rho} L^{2i}_{t + 1} \pi^{i, 11}_{t} + \gamma^{i} \mathrm{e}^{- \rho} L^{2i}_{t + 1} \pi^{i, 12}_{t} + 2 X^{i}_{t} C^{j\ast}_{t} D^{j\ast}_{t} \nonumber\\
&\quad + Y^{2i}_{t} C^{j\ast}_{t} + Y^{1i}_{t} D^{j\ast}_{t} + \frac{B^{i\ast\ast}_{t} D^{i\ast\ast}_{t}}{2 A^{i}_{t}}; \nonumber\\
H^{2i}_{t} &\coloneqq \mathrm{e}^{- 2 \rho} H^{2i}_{t + 1} - \frac{1}{2} \gamma^{i} \mathrm{e}^{- 2 \rho} (L^{2i}_{t + 1})^{2} \pi^{i, 11}_{t} + X^{i}_{t} (D^{j\ast}_{t})^{2} + Y^{2i}_{t} D^{j\ast}_{t} + \frac{(D^{i\ast\ast}_{t})^{2}}{4 A^{i}_{t}}; \nonumber\\
H^{3i}_{t} &\coloneqq \mathrm{e}^{- \rho} H^{3i}_{t + 1} - \gamma^{i} \mathrm{e}^{- \rho} L^{2i}_{t + 1} L_{t + 1}^{5i} \pi^{i, 11}_{t} + \left( \widetilde{\sigma}^{11}_{t} \pi^{i, 11}_{t} + \widetilde{\sigma}^{12}_{t} \pi^{i, 21}_{t} \right) a^{\mathcal{I}}_{t} \mathrm{e}^{- \rho} L^{2i}_{t + 1} \nonumber\\
&\quad + 2 X^{i}_{t} D^{j\ast}_{t} M^{j\ast}_{t} + Y^{5i}_{t} D^{j\ast}_{t} + Y^{2i}_{t} M^{j\ast}_{t} + \frac{D^{i\ast\ast}_{t} M^{i\ast\ast}_{t}}{2 A^{i}_{t}}; \nonumber\\
J^{1i}_{t} &\coloneqq J^{1i}_{t + 1} - \gamma^{i} \delta^{i}_{t} L_{t + 1}^{3i} \pi^{i, 11}_{t} + \gamma^{i} L_{t + 1}^{3i} \pi^{i, 12}_{t} + 2 X^{i}_{t} B^{j\ast}_{t} C^{j\ast}_{t} + Y^{1i}_{t} B^{j\ast}_{t} + Y^{3i}_{t} C^{j\ast}_{t} + \frac{B^{i\ast\ast}_{t} C^{i\ast\ast}_{t}}{2 A^{i}_{t}}; \nonumber\\
J^{2i}_{t} &\coloneqq \mathrm{e}^{- \rho} J^{2i}_{t + 1} - \gamma^{i} \mathrm{e}^{- \rho} L^{2i}_{t + 1} L_{T + 1}^{3i} \pi^{i, 11}_{t} + 2 X^{i}_{t} B^{j\ast}_{t} D^{j\ast}_{t} + Y^{2i}_{t} B^{j\ast}_{t} + Y^{3i}_{t} D^{j\ast}_{t} + \frac{C^{i\ast\ast}_{t} D^{i\ast\ast}_{t}}{2 A^{i}_{t}}; \nonumber\\ 
J^{3i}_{t} &\coloneqq J^{3i}_{t + 1} - \frac{1}{2} \gamma^{i} (L_{t + 1}^{3i})^{2} \pi^{i, 11}_{t} + X^{i}_{t} (B^{j\ast}_{t})^{2} + Y^{3i}_{t} B^{j\ast}_{t} + \frac{(C^{i\ast\ast}_{t})^{2}}{4 A^{i}_{t}}; \nonumber\\
J^{4i}_{t} &\coloneqq J^{4i}_{t + 1} - \gamma^{i} L_{t + 1}^{3i} L_{t + 1}^{5i} \pi^{i, 11}_{t} + \left( \widetilde{\sigma}^{11}_{t} \pi^{i, 11}_{t} + \widetilde{\sigma}^{12}_{t} \pi^{i, 21}_{t} \right) a^{\mathcal{I}}_{t} L_{t + 1}^{3i} \nonumber\\
&\quad + 2 X^{i}_{t} B^{j\ast}_{t} M^{j\ast}_{t} + Y^{5i}_{t} B^{j\ast}_{t} + Y^{3i}_{t} M^{j\ast}_{t} + \frac{C^{i\ast\ast}_{t} M^{i\ast\ast}_{t}}{2 A^{i}_{t}}; \nonumber\\
L^{1i}_{t} &\coloneqq - \left( \widetilde{\sigma}^{11}_{t} \pi^{i, 11}_{t} + \widetilde{\sigma}^{12}_{t} \pi^{i, 21}_{t} \right) b^{\mathcal{I}}_{t} \delta^{i}_{t} + \left( \widetilde{\sigma}^{11}_{t} \pi^{i, 12}_{t} + \widetilde{\sigma}^{12}_{t} \pi^{i, 22}_{t} \right) b^{\mathcal{I}}_{t} + 2 X^{i}_{t} C^{j\ast}_{t} F^{j\ast}_{t} + Y^{4i}_{t} C^{j\ast}_{t} + Y^{1i}_{t} F^{j\ast}_{t} \nonumber\\
&\quad + \frac{B^{i\ast\ast}_{t} F^{i\ast\ast}_{t}}{2 A^{i}_{t}}; \nonumber\\ 
L^{2i}_{t} &\coloneqq - \left( \widetilde{\sigma}^{11}_{t} \pi^{i, 11}_{t} + \widetilde{\sigma}^{12}_{t} \pi^{i, 21}_{t} \right) b^{\mathcal{I}}_{t} \mathrm{e}^{- \rho} L^{2i}_{t + 1} + 2 X^{i}_{t} D^{j\ast}_{t} F^{j\ast}_{t} + Y^{4i}_{t} D^{j\ast}_{t} + Y^{2i}_{t} F^{j\ast}_{t} + \frac{D^{i\ast\ast}_{t} F^{i\ast\ast}_{t}}{2 A^{i}_{t}}; \nonumber\\ 
L^{3i}_{t} &\coloneqq - \left( \widetilde{\sigma}^{11}_{t} \pi^{i, 11}_{t} + \widetilde{\sigma}^{12}_{t} \pi^{i, 21}_{t} \right) b^{\mathcal{I}}_{t} L^{3i}_{t + 1} + 2 X^{i}_{t} B^{j\ast}_{t} F^{j\ast}_{t} + Y^{4i}_{t} B^{j\ast}_{t} + Y^{3i}_{t} F^{j\ast}_{t} + \frac{C^{i\ast\ast}_{t} F^{i\ast\ast}_{t}}{2 A^{i}_{t}}; \nonumber\\ 
L^{4i}_{t} &\coloneqq - \frac{1}{2 \gamma^{i}} \left\{ \left( \widetilde{\sigma}^{11}_{t} \right)^{2} \pi^{i, 11}_{t} + 2 \widetilde{\sigma}^{11}_{t} \widetilde{\sigma}^{12}_{t} \pi^{i, 12}_{t} + \widetilde{\sigma}^{12}_{t} \widetilde{\sigma}^{22}_{t} \pi^{i, 22}_{t} \right\} (b^{\mathcal{I}}_{t})^{2} - \frac{\widetilde{\sigma}^{11}_{t}}{2 \gamma^{i}} (b^{\mathcal{I}}_{t})^{2} + X^{i}_{t} (F^{j\ast}_{t})^{2} + Y^{4i}_{t} F^{j\ast}_{t} + \frac{(F^{i\ast\ast}_{t})^{2}}{4 A^{i}_{t}}; \nonumber\\ 
L^{5i}_{t} &\coloneqq - \left( \widetilde{\sigma}^{11}_{t} \pi^{i, 11}_{t} + \widetilde{\sigma}^{12}_{t} \pi^{i, 21}_{t} \right) b^{\mathcal{I}}_{t} L^{5i}_{t + 1} + \frac{1}{\gamma^{i}} a^{\mathcal{I}}_{t} b^{\mathcal{I}}_{t} \Big\{ \left( \widetilde{\sigma}^{11}_{t} \right)^{2} \pi^{i, 11}_{t} + 2 \widetilde{\sigma}^{11}_{t} \widetilde{\sigma}^{12}_{t} \pi^{i, 12}_{t} + \widetilde{\sigma}^{12}_{t} \widetilde{\sigma}^{22}_{t} \pi^{i, 22}_{t} \Big\} + \frac{1}{\gamma^{i}} \widetilde{\sigma}^{11}_{t} a^{\mathcal{I}}_{t} b^{\mathcal{I}}_{t} \nonumber\\
&\quad + 2 X^{i}_{t} F^{j\ast}_{t} M^{j\ast}_{t} + Y^{5i}_{t} F^{j\ast}_{t} + Y^{4i}_{t} M^{j\ast}_{t} + \frac{F^{i\ast\ast}_{t} M^{i\ast\ast}_{t}}{2 A^{i}_{t}}; \nonumber\\ 
Z^{i}_{t} &\coloneqq Z^{i}_{t + 1} + \left( \widetilde{\sigma}^{11}_{t} \pi^{i, 11}_{t} + \widetilde{\sigma}^{12}_{t} \pi^{i, 21}_{t} \right) a^{\mathcal{I}}_{t} L_{t + 1}^{5i} - \frac{1}{2 \gamma^{i}} \left( a^{\mathcal{I}}_{t} \right)^{2} \Big\{ \left( \widetilde{\sigma}^{11}_{t} \right)^{2} \pi^{i, 11}_{t} + 2 \widetilde{\sigma}^{11}_{t} \widetilde{\sigma}^{12}_{t} \pi^{i, 12}_{t} + \widetilde{\sigma}^{12}_{t} \widetilde{\sigma}^{22}_{t} \pi^{i, 22}_{t} \Big\} \nonumber\\
&\quad - \frac{1}{2 \gamma^{i}} \widetilde{\sigma}^{11}_{t} \left( a^{\mathcal{I}}_{t} \right)^{2} + x^{i}_{t} + X^{i}_{t} \big( M^{j\ast}_{t} \big)^{2} + Y^{5i}_{t} M^{j\ast}_{t} + \frac{(M^{i\ast\ast}_{t})^{2}}{4 A^{i}_{t}}.
\end{align}
\qed

\end{document}